\pdfoutput=1
\documentclass[journal, twocolumn, 11pt]{IEEEtran}
\usepackage{pgf, tikz}
\usepackage{calc}
\usetikzlibrary{calc}
\usetikzlibrary{arrows}
\usepackage{comment}
\definecolor{mycolor1}{rgb}{0.10000,0.60000,0.30000}%

\usepackage{amssymb}
\usepackage{color}
\usepackage{bbm}
\usepackage[mode=buildnew]{standalone}
\usepackage{research5IEEE}

\newfont{\bbb}{msbm10 scaled 500}

\newfont{\bb}{msbm10 scaled 1100}

\newcommand{\xv}{{\bf x}}

\newcommand{\onev}{{\bf 1}}

\newcommand{\Bc}{{\cal B}}
\newcommand{\CDc}{{\cal{CD}}}
\newcommand{\Cc}{{\cal C}}
\newcommand{\Dc}{{\cal D}}
\newcommand{\Ec}{{\cal E}}

\newcommand{\Ic}{{\cal I}}

\newcommand{\Pc}{{\cal P}}

\newcommand{\Sc}{{\cal S}}
\newcommand{\Tc}{{\cal T}}

\newcommand{\Wc}{{\cal W}}

\newcommand{\Xc}{{\cal X}}
\newcommand{\Yc}{{\cal Y}}
\newcommand{\Zc}{{\cal Z}}

\renewcommand{\xv}{\pmb{x}}

\usepackage{color}
\usepackage{cite}
\usepackage{graphics}
\usepackage{tikz}
\usepackage{pgfplots}
\usepackage{stfloats}
\usepackage{setspace}

\usepackage{float}
\usepackage{placeins}
\usetikzlibrary{shapes,  arrows,  decorations.markings,  arrows.meta}
\usetikzlibrary{patterns} 
\allowdisplaybreaks
\usepackage{epsfig}

\usepackage{rotating}
\usepackage{amsmath}
\usepackage{amsfonts}
\usepackage{url}

\usepackage[caption=false, font=footnotesize]{subfig}
\usepackage{epstopdf}
\usepackage{amsthm}
\usepackage{mathtools}

\usepackage{algorithm}
\usepackage[noend]{algpseudocode}

\usepackage{array}

\usepackage{cite}

\newtheorem{theorem}{Theorem}

\newtheorem{definition}{Definition}
\newtheorem{lemma}{Lemma}
\newtheorem{proposition}{Proposition}
\newtheorem{corollary}{Corollary}
\newtheorem{remark}{Remark}
\newcommand{\qt}{\textnormal{q}}

\renewcommand{\C}{\mathsf{C}}
\newcommand{\D}{\mathsf{D}}
\newcommand{\B}{\mathsf{B}}

\newcommand{\R}{\mathsf{R}}
\newcommand\redsout{\bgroup\markoverwith{\textcolor{red}{\rule[0.5ex]{2pt}{0.4pt}}}\ULon}

\begin{document}
	
	\title{An Information-Theoretic Approach to Joint Sensing and Communication}

	\author{Mehrasa Ahmadipour, Mari Kobayashi, Mich\`ele Wigger, Giuseppe Caire\thanks{Manuscript received July 28, 2021; revised May 4, 2022; accepted May 6, 2022. Part of this material was presented at IEEE Int. Symp. Inf. Theory (ISIT) 2018 \cite{kobayashi2018joint} and  at {IEEE} Inf. Theory Workshop {(ITW)} 2020  \cite{ahmadipour2021joint}.}	
		\thanks{ M. Ahmadipour and  M. Wigger are with LTCI, Telecom Paris, IP Paris, F-91120 Palaiseau, France,   \{mehrasa.ahmadipour,michele.wigger\}@telecom-paris.fr. }
		\thanks{{Mari Kobayashi is with Apple Technology Engineering B.V. Co. KG, 85579 Neubiberg, Germany.
				This work was done while she was at Technical University of Munich (email: kobamari@gmail.com). 
		}}
		\thanks{{ G. Caire is with 
				Technical University of Berlin, Germany,
				caire@tu-berlin.de.	}}
		\thanks{The works of M. Ahmadipour and M. Wigger were supported by the European Research Council (ERC) under the European Union’s Horizon 2020 programme, grant agreement number 715111. The works of M. Kobayashi and G. Caire were supported by the DFG, Grant agreement numbers KR 3517/11-1 and CA 1340/11-1, respectively. }
	}
	
	\maketitle

	\begin{abstract} 
		A communication setup is considered where a single transmitter wishes to convey  messages to one or two receivers and simultaneously estimates the states of the receivers \textcolor{black}{through the backscattered signals of the emitted waveform}. The scenario at hand is motivated by joint radar and communication, which aims to co-design radar sensing and communication over a shared  spectrum and hardware. 
	\textcolor{black}{
		In this paper, we model the  communication channel as a simple memoryless   channel with independent and identically distributed (i.i.d.) time-varying state sequences and we model the backscattered signals by (strictly causal) generalized feedback.  
		 For  single-receiver systems of this form, we fully characterize the capacity-distortion tradeoff,  defined as the largest  rate at which a message can reliably be conveyed  to the receiver while simultaneously allowing the transmitter to sense the state sequence with a given allowed distortion. Our results show a tradeoff between the achievable rates and distortions, and that this tradeoff only stems from a common choice of the input distribution (the waveform) but not from other properties of the utilized codes. To better illustrate the capacity-distortion tradeoff, we  propose a numerical method to compute the optimal inputs (waveforms)  that achieve the desired tradeoff. 
		For two-receiver systems with two states, we characterize the capacity-distortion tradeoff region of {physically degraded broadcast channels (BC)}  as a rather straightforward extension of  the single receiver case. Here, a tradeoff not only arises between sensing and communication performances but also between the various rates and the distortions of the different states. Similarly to the single-receiver case, the optimal co-design scheme exploits the generalized feedback  signals only for sensing but not for improving communication performance. This is different for general two-receiver BCs, where optimal co-design schemes exploit generalized feedback also to improve capacity. However, as we show,  also for BCs the optimal sensing performance only depends on the chosen input distribution (waveform) but not on the code construction used to accomplish the communication task. For  general BCs,  we  provide inner and outer bounds on the capacity-distortion region,  as well as a sufficient condition when this capacity-distortion region is equal to the product of the capacity region and the set of achievable distortions, in which case no tradeoff between sensing and communication occurs.  A number of illustrative examples demonstrate that the optimal co-design schemes outperform conventional schemes that split the resources between  sensing and communication, both for single-receiver and BC systems. }
		
	\end{abstract}
	\IEEEpeerreviewmaketitle
	

\section{Introduction}
Future generation wireless networks are expected to support several autonomous and intelligent applications that strongly rely on accurate sensing and localization techniques \cite{6G,bourdoux20206g}. {\color{black}An example are  intelligent transportation systems where  vehicles interact in a cooperative radar sensor network with the goal  to provide unique safety features and intelligent traffic routing.
The key enabler of such applications is the ability to sense the dynamically changing environment continuously, hereafter called the \emph{state}, and to react accordingly by {exchanging information}.  The standard assumption of such a  joint radar sensing and communication system is a transmitter equipped with a co-located radar receiver that wishes to convey a message to a (already detected) receiver and simultaneously estimate the state parameters of that receiver.

	A common but naive approach to address   sensing and communication is to separate the two tasks in independent systems and accordingly split the available resources such as bandwidth and power between the two systems. In our information-theoretic model that we present shortly, such a system corresponds to time-sharing between communication and sensing; we shall call this   {\it basic time-sharing (TS)}. 
	The high cost of spectrum and hardware however  encourages integrating the sensing and communications tasks via a single waveform and a single hardware platform (see, e.g., \cite{zheng2019radar,liu2020joint} and references therein). First attempts towards such integrated systems  use a standard communication system and exploit the backscattered signal from this waveform for sensing purposes, where the employed transmit  waveform is  either  optimized for sensing or for communication; we shall call these \emph{the sensing and communication modes of improved TS.}

 The scenario at hand has been extensively studied in the literature (see e.g. \cite{sturm2011waveform,gaudio2019effectiveness} and references therein). 
In particular, several {\it{ joint sensing and communication}} schemes, or co-design schemes, have been proposed to optimize performance metrics capturing some tension between two performances \cite{bliss2014cooperative,chiriyath2016inner,paul2017survey,kumari2018ieee,kumari2017performance}. 
Despite of these works providing system guidelines or proposing waveforms suitable to some specific scenarios, none has addressed the fundamental performance limits above which a joint sensing and communication system cannot operate irrespectively of computational complexities, choices of state parameters, or further assumptions. This observation inspires us to study the fundamental limit of joint sensing and communication from an information-theoretic perspective. 
	
	Our  work is the first information-theoretic work on joint sensing and communication. We  emphasize the difference  to the information-theoretic works in \cite{kim2008state,zhang2011joint,choudhuri2013causal, Shraga,Sibi} where sensing (state-estimation) is performed at the receiver and not at the transmitter, which models different real-world applications. In  \cite{kim2008state,choudhuri2013causal, Shraga,Sibi}, the transmitter even knows the state a priori. 

%

In this paper, we build on a simple single-transmitter communication model with a discrete  memoryless channel and independent and identically distributed (i.i.d.) state-sequences. The transmitter observes  strictly causal {\it generalized feedback signals}, used for state sensing, while each receiver is assumed to perfectly know its corresponding channel state. The generalized feedback model captures two underlying assumptions used in radar signal processing. 
On the one hand, it captures the inherently passive nature of the backscattered signal observed at the transmitter, which cannot be controlled but is determined by its surrounding environment. 
On the other hand, it models the fact  that the  backscattered signal depends on the waveform employed by the transmitter. It is thus clear, that the employed waveform affects both the communication and sensing performances of the system and should be designed in a synergistic manner. 
%
Our goal is to characterize the fundamental  tradeoff between the communication and sensing performance of such systems and the improvements an optimally designed scheme achieves over the separation scheme (i.e., the described basic TS) and over integrated systems that either prioritize sensing or communication (i.e., above described improved TS). To this purpose, we consider the {\it capacity-distortion tradeoff} as a performance measure since it suitably balances between two ultimate objectives:  maximizing communication rate and minimizing state estimation error or {\it{distortion}}. 
The presented model was introduced in our conference publications  \cite{kobayashi2018joint,ahmadipour2021joint} and was also extended to the two-user multiple-access channel in \cite{kobayashi2019joint, ahmadipour2022jointMAC}.
	
	In this work we consider the single-transmitter single-receiver point-to-point (P2P) channel and the single-transmitter two-receiver BC. 
For the P2P channel we exactly  characterize the capacity-distortion-cost tradeoff, which allows  to  quantify the merit of an optimal co-design scheme over the described basic and improved TS schemes. Not surprisingly, our results show  that  without loss in optimality the communication scheme can ignore the generalized feedback signals, which are only used for state sensing, and  communication and sensing performances only depend on each  other through the choice of the common waveform. Our results further show that in most situations a tradeoff between the simultaneously achievable sensing and communication performances arises. Based on  our results we further identify  ``matched" situations where the same waveform simultaneously achieves capacity and minimum distortion. A Blahut-Arimoto type algorithm is presented that evaluates the capacity-distortion-cost tradeoff numerically.

While feedback does not increase  capacity of memoryless P2P channels, it can significantly increase capacity of memoryless BCs \cite{shayevitz2012capacity,venkataramanan2013achievable,gastpar2014coding} because it enables the transmitter to send some common information that is useful to both receivers at the same time (see e.g., \cite[Section 17]{el2011network}).  ln our joint sensing and communication-over-BC setup, the generalized feedback  thus improves both   sensing and  communication performances. Nevertheless, like in the P2P setup, the two performances only depend on each other through the common choice of the waveform. In other words, we show that the optimal state-sensing is independent of the employed BC-feedback-code and only depends on the chosen waveform but not on other details of the code construction. This allows to base  joint coding and sensing systems on known BC-feedback code constructions such as \cite{shayevitz2012capacity,venkataramanan2013achievable,gastpar2014coding}. Based on the scheme in \cite{shayevitz2012capacity}, we provide a general inner bound 
on the capacity-distortion region for general memoryless BCs with generalized feedback. We also provide a general outer bound by extending a known converse  technique that  reveals the outputs at one of the receivers to the other receiver. Inner and outer bounds coincide only in special cases. 
Completely characterizing the capacity-distortion tradeoff region for a general memoryless state-dependent BC seems extremely challenging since even the capacity region (without sensing) is unknown both in the case without and with feedback (see e.g., \cite{Marton-noFB,shayevitz2012capacity,venkataramanan2013achievable,gastpar2014coding, Gohari2020OuterBC}). 
Instead, we characterize the capacity-distortion  region for the special case of physically degraded BCs. 
Analogously to the single-user case, feedback  does not enlarge the capacity of physically degraded BCs and is useful only for sensing but not for communication. Through various numerical examples we illustrate the merit of optimal co-design schemes the basic and improved TS for physically degraded and general BCs.}
\subsection{Contributions}
The paper provides the following technical contributions: 
\begin{enumerate}
	\item It characterizes the capacity-distortion-cost tradeoff of state-dependent memoryless channels in Theorem \ref{th:tradeoff} and states the optimal estimator (a deterministic symbol-by-symbol  estimator) in Lemma \ref{1user:lemma:Shat}.   A modified Blahut-Arimoto algorithm \cite{arimoto1972algorithm,blahut1972computation} is proposed to calculate the tradeoff region numerically. To this end, the optimality of an alternating optimization approach is proved in Theorem \ref{theorem:modifiedBA}. 
	\item As a rather straightforward extension of Theorem~\ref{th:tradeoff}, we characterize the capacity-distortion tradeoff region of physically degraded state-dependent memoryless broadcast channels in Theorem \ref{Th:physically_BC}. 
	\item For general state-dependent BCs, we provide an outer bound on the capacity-distortion region in Theorem~\ref{outer1} and an inner bound in Proposition \ref{prp:inner}. The inner bound is based on \cite{shayevitz2012capacity} and can be achieved using a block-Markov strategy that combines  Marton coding with a lossy version of  Gray-Wyner coding with side-information.  
	\item Corollary \ref{cor1:notradeoff} (for single-user channels) and Proposition \ref{prp1:BC:notradeoff} (for broadcast channels) identify sufficient conditions for channels where no capacity-distortion tradeoff arises. 
	\item Many illustrative examples are provided to demonstrate the benefits of the optimal co-design scheme compared to the aforementioned baseline schemes. These include a binary channel with a multiplicative Bernoulli state in Corollary \ref{cor:ex2}, a real Gaussian channel, a binary BC with multiplicative Bernoulli states in Corollaries \ref{corollary:binaryBC1} and \ref{corollary:binaryBC2}, as well as the state-dependent Dueck BC in Corollaries \ref{corollary:DueckOuter} and \ref{corollary:DueckInner}. 
\end{enumerate}
\subsection{Organization}
The rest of this paper is organized as follows. The following Section \ref{section:singlRx} formulates the joint sensing and communication problem in a single-receiver channel and provides the corresponding capacity-distortion-cost tradeoff. Section \ref{section:physBC} extends the obtained results to two-user broadcast channels. Finally, Section \ref{section:conclusion} concludes the paper.  
%
\subsection{Notation}
We use calligraphic letters to denote sets,  e.g.,  $\Xc$. The sets of real and nonnegative real numbers, however, are denoted by $\mathbb{R}$ and $\mathbb{R}_0^+$.
Random variables are denoted by uppercase letters,  e.g.,  $X$,  and their realizations by lowercase letters,  e.g.,  $x$.
For vectors, we use boldface notation,  i.e.,  lower case boldface letters such as $\xv$ for deterministic vectors.
We use $[1:X]$ to denote the set $\{1, \cdots, X\}$. 
We use $X^n$ for the tuple of random variables $(X_1, \cdots, X_n)$. 
We abbreviate \emph{independent and identically distributed} as \emph{i.i.d.} {and \emph{probability mass function} as \emph{pmf}.} Logarithms are taken with respect to base $2$.
We use $\perp$ to indicate independence between random variables.
	\section{A Single Receiver}\label{section:singlRx} 
	\subsection{System Model}\label{section:model}
	\begin{figure}[t]
		\begin{center}	
			\includegraphics[scale=0.8]{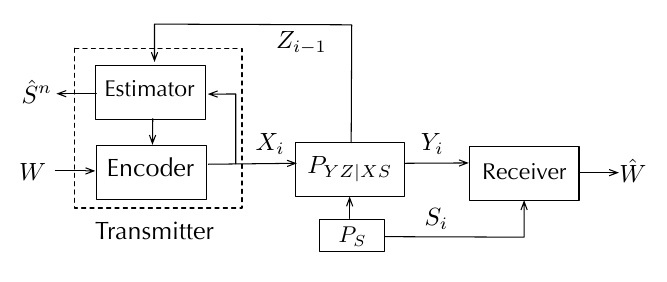}
			\caption{Joint sensing and communication model.}
			\label{fig:Model1User}
			\vspace{-2em}
		\end{center}
	\end{figure}
	
	Consider the point-to-point communication scenario depicted in Fig.~\ref{fig:Model1User}, where a transmitter wishes to communicate a message to a receiver over a memoryless state-dependent channel and simultaneously estimate the state from generalized feedback. In order to formulate the joint sensing and communication problem, we consider a state-dependent memoryless channel such that the channel output at the receiver $Y_{i}$ and the feedback signal $Z_i$ at a given time $i$ are generated according to its stationary channel law $P_{YZ|XS}(\cdot,\cdot|x_i,s_i)$ given the time-$i$ channel input $X_i=x_i$ and  state realization $S_i=s_i$, irrespective of the past inputs, outputs and state signals. 
	Except for some Gaussian examples, we assume that the channel states $S_i$,  inputs $X_i$,  outputs $Y_{i}$,  and feedback signals $Z_i$ take value in  finite sets $\Sc$, $\Xc$,  $\Yc$, and $\Zc$, respectively. The state sequence $\{S_i\}_{i\geq 1}$ is assumed i.i.d. according to a given state distribution $P_S(\cdot)$ and perfectly known to the receiver. 
	
	A $(2^{n\R},  n)$ code for the state-dependent memoryless channel (SDMC) consists of
	\begin{enumerate}
		\item a discrete message set $\Wc$ of size $|\Wc| \geq 2^{n\R}$;
		\item  a sequence of encoding functions $\phi_i\colon \Wc\times \Zc^{i-1} \to \Xc$,  for $i=1, 2, \ldots, n$;
		\item 
		a decoding function  $g\colon \Sc^n \times \Yc^n \to \Wc$; 
		\item
		a state estimator $h \colon \Xc^n \times \Zc^n \to \hat{\Sc}^n$,  where  $\hat{\Sc}$ denotes a given finite reconstruction  alphabet.
	\end{enumerate}

	For a given code,  the random message $W$ is uniformly distributed over the message set $\Wc$ and the inputs are obtained as $X_i=\phi_i(W,  Z^{i-1})$,  for $i=1, \ldots,  n$. The corresponding channel outputs $Y_{i}$ and $Z_i$ at time $i$ are obtained from the state $S_{i}$ and the input $X_i$ according to the  SDMC transition law $P_{YZ|SX}$. Let $\hat{S}^n:= (\hat{S}_{1}, \cdots, \hat{S}_{n} )=h(X^n,  Z^n)$ denote the state estimate at  the transmitter and $\hat{W}=g(S^n, Y^n)$  the decoded message at the receiver.

	The quality of the state estimates is measured by the expected average per-block distortion 
	\begin{equation}\label{def:distortion}
		\Delta^{(n)}:= \mathbb{E}[d(S^n,  \hat{S}^n)] 
		=\frac{1}{n} \sum_{i=1}^n \mathbb{E}[d(S_i,  \hat{S}_i)]
	\end{equation}
	where $d: \Sc\times \hat{\Sc} \mapsto \mathbb{R}_0^+$ is a given bounded \emph{distortion function}:
	\begin{equation}
		\max_{(s,  \hat{s})\in \Sc\times \hat{\Sc}} d(s,  \hat{s})<\infty.\label{eq:bounded_distortion}
	\end{equation}
	{In practical communication systems, we typically impose an expected cost constraint  on the channel inputs such as an average or peak power
		constraint. These cost constraints  can often be expressed as}
	\begin{equation}\label{eq:cost}
		\mathbb{E}[b(X^n)]  = \frac{1}{n} \sum_{i=1}^n \mathbb{E}[b(X_i)]
	\end{equation}
	{for some given} cost function{s} $b\colon \Xc \mapsto \mathbb{R}_0^+$.

	\begin{definition} \label{def:RDpairs}
		A rate-distortion-cost tuple $(\R,  \D, \B)$ is said achievable if there exists  a sequence (in $n$) of  $(2^{n\R},  n)$ codes that simultaneously satisfy
		\begin{subequations}\label{eq:asymptotics_user1}
			\begin{IEEEeqnarray}{rCl}
				\lim_{n\to \infty}	P_e^{(n)} &=&0,\label{eq:Pe} \label{eq:asymptotics:Pe}\\
				\varlimsup_{n\to \infty}\Delta^{(n)}& \leq& \D,  \label{eq:asymptotics:dist}\\
				\varlimsup_{n\rightarrow \infty} \frac{1}{n} \sum_{i=1}^n \mathbb{E}[b(X_i)] &\leq &\B \label{eq:cost_constraint}
			\end{IEEEeqnarray}				
		\end{subequations}
		for $	P_e^{(n)}:= \textnormal{Pr}\left( \hat{W} \neq W \right)$.
		
		
		The \emph{capacity-distortion-cost tradeoff $\C(\D, \B)$} is the largest rate $\R$ such that the rate-distortion-cost tuple $(\R,  \D, \B)$ is achievable.

	\end{definition}
	
	The main result of this section is an exact characterization of  $\C(\D,\B)$. 
	We begin by describing the optimal estimator $h$, which is independent of the choice of  encoding and decoding functions, and operates on a symbol-by-symbol basis, i.e., it computes estimate $\hat{S}_i$ only in function of $X_i$ and $Z_i$ but not  of the other inputs and feedback signals.
	
	\begin{lemma}\label{1user:lemma:Shat} 
		Define the function 
		\begin{IEEEeqnarray}{rCl}\label{eq:optimal_estimator}
			\hat{s}^*(x, z) &:= & {\rm arg}\min_{s'\in \hat{\Sc}} \sum_{s\in \Sc} P_{S|XZ}(s|x,z) d(s,  s'), \IEEEeqnarraynumspace
		\end{IEEEeqnarray}	
		where  ties can be broken arbitrarily and 
		\begin{equation}
			P_{S|XZ}(s|x,z)=\frac{ P_S(s) P_{Z|SX}(z|s,x)}{\sum_{\tilde{s} \in \Sc} P_S(\tilde{s}) P_{Z|SX}(z|\tilde{s},x)}.
		\end{equation}
		Irrespective of the choice of encoding and decoding functions, distortion $\Delta^{(n)}$ in \eqref{eq:asymptotics:dist} is minimized by the estimator
		\begin{equation}\label{eq:symbolwise}
			h^*(x^n,z^n):= ( \hat{s}^*(x_1,z_1), \hat{s}^*(x_2,z_2), \ldots, \hat{s}^*(x_n,z_n)).
		\end{equation}
			\textcolor{black}{Notice that the function $\hat{s}(\cdot, \cdot)$}	only depends on the SDMC channel law $P_{YZ|SX}$ and the state distribution $P_S$.
	\end{lemma}
	\begin{IEEEproof}
		See Appendix~\ref{app:lemma:Shat}.
	\end{IEEEproof}

	The optimal state estimator is thus a symbolwise estimator directly applied to the sequences observed at the transmitter. As we shall see later in this article,  this optimality of the symbolwise estimator   extends  also to the broadcast scenario. 

	Lemma \ref{1user:lemma:Shat} implies that we can focus without loss in optimality on a symbol-by-symbol deterministic estimator. 
	Based on \eqref{eq:optimal_estimator}, we define the estimation cost $c(x)$ {of the optimal estimator }as 
	\begin{equation} \label{eq:newcost}
		{c(x) : = \E{d(S,\hat{s}^*(X,Z))|X=x}.} 
	\end{equation}
	We are ready to present the capacity-distortion-cost tradeoff. 
	\subsection{Capacity-Distortion-Cost Tradeoff}\label{section:tradeoff}
	In order to characterize some useful properties of the capacity-distortion-cost function, we define the following sets:
	\begin{subequations}\label{eq:PxSet}
		\begin{align}  \label{eq:SetB}
			\Pc_{\B} &=\bigg\{ P_X ~\bigg|  \sum_{x\in \Xc} P_{X}(x) b(x) \leq \B\bigg\},
			\\
			\label{eq:SetD}
			\Pc_{\D} &=\bigg\{ P_X ~\bigg| \sum_{x\in \Xc} P_{X}(x) c(x) \leq \D \bigg\}.
		\end{align}
	\end{subequations}
	Then,  the minimum distortion for a given cost $\B$ is given by
	\begin{equation}\label{eq:Dmin}
		\D_{\min}(\B) := \min_{P_X \in \Pc_{\B}} \sum_{x\in \Xc} P_X(x) c(x). 
	\end{equation}

	\begin{definition}\label{def:Cinf}
		Define the \emph{information-theoretic tradeoff function} $\C_{\textnormal{inf}}: [\D_{\min}(\B),\infty)\times [0, \infty) \to \mathbb{R}_0^+$ as 
		\begin{equation}\label{eq:Cinf}
			\C_{\textnormal{inf}}(\D, \B) := \max_{P_X \in \Pc_{\D} \cap \Pc_{\B} } I(X;Y\mid S) 
		\end{equation}
		where
		$(X,  S,  Y,  Z)\sim$
		$P_{X} P_{S} P_{YZ|SX}$ and the maximum is over all $P_X$ satisfying both the distortion and cost constraints 
		\eqref{eq:SetD} and \eqref{eq:SetB}. 
	\end{definition}
	
	\begin{lemma}\label{lemma:properties} 
		Given a  SDMC $P_{YZ|SX}$ with state-distribution $P_S$, the 
		capacity-distortion-cost tradeoff function $\C_{\textnormal{inf}}(\D, \B)$ has the following properties.
		\begin{itemize}
			\item[i)] $\C_{\textnormal{inf}}(\D, \B)$ is non-decreasing and concave in $\D\geq \D_{\min}(\B)$ and $\B \geq 0$. 
			\item[ii)] 
			$\C_{\textnormal{inf}}(\D, \B)$ saturates at the channel capacity: 
			\begin{equation}
				\C_{\textnormal{inf}}(\D, \B) = \C_{\textnormal{NoEst}}(\B), \quad \forall \D \geq \D_{\max}(\B),
			\end{equation}
			where $\C_{\textnormal{NoEst}}(\B):=\max_{P_X \in \Pc_{\B}} I(X;Y|S)$ denotes the classical channel capacity of the SDMC for a given cost $\B$,  
			and $\D_{\max}(\B)$ denotes the corresponding distortion
			\begin{equation}
				\D_{\max}(\B) := \sum_{x\in \Xc} P_{X_{\max}}(x) c(x). 
			\end{equation}
			for $P_{X_{\max}} := \argmax_{P_X \in \Pc_{\B}} I(X;Y|S)$. 
		\end{itemize}
	\end{lemma}
	
	\begin{IEEEproof} 
		The proof is a straightforward extension of \cite[Corollary 1]{zhang2011joint} to the case of two cost functions and the state dependent channel. The nondecreasing property follows immediately from the definition in \eqref{eq:Cinf} because
		we have $\Pc_{\D_1}\subseteq \Pc_{\D_2}$ and $\Pc_{\B_1}\subseteq \Pc_{\B_2}$ for any $\D_1\leq \D_2$ and $\B_1\leq \B_2$. 
		
		In order to verify the concavity of $\C_{\textnormal{inf}}(\D,\B)$ with respect to $(\D, \B)$, 
		we consider time-sharing between two input distributions, denoted by $P_X^{(1)}$ and $P_X^{(2)}$, that achieve $\C_{\textnormal{inf}}(\D_1, \B_1)$ and $\C_{\textnormal{inf}}(\D_2, \B_2)$, respectively. 
		To make the dependency of the mutual information with respect to the input distribution {more explicit}, we adapt the 
		following notation: for any pmf $P_X$ over the input alphabet $\Xc$, let {$\Ic(P_X, P_{Y|XS} \mid P_S):= I(X;Y\mid S)$} for  
		$(S,X,Y)\sim P_S P_X P_{Y|XS}$.  

			For any $\theta\in(0,1)$, we have: 
			\begin{IEEEeqnarray}{rCl}\label{eq:concavity}
			&&	\theta \C_{\textnormal{inf}}(\D_1, \B_1) + (1-\theta) \C_{\textnormal{inf}}(\D_2, \B_2)
				 \nonumber \\
			&&	\stackrel{(a)}=  \theta \Ic\left(P_X^{(1)}, P_{Y|XS}  \; \Big| \;  P_S\right) 
			 \nonumber \\
			&&\hspace{3cm}+(1-\theta)  \Ic\left(P_X^{(2)}, P_{Y|XS} \; \Big| \; P_S\right) \nonumber \\
				&& \stackrel{(b)}\leq \Ic\left(\theta P_X^{(1)} + (1-\theta) P_X^{(2)}, P_{Y|XS} \; \Big| \;  P_S\right)\nonumber  \\
					&& \stackrel{(c)}= \C_{\textnormal{inf}}\left(\theta \D_1 + (1-\theta) \D_2,\theta \B_1 + (1-\theta) \B_2 \right).
			\end{IEEEeqnarray}
			where (a) follows by definition, (b) follows from the concavity of the mutual information functional with respect to the input distribution, (c) follows by the linearity of the constraints and because for any $k=1,2$ the pmf  $P_{X}^{(k)}$ has expected cost no larger than $\B_k$  and expected distortion no larger than $\D_k$.  
			This establishes the concavity of $\C_{\textnormal{inf}}(\D, \B)$. 
	\end{IEEEproof}

	We now state the main result of this section.
	\begin{theorem} \label{th:tradeoff}
		The  capacity-distortion-cost tradeoff of a SDMC $P_{YZ|SX}$ with state-distribution $P_S$ is:
		\begin{equation}\label{Th1:user1:Tradeoff}
			\C(\D, \B) = \C_{\textnormal{inf}}(\D, \B),  \quad \D \geq \D_{\min}(\B), \;\; \B \geq 0. 
		\end{equation}
	\end{theorem}
	\begin{IEEEproof}
	See Appendix~\ref{app:main_result_P2P}.
	\end{IEEEproof}

	The proof of Theorem~\ref{th:tradeoff} is similar to the proof of the classic capacity-cost function \cite{MUIT_Kramer}, except that one also has to account for the sensing performance. Both in the converse proof and the achievability proof, this  can be accomplished  by evaluating the performance of the optimal (per-symbol) estimator $\hat{s}^*(\cdot,\cdot)$  in Lemma~\ref{1user:lemma:Shat}.  In particular,  a standard random coding argument can be used to prove achievability of Theorem~\ref{th:tradeoff}. 

On a different note,  capacity of a memoryless channel is known to be achieved  with i.i.d. inputs. Also because of the memoryless nature of the optimal estimator $h(\cdot, \cdot)$ in Lemma~\ref{1user:lemma:Shat}, this  observation extends to our  joint sensing and communication setup.

Appendix~\ref{app:BA} presents a Blahut-Arimoto type algorithm that can be used to solve the optimization problem \eqref{eq:Cinf}, which  characterizes the capacity-distortion-cost tradeoff $\C_{\textnormal{inf}}(\D, \B)$. It is used to evaluate the capacity-distortion-cost tradeoff for   the Gaussian example in Subsection~\ref{sec:Gaussian} ahead.

	Combining Lemma~\ref{lemma:properties} and Theorem~\ref{th:tradeoff}, we can conclude that the rate-distortion tradeoff function $\C(\D,\B)$ is non-decreasing and concave in $\D\geq \D_{\min}$ and $\B\geq0$, and {for any $\B\geq 0$  it saturates at the channel capacity $\C_{\textnormal{NoEst}}(\B)$}.  
	{For many channels,  given  $\B\geq 0$, the tradeoff $\C(\D, \B)$ is strictly increasing in $\D$ until it reaches  $\C_{\textnormal{NoEst}}(\B)$. However, for SDMBCs and costs $\B\geq 0$ where the capacity-achieving input distribution $P_{X_{\max}}:= \argmax_{P_X \in \mathcal{P}_{\B}} I(X;Y\mid S)$ also achieves minimum distortion $\D_{\min}(\B)$ in \eqref{eq:Dmin},  the capacity-distortion tradeoff is constant  $\C(\D,\B)=\C_{\textnormal{NoEst}}(\B)$, irrespective of the allowed distortion $\D$. This is in particular the case, when the expected distortion $\E{d(S, \hat{s}^*(X,Z))}$ does not depend on the input distribution $P_X$.  The following corollary identifies a set of SDMCs $P_{YZ|SX}$ and state distributions $P_S$ where this holds for all costs $\B\geq 0$.}
	
	\begin{corollary}\label{cor1:notradeoff}
		Assume that  there exists a function $\psi(\cdot)$ with domain $\Xc\times \Zc$ so that  irrespective of the input distribution $P_X$ the following two 
		conditions hold:
		\begin{IEEEeqnarray}{rCl}
			&(S, \psi (X, Z)) \perp  X, \label{1user:cond1}\\
			&S \markov \psi (X,  Z)\markov (X,  Z),  \label{1user:cond2}
		\end{IEEEeqnarray} 
		for $(S,X,Z)\sim P_{S} P_{X} P_{Z|SX}$.
		In this case,  {for any given $\B$,} the rate-distortion tradeoff function $\C(\D,\B)$ is constant over $\D\geq \D_{\min}$ and equal to the channel capacity of the SDMC:
		\begin{equation}
			\C(\D, \B)= \C_{\textnormal{NoEst}}(\B), \qquad \forall \D\geq \D_{\min}(\B), \;\; \B\geq 0.
		\end{equation}
		\begin{IEEEproof}
			See Appendix~\ref{app:notradeoff}.
		\end{IEEEproof}	
	\end{corollary}
	
	The following state-dependent erasure channel satisfies the conditions in above corollary. Let $S$ be Bernoulli-$p$ and $Y$ equal to the erasure symbol ``?" when $S=1$ and $Y=X$ when $S=0$. Moreover, assume perfect output feedback, i.e.,  $Y=Z$. For the choice  $\psi(X,Z)=\mathbbm{1}\{ Z= ``?"\}=S$ both Markov chains in Corollary~\ref{cor1:notradeoff} are trivially satisfied because $S$ and $X$ are independent.

		\textcolor{black}{	\begin{remark}\label{rem:CSI}
	Theorem~\ref{th:tradeoff} is easily adapted to the more general case of imperfect channel state information (CSI), i.e., to a scenario where the receiver does not observe the state-sequence $S^n$  but a related sequence $S_R^n$, where $(S^n, S_R^n)$ are i.i.d. according to an arbitrary distribution $P_{SS_R}$. In this case, Theorem~\ref{th:tradeoff} remains valid if in Definition \eqref{eq:Cinf} the state $S$  is replaced by  $S_R$, i.e., 
			\begin{IEEEeqnarray}{rCl}\label{Th1:user1:Tradeoff_imperf}
				\C^{\textnormal{imp}}(\D, \B)& =& \max_{P_X \in \Pc_{\D} \cap \Pc_{\B} } I(X;Y\mid S_R) , 
				\\
				\nonumber  && \hspace{2cm}
				\qquad \D \geq \D_{\min}(\B), \;\; \B \geq 0,
			\end{IEEEeqnarray}
			where $(X,S_R,Y,Z)\sim P_XP_{SS_R}P_{YZ\mid SS_RX}$ and	the definitions of the sets 
			$\Pc_{\B}$ and 	$\Pc_{\D}$  are kept as in \eqref{eq:SetB} and \eqref{eq:SetD}, same as the definition of the function $c(x)$ in \eqref{eq:newcost}.	
			\\
			Notice that the symbolwise estimator in \eqref{eq:symbolwise} remains optimal also in this related setup.			
		\end{remark}
	\begin{proof}
	See Appendix~\ref{app:imperfectCSIR}.
\end{proof}	}

	\newcommand{\capa}{\C_{\textnormal{NoEst}}}
	\subsection{Examples}
	Before presenting our examples, we present two baseline schemes.
	\subsubsection{Baseline Schemes} 
	
	We consider two baseline schemes that time share (TS) between two operating modes. 
	{The first baseline scheme, termed \emph{Basic TS scheme},  is unable to simultaneously perform  the sensing and communication tasks and  splits its  resources (time or bandwidth)
		between the following two modes: 
		\begin{itemize}
			\item \emph{\underline{Sensing mode without communication} {(achieves rate-distortion pair $(0, \D_{\min}(\B))$)}}\\[0.15em]
			The input pmf $P_X$ is chosen to minimize the  distortion:
			\begin{equation}\label{eq:P_Xmin}
				P_{X_{\min}}:= \argmin_{P_X\in \Pc_{\Bc} } \sum_x P_X(x) c(x),
			\end{equation}
			and thus the minimum distortion $\D_{\min}(\B)$ defined in \eqref{eq:Dmin} is achieved.
			Due to the lack of communication capability, the communication rate is zero. \\[0.01em]
			\item \emph{\underline{Communication mode without sensing} (achieves $(\capa(\B),  \D_{\textnormal{trivial}}(\B))$)}\\[0.15em]
			The input pmf $P_X$ is chosen to maximize  the  rate:   
			\begin{equation} \label{eq:Pmax}
				P_{X_{\max}}=\argmax_{P_X\in \Pc_{\Bc} } I(X;Y \mid S), 
			\end{equation}
			and this mode thus communicates at a rate equal to the channel capacity $\capa(\B)$. 
			Due to the lack of proper sensing capabilities, the estimator is set to  a constant value regardless of the feedback and the input signals. The mode thus
			achieves  distortion 
			\begin{equation}\label{eq:ConstEstimator}
				\D_{\textnormal{trivial}}(\B) :=\min_{s'\in \hat{\Sc}} \sum_{s\in \Sc} P_{S}(s) d(s,  s'). 
			\end{equation}
		\end{itemize}
		The second  baseline scheme is called \emph{Improved TS scheme} and  can simultaneously perform the communication and sensing tasks. This scheme time-shares between the following modes.    
		\begin{itemize}
			\item \emph{\underline{Sensing mode with communication} (achieves $(\R_{\min}(\B), \D_{\min}(\B))$)}\\
			The input pmf $P_X$ is choosen according to \eqref{eq:P_Xmin} to achieve the minimum distortion. The 
			chosen pmf $P_{X_{\min}}$ can achieve the following communication rate: 
			\begin{equation}
				\R_{\min}:=  I(X_{\min};Y\mid S), \qquad \textnormal{for } X_{\min}\sim P_{X_{\min}}.
			\end{equation}
			\item \emph{\underline{Communication mode with sensing} (achieves $(\capa (\B),\D_{\max}(\B))$)}\\
			The input pmf $P_{X_{\max}}$ is chosen as in \eqref{eq:Pmax} to maximize the communication rate. The mode thus communicates at the capacity $\capa(\B)$ of the channel. 
			Sensing is performed by means of the optimal estimator in  \eqref{eq:optimal_estimator}. The mode thus achieves distortion 
			\begin{equation}
				\D_{\max} := \sum_{x\in \Xc} P_{X_{\max}} (x) c(x), \quad \textnormal{for } X_{\max} \sim P_{X_{\max}}.
			\end{equation}
		\end{itemize}
	}
	{It is worth noticing that for any cost $\B\geq 0$, the two operating points of the two modes in the Improved TS scheme, $(\R_{\min}(\B), \D_{\min}(\B))$ and $(\capa (\B),\D_{\max}(\B))$,  also lie  on the capacity-distortion-cost tradeoff curve $\C(\D,\B)$ presented in Theorem~\ref{th:tradeoff}. These two  points are thus also operating points of any optimal co-design scheme. As we will see at hand of the following examples, all other operating points of the Improved TS scheme  are  typically suboptimal compared to an optimal co-design scheme.}

	\subsubsection{Example 1: Binary Channel with Multiplicative Bernoulli State}
	Consider a channel $ Y = S X$ with binary alphabets $\Xc=\Sc=\Yc=\{0,1\}$ and 
	where the state $S$ is Bernoulli-$q$, for $q\in (0,1)$. 
	We assume  perfect output feedback to the transmitter $Y=Z$, and consider the Hamming distortion measure $d(s,  \hat{s}) = s \oplus \hat{s}$.  {No cost constraint is imposed.}

	The following corollary specializes  Theorem~\ref{th:tradeoff} to this example.  \begin{corollary}\label{cor:ex2}
		The capacity-distortion tradeoff of a binary channel with  multiplicative Bernoulli state is given by
		\begin{equation}
			\C(\D) = q  H_{\textnormal{b}}\left(\frac{\D}{\min\{q, 1-q\}}\right),
		\end{equation}
		
		where $H_{\textnormal{b}}(p)$ denotes the binary entropy function.  In other words, the curve $\C(\D)$ is parameterized as 
		\begin{equation}
			\{( \C=q H_{\textnormal{b}}(p), \; \D= p \min\{q, 1-q\}) \colon p\in[0,1/2]\}.
		\end{equation}
	\end{corollary}
	
	\begin{IEEEproof}
		Since $Y$ is deterministic given $(S,  X)$, and  it equals $0$ whenever $S=0$, we have:
		\begin{IEEEeqnarray}{rCl}
			I(X;Y \mid S) 
			&=&P_S(0)  H(Y \mid S=0) 
			\nonumber\\&&\qquad+ P_S(1)  H(Y \mid S=1) \IEEEeqnarraynumspace \nonumber\\
			&=&P_S(1)  H(X).
		\end{IEEEeqnarray}
		Setting $p:=P_{X}(0)$, 
		we obtain 
		\begin{equation} \label{eq:Ip}
			I(X;Y \mid S)= q H_{\textnormal{b}}(p).
		\end{equation} 
		
		To calculate the distortion,  we notice that the optimal estimator $\hat{s}^* (\cdot, \cdot)$ in Lemma \ref{1user:lemma:Shat} sets 
		\begin{equation}
			\hat{s}^*(x,z)= \begin{cases} z, & \textnormal{if } x=1 \\
				\argmax_{s\in\{0,1\}}  P_{S}(s), &\textnormal{if }  x=0.
			\end{cases}
		\end{equation}
		In fact, whenever $x=1$ the transmitter acquires full state knowledge because $z=y=s$. In this case $c(x=1)=0$. For $x=0$, the transmitter does not receive any useful information about the state and hence uses the best constant estimator, irrespective of the feedback $z$. In this case, 
		\begin{IEEEeqnarray}{rCl}\label{eq:SingleEstCost}
			c(x=0)&=& \E{ d\Big(S,  \argmax_{s\in\{0,1\}}  P_{S}(s)\Big)\Big|X=0} \nonumber
			\\&=& \min_{s\in\{0,1\}} P_S(s) = \min\{q,1-q\}, 
		\end{IEEEeqnarray}
		where we used the independence of $S$ and $X$. 
		The expected distortion of the optimal estimator thus evaluates to:
		\begin{align}\label{eq:SingleD}
			\D  & = \sum_{x} P_X(x) c(x) = P_X(0) c(0) = p \min\{q, 1-q\}. 
		\end{align}
	\end{IEEEproof}

	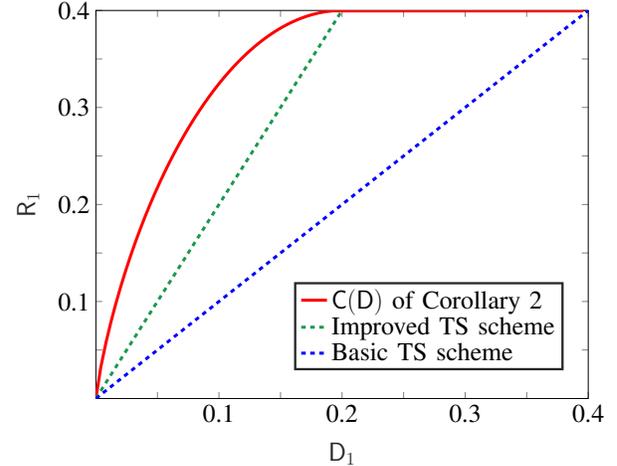
\begin{figure}[!h]
		\centering
		\hspace{-1cm}

		\begin{tikzpicture}[scale=0.57]

			\begin{axis}[%
				width=4.521in,
				height=3.566in,
				at={(0.18in,0.481in)},
				scale only axis,
				xmin=0,
				xmax=0.4,
				xlabel style={at={(0.5,-0.04)},font=\color{white!15!black}},
				xlabel={\LARGE$\D_1$},
				ymin=0,
				ymax=0.4,
				ylabel style={at={(-0.03,0.5)},font=\color{white!15!black}},
				ylabel={\LARGE$\R_1$},
				legend style={legend cell align=left,  at={(0.95,0.3)}, align=left, draw=white!15!black, line width=1.5pt},
				%
				%
				x tick label style={font=\LARGE},
				y tick label style={font=\LARGE},
				xticklabels={ , ,  , 0.1, ,0.2, ,0.3, ,0.4,,},
				yticklabels={ , ,  , 0.1, ,0.2, ,0.3, ,0.4,,},
				]
				
				\addplot [color=red, line width=2pt]
				table[row sep=crcr]{%
					0	0\\
					0.004	0.0323172543583645\\
					0.00800000000000001	0.0565762170167283\\
					0.012	0.0777567431326305\\
					0.016	0.096916875632966\\
					0.02	0.114558782846383\\
					0.024	0.130977967661791\\
					0.028	0.146369460360089\\
					0.032	0.160871676080909\\
					0.036	0.174587926825641\\
					0.04	0.187598237435712\\
					0.044	0.199966383265811\\
					0.048	0.211744346114946\\
					0.052	0.222975274011196\\
					0.056	0.233695524657142\\
					0.06	0.24393612188656\\
					0.064	0.253723821856226\\
					0.068	0.263081911497688\\
					0.072	0.272030818291312\\
					0.076	0.280588583953559\\
					0.08	0.288771237954945\\
					0.084	0.296593095972509\\
					0.088	0.304067001184786\\
					0.092	0.311204521418615\\
					0.096	0.318016111753809\\
					0.1	0.324511249783653\\
					0.104	0.330698548997047\\
					0.108	0.33658585448327\\
					0.112	0.342180324224052\\
					0.116	0.347488498535762\\
					0.12	0.352516359692277\\
					0.124	0.357269383351143\\
					0.128	0.361752583089798\\
					0.132	0.365970549111891\\
					0.136	0.369927481989212\\
					0.14	0.373627222150197\\
					0.144	0.377073275702197\\
					0.148	0.380268837074826\\
					0.152	0.38321680889052\\
					0.156	0.385919819402035\\
					0.16	0.388380237781867\\
					0.164	0.39060018750313\\
					0.168	0.392581558013461\\
					0.172	0.394326014871568\\
					0.176	0.395835008488822\\
					0.18	0.397109781595123\\
					0.184	0.39815137552809\\
					0.188	0.398960635427096\\
					0.192	0.399538214398081\\
					0.196	0.399538214398081\\
					0.2	0.399538214398081\\
					0.204	0.399538214398081\\
					0.208	0.399538214398081\\
					0.212	0.399538214398081\\
					0.216	0.399538214398081\\
					0.22	0.399538214398081\\
					0.224	0.399538214398081\\
					0.228	0.399538214398081\\
					0.232	0.399538214398081\\
					0.236	0.399538214398081\\
					0.24	0.399538214398081\\
					0.244	0.399538214398081\\
					0.248	0.399538214398081\\
					0.252	0.399538214398081\\
					0.256	0.399538214398081\\
					0.26	0.399538214398081\\
					0.264	0.399538214398081\\
					0.268	0.399538214398081\\
					0.272	0.399538214398081\\
					0.276	0.399538214398081\\
					0.28	0.399538214398081\\
					0.284	0.399538214398081\\
					0.288	0.399538214398081\\
					0.292	0.399538214398081\\
					0.296	0.399538214398081\\
					0.3	0.399538214398081\\
					0.304	0.399538214398081\\
					0.308	0.399538214398081\\
					0.312	0.399538214398081\\
					0.316	0.399538214398081\\
					0.32	0.399538214398081\\
					0.324	0.399538214398081\\
					0.328	0.399538214398081\\
					0.332	0.399538214398081\\
					0.336	0.399538214398081\\
					0.34	0.399538214398081\\
					0.344	0.399538214398081\\
					0.348	0.399538214398081\\
					0.352	0.399538214398081\\
					0.356	0.399538214398081\\
					0.36	0.399538214398081\\
					0.364	0.399538214398081\\
					0.368	0.399538214398081\\
					0.372	0.399538214398081\\
					0.376	0.399538214398081\\
					0.38	0.399538214398081\\
					0.384	0.399538214398081\\
					0.388	0.399538214398081\\
					0.392	0.399538214398081\\
					0.396	0.399538214398081\\
				};
				\addlegendentry{\LARGE  $\C(\D)$ of Corollary~\ref{cor:ex2}}
				
				\addplot [color=mycolor1, dashed,   line width=2.0pt]
				table[row sep=crcr]{%
					0	0\\
					0.2	0.4\\
				};
				\addlegendentry{\LARGE Improved TS scheme}

				\addplot [color=blue, dashed,  line width=2.0pt]
				table[row sep=crcr]{%
					0	0\\
					0.4	0.4\\
				};
				\addlegendentry{\LARGE Basic TS scheme}

			\end{axis}

			\begin{axis}[%
				width=5.833in,
				height=4.375in,
				at={(0in,0in)},
				scale only axis,
				xmin=0,
				xmax=1,
				ymin=0,
				ymax=1,
				axis line style={draw=none},
				ticks=none,
				axis x line*=bottom,
				axis y line*=left
				]
			\end{axis}
		\end{tikzpicture}%
		\vspace{-0.2cm}
		\caption{Capacity-distortion tradeoff of  the binary channel with multiplicative Bernoulli state of parameter
			$q=0.4$.}
		\label{fig:1user:BinaryExample}
	\end{figure}
	 
	{
		The  capacity-distortion tradeoff of Corollary~\ref{cor:ex2} is  illustrated in Fig.~\ref{fig:1user:BinaryExample} for state parameter $q=0.4$. The figure also compares the performances of the two baseline TS schemes. We observe a significant gain of an optimal co-design  scheme over the two TS baseline schemes. We conclude this example with a derivation of the parameters of the TS schemes.
		
		The capacity-achieving input distribution is easily found as $P_{X_{\max}}(0)=P_{X_{\max}}(1)=1/2$, and by \eqref{eq:Ip} and \eqref{eq:SingleD} we find  $\capa=q$ and $\D_{\max}=\min\{q,1-q\}/2$.
		Minimum distortion $\D_{\min}=0$ is achieved by always sending $X=1$, i.e., $P_{X_{\min}}(1)=1$ and $P_{X_{\min}}(0)=0$,  in which case $D_{\min}=0$ and  $R_{\min}=0$, see also \eqref{eq:Ip} and \eqref{eq:SingleD}. The Improved TS scheme thus   achieves all pairs on the line connecting the two points $(0,0)$ with $(q,\min\{q, 1-q\}/2 )$. To determine the performance of the basic TS scheme, we recall that the 
		best constant estimator (that does not consider the feedback) is $\hat{s}_{\rm const} = \argmax_{s\in\{0,1\}}  P_{S}(s)$ , which allows to conclude that  $\D_{\textnormal{trivial}}=\min\{q, 1-q\}$. The basic TS scheme thus achieves all rate-distortion pairs on the line connecting the points $(0,0)$ and $(q, \min\{q, 1-q\})$.}

	\subsubsection{Example 2: Real Gaussian Channel with Rayleigh Fading}\label{sec:Gaussian}
	\begin{figure}[h]
		\begin{center}	
			%
			%
			\definecolor{mycolor1}{rgb}{1.00000,0.00000,1.00000}%
			\begin{tikzpicture}[scale=0.48]
				
				\begin{axis}[%
					width=6.028in,
					height=4.754in,
					at={(1.011in,0.642in)},
					scale only axis,
					xmin=0.12,
					xmax=1,
					xlabel style={at={(0.5,-0.04)},font=\color{white!15!black}},
					xlabel={\huge Distortion},
					ymin=0,
					ymax=1.27,
					ylabel style={at={(-0.04,0.475)},font=\color{white!15!black}},
					ylabel={\huge{Capacity}},
					axis background/.style={fill=white},
					legend style={at={(0.634,0.229)}, anchor=south west, legend cell align=left, align=left, draw=white!15!black},
					x tick label style={font=\LARGE},
					y tick label style={font=\LARGE},
					xticklabels={, ,0.15, ,0.25, ,0.35, ,0.45, ,0.55, ,0.65, ,0.75, ,0.85, ,0.95},
						yticklabels={, , ,0.2, ,0.4, ,0.6, ,0.8, ,1, ,1.2, ,},
					]
					\addplot [color=blue, dashed, line width=3.0pt]
					table[row sep=crcr]{%
						0.1594	0\\
						1	1.213\\
					};
					\addlegendentry{\huge Basic TS}
					
					
					\addplot [color=red, line width=3.0pt, forget plot]
					table[row sep=crcr]{%
						0.367	1.213\\
						1	1.213\\
					};
					
					\addplot [color=green!60!black, dashed, line width=3.0pt]
					table[row sep=crcr]{%
						0.1594	0.7325\\
						0.367	1.213\\
					};
					\addlegendentry{\huge Improved TS}
					
					\addplot [color=red, line width=3.0pt, only marks, mark=asterisk, mark options={solid, red}, forget plot]
					table[row sep=crcr]{%
						0.1594	0.7325\\
						0.1599	0.7326\\
						0.1604	0.7347\\
						0.1613	0.7418\\
						0.1625	0.7616\\
						0.1632	0.7627\\
						0.1633	0.7696\\
						0.1646	0.7756\\
						0.1656	0.777\\
						0.1661	0.7793\\
						0.1664	0.7874\\
						0.1676	0.79\\
						0.1713	0.794\\
						%
						0.172	0.824\\
						%
						%
						%
						%
						0.18	0.8626\\
						0.1872	0.9141\\
						0.2007	0.9805\\
						0.2237	1.0556\\
						0.2581	1.133\\
						0.2912	1.1759\\
						0.3017	1.1862\\
						0.3071	1.191\\
						0.3128	1.1948\\
						0.3246	1.2019\\
						0.3368	1.2073\\
						0.3494	1.211\\
						0.3624	1.2128\\
						0.3624	1.2128\\
						0.367	1.213\\
					};
					
					\addplot [color=red, line width=3.0pt]
					table[row sep=crcr]{%
						0.367	1.213\\
						0.3624	1.2128\\
						0.3494	1.211\\
						0.3368	1.2073\\
						0.3246	1.2019\\
						0.3128	1.1948\\
						0.3017	1.1862\\
						0.2912	1.1759\\
						0.2815	1.1639\\
						0.2718	1.1515\\
						0.2581	1.133\\
						0.2237	1.0556\\
						0.2007	0.9805\\
						0.1872	0.9141\\
						0.1633	0.7696\\
						0.1594	0.7325\\
					};
					\addlegendentry{\huge Co-design}
					
				\end{axis}
				
				\begin{axis}[%
					width=7.778in,
					height=5.833in,
					at={(0in,0in)},
					scale only axis,
					xmin=0,
					xmax=1,
					ymin=0,
					ymax=1,
					axis line style={draw=none},
					ticks=none,
					axis x line*=bottom,
					axis y line*=left,
					x tick label style={font=\bfseries\boldmath},
					y tick label style={font=\bfseries\boldmath},
					z tick label style={font=\bfseries\boldmath},
					]
				\end{axis}
			\end{tikzpicture}%
			\caption{Capacity-distortion tradeoff of fading AWGN channel $\B=10$ dB and $\sigma^2_{\textnormal{fb}}=1$.}
			\label{fig:Gaussian}
		\end{center}
	\end{figure}
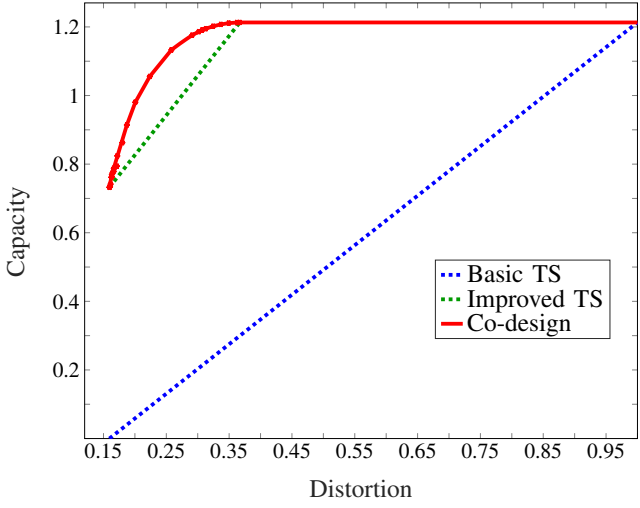

	\newcommand{\Vfb}{\sigma_{\textnormal{fb}}^2}
	We consider the real Gaussian channel with Rayleigh fading:
	\begin{equation}\label{ex:Gaussian1user}
		Y_i = S_i X_i + N_i ,
	\end{equation}
	where $X_i$ is the channel input satisfying {$\varlimsup_{n \to \infty} \frac{1}{n}\sum_i\E{|X_i|^2} \leq \B=10$dB},  and both sequences $\{N_i\}$ and $\{S_i\}$ are independent of each other and  i.i.d. Gaussian  with zero mean and unit variance. 
	The transmitter observes the noisy feedback
	\begin{equation}
		Z_i = Y_i + N_{\textnormal{fb},i},
	\end{equation}
	where $\{N_{\textnormal{fb},i}\}$ are i.i.d. zero-mean Gaussian of variance $\Vfb \geq 0$. {We consider the quadratic distortion measure $d(s,\hat{s})=(s-\hat{s})^2$.}

	First, we characterize the two operating points achieved by the Improved TS baseline scheme. 
	The capacity of this channel is achieved with a Gaussian input $X_{\max}\sim \mathcal{N}(0,\B)$, and thus the communication mode with sensing achieves the rate-distortion pair 
	\begin{IEEEeqnarray}{rCl}
		\capa(\B)&=&\frac{1}{2} \E{\log (1+|S|^2 \B)}=1.213, \quad
		\\ \D_{\max}(\B)&=&\E{\frac{(1+\Vfb)}{1+|X_{\max}|^2+\Vfb}}=0.367,
	\end{IEEEeqnarray}
	where  we have set $\sigma^2_{\textnormal{fb}}=1$ and $P=10$dB to obtain the numerical values.
	Minimum distortion $\D_{\min}$ is achieved by $2$-ary pulse amplitude modulation (PAM),  and thus  the sensing mode with communication  achieves rate-distortion pair
	\begin{equation}
		\R_{\min}(\B)=0.733, \quad \D_{\min}(\B)=\frac{1+\sigma_{\textnormal{fb}}^2}{1+P+\Vfb}={0.166},
	\end{equation}
	where the numerical value again corresponds to $\sigma_{\textnormal{fb}}=1$ and $B=10$dB.
	{Next, we characterize the performance of the basic TS baseline scheme. The best constant estimator for this channel is $\hat{s}=0$, and the communication mode without sensing  achieves rate-distortion pair $(\capa(\B),\D_{\textnormal{trivial}}(\B)=1)$. 
		The sensing mode without communication achieves rate-distortion pair $(0, \D_{\min}(\B))$.    
		
		In Fig.~\ref{fig:Gaussian}, we compare the rate-distortion tradeoff achieved by these two TS baseline schemes with a numerical approximation of the capacity-distortion-cost tradeoff $\C(\D,\B)$ of this channel. 
		As previously explained,  $\C(\D,\B)$ also passes through the two end points $(\R_{\min}(\B), \D_{\min}(\B))$ and $(\capa(\B),\D_{\max}(\B))$  of the  Improved TS scheme. 
		We use the  Blahut-Arimoto type Algorithm~1   to obtain a numerical approximation of the points on  $\C(\D,\B)$ in between these two operating points. Specifically, 
		the input alphabet is quantized to a $M=16$-ary PAM constellation 
		\begin{equation} \Xc_{\qt}:=\{(2m-1-M)\kappa, m=1,\cdots ,M\},
		\end{equation} where $\kappa:= \sqrt{{3P}/{(M^2-1)}}$. The Gaussian noise $N$ is quantized with a centered equally-spaced $50$-points alphabet, and the state $S$ is quantized by applying an equally-spaced 8000-points quantizer on the Chi-square distributed random variable $S^2$.  
		Denoting the quantized input, noise, and state by $X_{\qt}$, $N_{\qt}$, and $S_{\qt}$, we keep our multiplicative-state, additive-noise channel model to generate the channel outputs used to run Algorithm~1 to obtain the numerical approximations:
		\begin{equation}
			Y_{\qt}=S_{\qt} X_{\qt} +N_{\qt}.
	\end{equation}}

	\section{Multiple Receivers}\label{section:physBC}
	
	In this section, we consider  joint sensing and communication over two-receiver broadcast channels. 
	
	\subsection{System Model}\label{sect:Model}
	%
	\begin{figure}[h]
		\centering
		\vspace{-7mm}
		\includegraphics[scale=0.7]{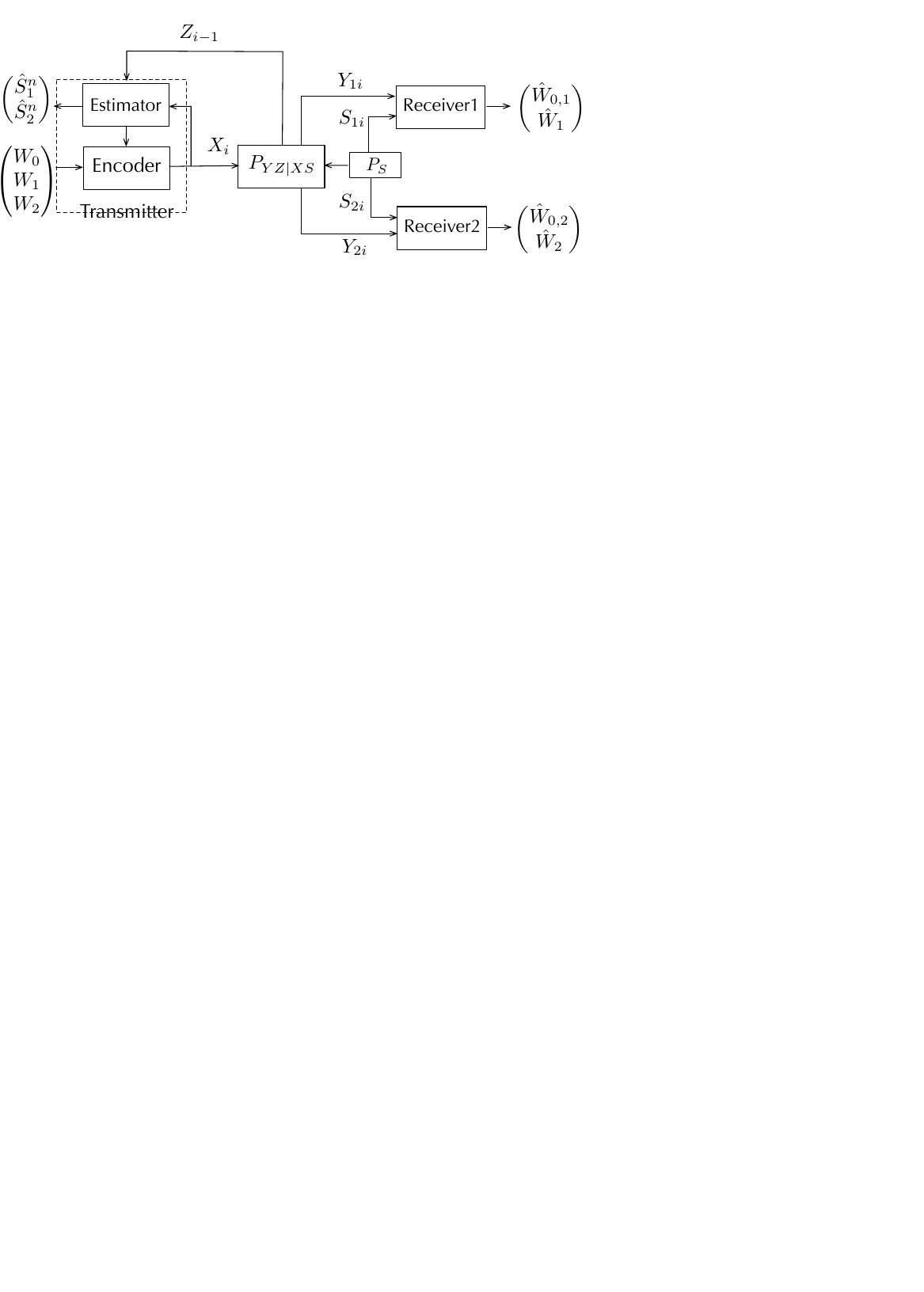}
			\vspace{-16cm}
		\caption{State-dependent  broadcast channel with generalized feedback and state-estimator at the transmitter.}
		\label{fig:ModelBC}
	
	\end{figure}
	Consider the two-receiver broadcast channel scenario depicted in Fig.~\ref{fig:ModelBC}.  The model comprises a two-dimensional memoryless state sequence $\{(S_{1, i},  S_{2, i})\}_{i\geq 1}$ whose samples at any given time $i$ are distributed according to a given joint law $P_{S_1S_2}$ over the state alphabets $\Sc_1\times \Sc_2$. Receiver~1 observes state sequence $\{S_{1,i}\}$ and Receiver~2 observes state sequence $\{S_{2,i}\}$. The transmitter communicates with both receivers over a state-dependent memoryless broadcast channel (SDMBC), where given time-$i$ input $X_i=x$ and state realizations $S_{1,i}=s_{1}$ and $S_{2,i}=s_{2}$,  the time-$i$ outputs $Y_{1,i}$ and $Y_{2,i}$ observed at the receivers and the transmitter's  feedback signal $Z_i$ are distributed according to the stationary channel transition law $ P_{Y_1Y_2Z|S_1S_2X}(\cdot,\cdot,\cdot|s_1,s_2,x)$. 
	We again assume that all 
	alphabets $\Xc,  \Yc_1,  \Yc_2,  \Zc, \Sc_1, \Sc_2$ are finite. 
	
	The goal of the transmitter is to convey a common message $W_0$ to both receivers and individual messages $W_1$ and $W_2$ to Receivers 1 and 2, respectively, while estimating the states sequences $\{S_{1,i}\}$ and $\{S_{2,i}\}$ within some target distortions. For simplicity, the input cost constraint is omitted. 

	A $(2^{n\R_0},2^{n\R_1}, 2^{n\R_2},  n)$ code for an SDMBC  thus consists of
	\begin{enumerate}
		\item three message sets $\Wc_0= [1:2^{n\R_0}]$, $\Wc_1= [1:2^{n\R_1}]$, and $\Wc_2= [1:2^{n\R_2}]$;
		\item a sequence of encoding functions $\phi_i\colon \Wc_0\times\Wc_1\times \Wc_2 \times \Zc^{i-1} \to \Xc$,  for $i=1, 2, \ldots, n$; 
		\item for each $k=1, 2$ a decoding function  $g_k \colon \Sc_k^n \times \Yc_k^n \to\Wc_0\times \Wc_k$; 
		\item for each $k=1, 2$  a state estimator  $h_k \colon \Xc^n \times \Zc^n \to \hat{\Sc}_k^n$,  where  $\hat{\Sc}_1$ and $\hat{\Sc}_2$ are given  reconstruction alphabets.
	\end{enumerate}
	
	For a given code,  we let  the random messages $W_0$, $W_1$, and $W_2$ be uniform over the message sets $\Wc_0$, $\Wc_1$, and $\Wc_2$ and the inputs $X_i=\phi_i(W_0, W_1, W_2,  Z^{i-1})$,  for $i=1, \ldots,  n$. The corresponding outputs $Y_{1, i} Y_{2, i},  Z_i$ at time $i$ are obtained from the states $S_{1, i}$ and $S_{2, i}$  and the input $X_i$ according to the  SDMBC transition law $P_{Y_1Y_2Z|S_1S_2X}$. Further,  for $k=1, 2$ let $\hat{S}_k^n:= (\hat{S}_{k, 1}, \cdots, \hat{S}_{k, n} )=h_k(X^n,  Z^n)$ be the transmitter's estimates for state $S_k^n$ and  $(\hat{W}_{0,k},\hat{W}_k)=g_k(S_k^n, Y_k^n)$  the  messages decoded by Receiver~$k$.
	The quality of the state estimates $\hat{S}_k^n$  is again measured by  bounded per-symbol distortion functions $d_k\colon \Sc_k\times \hat{\Sc}_k \mapsto [0, \infty)$, i.e., we assume 
	\begin{equation}
		\max_{s_k \in \Sc_k, \hat{s}_k \in \hat{\Sc}_k} d_k(s_k,\hat{s}_k) < \infty, \quad k=1,2.
	\end{equation}
	
	Our  interest is in the two  \emph{expected average per-block distortions}
	\begin{equation}
		\Delta_k^{(n)}:= \frac{1}{n} \sum_{i=1}^n \mathbb{E}[d_k(S_{k, i},  \hat{S}_{k, i})],  \quad k=1, 2,
	\end{equation}
	and the joint probability of error
	\begin{IEEEeqnarray}{rCl}
		P^{(n)}_e& := &\textnormal{Pr}\Big((\hat{W}_{0,k}, \hat{W}_1)\neq (W_0,W_1) 
	\nonumber\\
	&&	\quad \textnormal{or} \quad (\hat{W}_{0,k}, \hat{W}_2)\neq (W_0,W_2) \Big).
	\end{IEEEeqnarray}
	\begin{definition} 
		A rate-distortion tuple $(\R_0,\R_1,  \R_2,  \D_1,  \D_2)$ is  achievable if there exists  a sequence (in $n$) of  $(2^{n\R_0},2^{n\R_1}, 2^{n\R_2},  n)$ codes that simultaneously satisfy
		\begin{subequations}\label{eq:asymptotics}
			\begin{IEEEeqnarray}{rCl}
				\lim_{n\to \infty}	P^{(n)}_e 
				&=&0 \\
				\varlimsup_{n\to \infty}	\Delta_k^{(n)}& \leq& \D_k,  \quad \textnormal{for } k=1, 2.\label{eq:asymptotics_dis}
			\end{IEEEeqnarray}
		\end{subequations}
	\end{definition}
	\begin{definition}
		The capacity-distortion region $\CDc$ is given by the closure of the union of all achievable rate-distortion tuples $(\R_0,\R_1,  \R_2, \D_1, \D_2)$.
	\end{definition}
	In the remainder of the section, we present bounds on the capacity-distortion region $\CDc$. As in the single-receiver case, one can easily determine the optimal estimator functions $h_1$ and $h_2$, which are independent of the encoding and decoding functions and operate on a symbol-by-symbol basis.

	\begin{lemma}\label{BC:lemma:Shat} 
		For  each $k=1, 2$, define the function
		\begin{IEEEeqnarray}{rCl}\label{eq:BCestimator}
			\hat{s}_{k}^*(x, z) &:= & {\rm arg}\min_{s'\in \hat{\Sc}_k} \sum_{s_k\in \Sc_k} P_{S_{k}|XZ}(s_k|x, z)  d(s_k,  s'), \IEEEeqnarraynumspace
		\end{IEEEeqnarray}
		where   ties can be broken arbitrarily.
		
		Irrespective of the choice of encoding and decoding functions, distortions $\Delta_1^{(n)}$ and $\Delta_2^{(n)}$ are minimized by the estimators for $k=1,2$
		\begin{IEEEeqnarray}{rCl}
			&&\hspace{-1cm}h_k^*(x^n,z^n) 
			\nonumber\\
			&=& (\hat{s}^*_k(x_1,z_1), \hat{s}^*_k(x_2,z_2),\ldots, \hat{s}^*_k(x_n,z_n)) .\IEEEeqnarraynumspace
		\end{IEEEeqnarray}
	\end{lemma}
	
	\begin{IEEEproof} 
		See	Appendix~\ref{app:lemma:Shat}.
	\end{IEEEproof}
	{Analogously to {the definition in Equation}~\eqref{eq:newcost} we can then define the optimal estimation cost for each input symbol $x\in\mathcal{X}$:	
		\begin{equation} 
			c_k(x) : = \E{d_k(S_k,\hat{s}_k^*(X,Z))|X=x}, \qquad k=1,2. 
	\end{equation}}
	 
	
	Characterizing the capacity-distortion  region is very challenging in general, because even the capacity regions of the SDMBC with and without feedback are unknown to date. We first present the exact capacity-distortion  region for the class of physically degraded SDMBCs and then provide bounds for general SDMBCs. 
	We shall also compare our results on the capacity-distortion  regions  to the performances achieved by simple  TS baseline schemes, in analogy to the single-receiver setup.

	Specifically, we again have a \emph{basic TS baseline scheme} that performs either sensing or communication at a time,  
	and an \emph{improved TS baseline scheme} that is able to perform both functions simultaneously via a common waveform by prioritizing either sensing or communication. {Analogously to the single-receiver setup, each of the two baseline schemes time-shares between a sensing mode and a communication mode. 
		However, since we now have two distortions and three rates, the choice of the ``optimal" pmf $P_X$ for each mode is not necessarily unique, but rather a continuum, depending on which function of the two distortions or the three rates one wishes to optimize. For fixed input pmf, 
		the difference between the communication mode \emph{without sensing} (employed by the basic TS scheme) and the communication mode  \emph{with sensing} (employed by the improved TS scheme) lies  in  the choice of the estimators. In the former mode, the transmitter applies the best \emph{constant estimators} for the two state-sequences, irrespective of  its inputs and feedback outputs.  In the latter mode, it applies the optimal estimators in Lemma~\ref{BC:lemma:Shat}, which depend on the input and the feedback output. Similarly, the difference between the communication modes \emph{without and with sensing} is that in the former all rates are zero and in the latter the chosen input pmf $P_X$ can be used for communication at positive rates.}

	\subsection{Capacity-Distortion Region for Physically Degraded SDMBCs}\label{sect:DegradedBC}
	This section  characterizes  the capacity-distortion  region  for \emph{physically degraded SDMBCs} and evaluates it for two  binary examples.

	\begin{definition}
		An SDMBC $P_{Y_1Y_2Z|S_1S_2X}$  with state pmf $P_{S_1S_2}$ is called \emph{physically degraded} if  there are conditional laws $P_{Y_	1|XS_1}$ and $P_{S_2Y_2|  S_1 Y_1}$ such that 
		\begin{equation}
			P_{Y_1Y_2|S_1S_2X} P_{S_1S_2} = P_{S_1}P_{Y_1|S_1X} P_{S_2Y_2|  S_1 Y_1}.
		\end{equation}
		That means for any arbitrary input $P_X$,  the tuple $(X, S_1, S_2, Y_1, Y_2)\sim P_{X} P_{S_1S_2} P_{Y_1Y_2|S_1S_2X}$ satisfies the Markov chain 
		\begin{equation}\label{eq:Mc}
			X \markov (S_1,  Y_1) \markov (S_2,  Y_2). 
		\end{equation}
	\end{definition}
	\begin{theorem}\label{Th:physically_BC}
		The  capacity-distortion region $\Cc\Dc$ of a physically degraded SDMBC is given by the closure of the set of all tuples 
		$(\R_0,\R_1, \R_2, \D_1, \D_2)$ for which there exists a joint law $P_{UX}$ so that the tuple $(U,  X,  S_1,  S_2,  Y_1,  Y_2,  Z)\sim P_{UX} P_{S_1S_2} P_{Y_1Y_2Z|S_1S_2X}$ satisfies  the two rate constraints 
		\begin{IEEEeqnarray}{rCl}
			\R_1&\leq& I(X;Y_1 \mid U, S_1)\label{R11}\\
			\R_0+\R_2&\leq& I(U;Y_2\mid S_2), \label{R22}
		\end{IEEEeqnarray}
		and the distortion constraints
		\begin{equation}\label{eq:dis:physicallyBC}
			\mathbb{E}[d_k(S_k,  \hat{s}_{k}^*(X, Z))]\leq \D_k,  \quad k=1, 2.
		\end{equation}	
	\end{theorem}
	\begin{IEEEproof} The achievability can be proved by standard superposition coding and using the optimal estimators in Lemma~\ref{BC:lemma:Shat}. The converse also follows from standard steps and the details are provided in Appendix \ref{app:converse_proof}.
	\end{IEEEproof}
	As mentioned in the proof, data communication is performed by simple superposition coding that ignores the feedback. Thus, also for physically degraded BCs feedback only  facilitates state sensing but is useless for  communications.

	\begin{remark}
	Similarly to the single-receiver case, an input cost-constraint as in \eqref{eq:cost_constraint} can be added to our model. Theorem~\ref{Th:physically_BC} remains valid in this case, if the choice of  the input distribution $P_X$ is limited  to satisfy the cost constraint
	\begin{equation}
	\sum_{x\in\mathcal{X}} P_X(x) b(x) \leq B.
	\end{equation} 
				The analogous remark also applies to the non-physically degraded BC ahead and the presented inner and outer bounds. 
	\end{remark}

		\begin{remark}
		Similarly to what we described in Remark~\ref{rem:CSI}, the result in Theorem~\ref{Th:physically_BC} can be extended to the case with imperfect receiver state-informations $S_{R,1}^n$ and $S_{R,2}^{n}$. For $(S^n, S_{R,1}^n, S_{R,2}^n)$ i.i.d. $\sim P_{SS_{R,1}, S_{R,2}}$ it suffices to replace in the rate-constraints \eqref{R11} and \eqref{R22} of Theorem~\ref{Th:physically_BC}  the state $S_1$ by $S_{R,1}$ and the state $S_2$ by $S_{R,2}$. The analogous remark also applies to the non-physically degraded BC ahead and the presented inner and outer bounds. 
%
		\end{remark}

	In what follows, we evaluate above Theorem~\ref{Th:physically_BC} for two examples.

	\subsubsection{Example 3: Binary BC with Multiplicative Bernoulli States}\label{ex:Binary1}
	Consider the physically degraded SDMBC with binary input and output alphabets $\Xc=\Yc_1=\Yc_2=\{0, 1\}$ and binary
	state alphabets $\Sc_1=\Sc_2=\{0, 1\}$. The channel input-output relation is described by 
	\begin{IEEEeqnarray}{rCl}\label{ex2}
		Y_k&=& S_k X,  \qquad k=1, 2,  
	\end{IEEEeqnarray}
	with 
	the joint state pmf 
	\begin{align}\label{ew:p_ss}
		P_{S_1 S_2} (s_1, s_2) &= \begin{cases}
			1-q,  & \text{if $(s_1,  s_2)= (0, 0)$} \\
			0,  & \text{if $(s_1,  s_2)= (0, 1)$}\\
			q \gamma,  & \text{if $(s_1,  s_2)= (1,  1)$} \\
			{q(1-\gamma)}& \text{if $(s_1,  s_2)= (1,  0)$}, 
		\end{cases}
	\end{align}
	for 
	$\gamma,  q\in[0, 1]$.
	Notice that $S_2$ is a degraded version of $S_1$, which together with the transition law \eqref{ex2} ensures the  Markov chain $X \markov (S_1,Y_1) \markov (S_2,Y_2)$ and the physically degradedness of the SDMBC. We consider output feedback 
	\begin{equation}
		\label{ex2:fb}
		Z=(Y_1,  Y_2),
	\end{equation} 
	and  set the common rate $\R_0=0$ for simplicity. 
	
	In this SDMBC, zero distortions $\D_1=\D_2=0$ can be achieved by deterministically choosing $X=1$ exactly as for the single-receiver case. This choice however cannot achieve any positive communication rates, i.e.,  $\R_1=\R_2=0$. In the sensing mode with and without communication, we thus have:
	\begin{equation}\label{eq:t1}
		(\R_1,  \R_2,  \D_1,  \D_2)= (0,  0,  0, 0).
	\end{equation}
	
	The optimal input distribution for communication is $X_{\max}\sim\mathcal{B}(1/2)$, in which case all rate-pairs $(\R_1,\R_2)$ satisfying 
	\begin{equation}
		\R_k \leq P_{S_k}(1) , \qquad k=1,2,
	\end{equation}
	are achievable. The input $X_{\max} \sim  \mathcal{B}(1/2)$ simultaneously maximizes both communication rates $\R_1,\R_2$. 
	
	In the communication mode \emph{without} sensing, the transmitter applies the optimal constant estimator for each state, namely 
	\begin{equation}
		\hat{s}_{{\rm const}, k}:=\argmax_{\hat{s}\in\{0,1\}}  P_{S_k}(\hat{s}), \qquad k=1,2,
	\end{equation}
	and thus achieves all tuples 
	\begin{equation}\label{eq:t2}
		(\R_1,  \R_2,  \D_1, \D_2)=(q r,  \gamma q(1-r),  \D_{1,  \max},  \D_{2,  \max})
	\end{equation} where  $\D_{1,  \max}:=\min\{q,  1-q\}$ and $\D_{2,  \max}:= \min\{\gamma q,  1-\gamma q\}$,  and $r\in [0, 1]$ denotes the time-sharing parameter between the two communication rates. 
	
	In the communication mode \emph{with} sensing, the same input  $X_{\max}$ is used. The transmitter however applies the optimal estimator for $k=1,2$:
	\begin{equation}\label{eq:estimator1one}
		\hat{s}_k^*(x,y_1,y_2)= \begin{cases} y_k, & \textnormal{if } x=1 \\
			\hat{s}_{{\rm const}, k}, &\textnormal{if }  x=0,
		\end{cases}
	\end{equation}
	and achieves the tuple
	\begin{equation}\label{eq:t3}
		(\R_1,  \R_2,  \D_1, \D_2)=\left(q r,  \gamma q (1-r),  \frac{\D_{1,  \max}}{2},  \frac{\D_{2,  \max}}{2}\right),
	\end{equation}
	where $r$ again denotes the time-sharing parameter between the two communication rates. 
	
	The basic and improved TS baseline schemes achieve the time-sharing lines between  points  \eqref{eq:t1} and \eqref{eq:t2} and  points \eqref{eq:t1} and \eqref{eq:t3}, respectively. The following corollary evaluates  Theorem~\ref{Th:physically_BC} to obtain the performance of the optimal co-design scheme. 
	
	%

	\begin{corollary}\label{corollary:binaryBC1}				
		The capacity-distortions region $\CDc$ of the binary physically degraded SDMBC in
		\eqref{ex2}--\eqref{ex2:fb} is  the set of all tuples $(\R_0,\R_1, \R_2, \D_1, \D_2)$ satisfying 
		\begin{subequations}\label{eq:con}
			\begin{IEEEeqnarray}{rCl}
				\R_0+\R_1&\leq&
				q H_\textnormal{b}(p) r, \label{R1f}\\
				\R_0+	\R_2&\leq &
				\gamma  q H_\textnormal{b}(p) (1-r), \label{R2f}\\
				\D_1 &\geq & p \min\{q,  1-q\}, \label{eq:conD1}\\
				\D_2 &\geq &p \min\{\gamma q,  1-\gamma  q \},  \label{eq:conD2}
			\end{IEEEeqnarray}
		\end{subequations}
		for some choice of the parameters $r,  p\in[0, 1]$. 
	\end{corollary}
	
	\begin{IEEEproof}
		We start by noticing that for this example $I(X;Y_1\mid U,S_1)= q H(X|U)$ and $I(U;Y_2 \mid S_2)=q \gamma (H(X)-H(X\mid U))$. Setting 
		$p:=P_X(0)$ 
		and $r:=\frac{H(X\mid U)}{H(X)}$, directly leads to the desired rate constraints. The distortion constraints are obtained from the optimal estimators in \eqref{eq:estimator1one}. Following the same steps as in the single-receiver case, i.e. \eqref{eq:SingleEstCost} and \eqref{eq:SingleD}, we obtain 
		\begin{align}
			\D_k \geq p \min\{P_{S_k}(0), P_{S_k}(1)\},
		\end{align}
		which concludes the proof.
	\end{IEEEproof}

	{Notice that   above Corollary~\ref{corollary:binaryBC1} reduces to Corollary \ref{cor:ex2} in the special case of $\R_0=\R_2=0$ and $\D_2=\infty$, i.e., when we ignore  Receiver 2.}
	
	\begin{figure}[!t]
		\vspace{-0.3cm}
		\centering
	%
	%
	\begin{tikzpicture}[scale=0.48]
		
	%
	%
%
		
		\begin{axis}[%
			width=6.028in,
			height=4.754in,
			at={(1.011in,0.642in)},
			scale only axis,
			unbounded coords=jump,
			xmin=0,
			xmax=0.6,
			tick align=outside,
			xlabel style={font=\color{white!15!black}},
			xlabel={\huge$\R_1$},
			ymin=0,
			ymax=0.35,
			ylabel style={font=\color{white!15!black}},
			ylabel={\huge$\D_1$},
			zmin=-0,
			zmax=0.35,
			zlabel style={font=\large},
			zlabel={\huge$\R_2$},
			view={43}{14},
			axis background/.style={fill=white},
			axis x line*=bottom,
			axis y line*=left,
			axis z line*=left,
			xmajorgrids,
			ymajorgrids,
			zmajorgrids,
			%
		%
		xticklabels={, ,, 0.2, , 0.4,  ,0.6},
		yticklabels={,0,  , , , 0.2, , ,0.35},
			zticklabels={,0, ,0.1, , 0.2, ,0.3, },
			x tick label style={font=\LARGE},
			y tick label style={font=\LARGE},
			z tick label style={font=\LARGE},
				%
			legend style={at={(0.7,0.8)}, anchor=north west, legend cell align=left, align=left, draw=white!15!black}
			]
			\addplot3 [color=green, dashed, line width=3.0pt]
			table[row sep=crcr] {%
				0	0	0\\
				0.6	0.2	0\\
			};\addlegendentry{\LARGE Improved TS}
				\addplot3 [color=blue, dashed, line width=2.0pt]
			table[row sep=crcr] {%
				0.6	0.2	0\\
				0	0.35	0.3\\
			};
		\addlegendentry{\LARGE Basic TS}
			\addplot3 [color=red, line width=2.0pt]
			table[row sep=crcr] {%
				0	0	0\\
				0	0.004	0.0242379407687734\\
				0	0.008	0.0424321627625462\\
				0	0.012	0.0583175573494729\\
				0	0.016	0.0726876567247244\\
				0	0.02	0.0859190871347869\\
				0	0.024	0.0982334757463429\\
				0	0.028	0.109777095270067\\
				0	0.032	0.120653757060682\\
				0	0.036	0.130940945119231\\
				0	0.04	0.140698678076784\\
				0	0.044	0.149974787449358\\
				0	0.048	0.158808259586209\\
				0	0.052	0.167231455508397\\
				0	0.056	0.175271643492857\\
				0	0.06	0.18295209141492\\
				0	0.064	0.19029286639217\\
				0	0.068	0.197311433623266\\
				0	0.072	0.204023113718484\\
				0	0.076	0.210441437965169\\
				0	0.08	0.216578428466209\\
				0	0.084	0.222444821979382\\
				0	0.088	0.22805025088859\\
				0	0.092	0.233403391063961\\
				0	0.096	0.238512083815357\\
				0	0.1	0.24338343733774\\
				0	0.104	0.248023911747785\\
				0	0.108	0.252439390862453\\
				0	0.112	0.256635243168039\\
				0	0.116	0.260616373901821\\
				0	0.12	0.264387269769208\\
				0	0.124	0.267952037513357\\
				0	0.128	0.271314437317348\\
				0	0.132	0.274477911833918\\
				0	0.136	0.277445611491909\\
				0	0.14	0.280220416612647\\
				0	0.144	0.282804956776648\\
				0	0.148	0.28520162780612\\
				0	0.152	0.28741260666789\\
				0	0.156	0.289439864551526\\
				0	0.16	0.291285178336401\\
				0	0.164	0.292950140627347\\
				0	0.168	0.294436168510096\\
				0	0.172	0.295744511153676\\
				0	0.176	0.296876256366617\\
				0	0.18	0.297832336196343\\
				0	0.184	0.298613531646068\\
				0	0.188	0.299220476570322\\
				0	0.192	0.299653660798561\\
				0	0.196	0.299913432525843\\
				0	0.2	0.3\\
				0	0.352	0.3\\
				0	0.356	0.3\\
			};
		\addlegendentry{\LARGE Co-design}
			\addplot3 [color=green, dashed, line width=3.0pt]
		table[row sep=crcr] {%
			0.6	0.2	0\\
			0	0.2	0.3\\
		};
		
		\addplot3 [color=green, dashed, line width=3.0pt]
		table[row sep=crcr] {%
			0	0.2	0.3\\
			0	0	0\\
		};

		\addplot3 [color=blue, dashed, line width=2.0pt]
		table[row sep=crcr] {%
			0	0	0\\
			0	0.35	0.3\\
		};
		
		\addplot3 [color=blue, dashed, line width=2.0pt]
		table[row sep=crcr] {%
			0.6	0.2	0\\
			0	0	0\\
		};
			\addplot3 [color=red, line width=2.0pt]
			table[row sep=crcr] {%
				0.03	0.002	0\\
				0.03	0.004	0.00923794076877335\\
				0.03	0.008	0.0274321627625462\\
				0.03	0.012	0.0433175573494729\\
				0.03	0.016	0.0576876567247244\\
				0.03	0.02	0.0709190871347869\\
				0.03	0.024	0.0832334757463429\\
				0.03	0.028	0.094777095270067\\
				0.03	0.032	0.105653757060682\\
				0.03	0.036	0.115940945119231\\
				0.03	0.04	0.125698678076784\\
				0.03	0.044	0.134974787449358\\
				0.03	0.048	0.143808259586209\\
				0.03	0.052	0.152231455508397\\
				0.03	0.056	0.160271643492857\\
				0.03	0.06	0.16795209141492\\
				0.03	0.064	0.17529286639217\\
				0.03	0.068	0.182311433623266\\
				0.03	0.072	0.189023113718484\\
				0.03	0.076	0.195441437965169\\
				0.03	0.08	0.201578428466209\\
				0.03	0.084	0.207444821979382\\
				0.03	0.088	0.21305025088859\\
				0.03	0.092	0.218403391063961\\
				0.03	0.096	0.223512083815357\\
				0.03	0.1	0.22838343733774\\
				0.03	0.104	0.233023911747785\\
				0.03	0.108	0.237439390862453\\
				0.03	0.112	0.241635243168039\\
				0.03	0.116	0.245616373901821\\
				0.03	0.12	0.249387269769208\\
				0.03	0.124	0.252952037513357\\
				0.03	0.128	0.256314437317348\\
				0.03	0.132	0.259477911833918\\
				0.03	0.136	0.262445611491909\\
				0.03	0.14	0.265220416612647\\
				0.03	0.144	0.267804956776648\\
				0.03	0.148	0.27020162780612\\
				0.03	0.152	0.27241260666789\\
				0.03	0.156	0.274439864551526\\
				0.03	0.16	0.276285178336401\\
				0.03	0.164	0.277950140627347\\
				0.03	0.168	0.279436168510096\\
				0.03	0.172	0.280744511153676\\
				0.03	0.176	0.281876256366617\\
				0.03	0.18	0.282832336196343\\
				0.03	0.184	0.283613531646068\\
				0.03	0.188	0.284220476570322\\
				0.03	0.192	0.284653660798561\\
				0.03	0.196	0.284913432525843\\
				0.03	0.2	0.285\\
				0.03	0.352	0.285\\
				0.03	0.356	0.285\\
			};
			\addplot3 [color=red, line width=2.0pt]
			table[row sep=crcr] {%
				0.06	0	0\\
				0.06	0.004	-0.0055\\
				0.06	0.008	0.0124321627625462\\
				0.06	0.012	0.0283175573494729\\
				0.06	0.016	0.0426876567247244\\
				0.06	0.02	0.0559190871347869\\
				0.06	0.024	0.0682334757463429\\
				0.06	0.028	0.079777095270067\\
				0.06	0.032	0.0906537570606818\\
				0.06	0.036	0.100940945119231\\
				0.06	0.04	0.110698678076784\\
				0.06	0.044	0.119974787449358\\
				0.06	0.048	0.128808259586209\\
				0.06	0.052	0.137231455508397\\
				0.06	0.056	0.145271643492857\\
				0.06	0.06	0.15295209141492\\
				0.06	0.064	0.16029286639217\\
				0.06	0.068	0.167311433623266\\
				0.06	0.072	0.174023113718484\\
				0.06	0.076	0.180441437965169\\
				0.06	0.08	0.186578428466209\\
				0.06	0.084	0.192444821979382\\
				0.06	0.088	0.19805025088859\\
				0.06	0.092	0.203403391063961\\
				0.06	0.096	0.208512083815357\\
				0.06	0.1	0.21338343733774\\
				0.06	0.104	0.218023911747785\\
				0.06	0.108	0.222439390862453\\
				0.06	0.112	0.226635243168039\\
				0.06	0.116	0.230616373901821\\
				0.06	0.12	0.234387269769208\\
				0.06	0.124	0.237952037513357\\
				0.06	0.128	0.241314437317348\\
				0.06	0.132	0.244477911833918\\
				0.06	0.136	0.247445611491909\\
				0.06	0.14	0.250220416612647\\
				0.06	0.144	0.252804956776648\\
				0.06	0.148	0.25520162780612\\
				0.06	0.152	0.25741260666789\\
				0.06	0.156	0.259439864551526\\
				0.06	0.16	0.261285178336401\\
				0.06	0.164	0.262950140627347\\
				0.06	0.168	0.264436168510096\\
				0.06	0.172	0.265744511153676\\
				0.06	0.176	0.266876256366617\\
				0.06	0.18	0.267832336196343\\
				0.06	0.184	0.268613531646068\\
				0.06	0.188	0.269220476570322\\
				0.06	0.192	0.269653660798561\\
				0.06	0.196	0.269913432525843\\
				0.06	0.2	0.27\\
				0.06	0.352	0.27\\
				0.06	0.356	0.27\\
			};
			\addplot3 [color=red, line width=2.0pt]
			table[row sep=crcr] {%
				0.09	0.008	-0.00\\
				0.09	0.012	0.0133175573494729\\
				0.09	0.016	0.0276876567247244\\
				0.09	0.02	0.0409190871347869\\
				0.09	0.024	0.0532334757463429\\
				0.09	0.028	0.064777095270067\\
				0.09	0.032	0.0756537570606818\\
				0.09	0.036	0.0859409451192309\\
				0.09	0.04	0.0956986780767844\\
				0.09	0.044	0.104974787449358\\
				0.09	0.048	0.113808259586209\\
				0.09	0.052	0.122231455508397\\
				0.09	0.056	0.130271643492857\\
				0.09	0.06	0.13795209141492\\
				0.09	0.064	0.14529286639217\\
				0.09	0.068	0.152311433623266\\
				0.09	0.072	0.159023113718484\\
				0.09	0.076	0.165441437965169\\
				0.09	0.08	0.171578428466209\\
				0.09	0.084	0.177444821979382\\
				0.09	0.088	0.18305025088859\\
				0.09	0.092	0.188403391063961\\
				0.09	0.096	0.193512083815357\\
				0.09	0.1	0.19838343733774\\
				0.09	0.104	0.203023911747785\\
				0.09	0.108	0.207439390862453\\
				0.09	0.112	0.211635243168039\\
				0.09	0.116	0.215616373901821\\
				0.09	0.12	0.219387269769208\\
				0.09	0.124	0.222952037513357\\
				0.09	0.128	0.226314437317348\\
				0.09	0.132	0.229477911833918\\
				0.09	0.136	0.232445611491909\\
				0.09	0.14	0.235220416612647\\
				0.09	0.144	0.237804956776648\\
				0.09	0.148	0.24020162780612\\
				0.09	0.152	0.24241260666789\\
				0.09	0.156	0.244439864551526\\
				0.09	0.16	0.246285178336401\\
				0.09	0.164	0.247950140627347\\
				0.09	0.168	0.249436168510096\\
				0.09	0.172	0.250744511153676\\
				0.09	0.176	0.251876256366617\\
				0.09	0.18	0.252832336196342\\
				0.09	0.184	0.253613531646068\\
				0.09	0.188	0.254220476570322\\
				0.09	0.192	0.25465366079856\\
				0.09	0.196	0.254913432525843\\
				0.09	0.2	0.255\\
				0.09	0.352	0.255\\
				0.09	0.356	0.255\\
			};
			\addplot3 [color=red, line width=2.0pt]
			table[row sep=crcr] {%
				0.12	0.012	-0.00\\
				0.12	0.016	0.0126876567247244\\
				0.12	0.02	0.0259190871347869\\
				0.12	0.024	0.0382334757463429\\
				0.12	0.028	0.049777095270067\\
				0.12	0.032	0.0606537570606818\\
				0.12	0.036	0.0709409451192309\\
				0.12	0.04	0.0806986780767844\\
				0.12	0.044	0.0899747874493584\\
				0.12	0.048	0.0988082595862093\\
				0.12	0.052	0.107231455508397\\
				0.12	0.056	0.115271643492857\\
				0.12	0.06	0.12295209141492\\
				0.12	0.064	0.13029286639217\\
				0.12	0.068	0.137311433623266\\
				0.12	0.072	0.144023113718484\\
				0.12	0.076	0.150441437965169\\
				0.12	0.08	0.156578428466209\\
				0.12	0.084	0.162444821979382\\
				0.12	0.088	0.16805025088859\\
				0.12	0.092	0.173403391063961\\
				0.12	0.096	0.178512083815357\\
				0.12	0.1	0.18338343733774\\
				0.12	0.104	0.188023911747785\\
				0.12	0.108	0.192439390862453\\
				0.12	0.112	0.196635243168039\\
				0.12	0.116	0.200616373901821\\
				0.12	0.12	0.204387269769208\\
				0.12	0.124	0.207952037513357\\
				0.12	0.128	0.211314437317348\\
				0.12	0.132	0.214477911833918\\
				0.12	0.136	0.217445611491909\\
				0.12	0.14	0.220220416612647\\
				0.12	0.144	0.222804956776648\\
				0.12	0.148	0.22520162780612\\
				0.12	0.152	0.22741260666789\\
				0.12	0.156	0.229439864551526\\
				0.12	0.16	0.231285178336401\\
				0.12	0.164	0.232950140627347\\
				0.12	0.168	0.234436168510096\\
				0.12	0.172	0.235744511153676\\
				0.12	0.176	0.236876256366617\\
				0.12	0.18	0.237832336196342\\
				0.12	0.184	0.238613531646068\\
				0.12	0.188	0.239220476570322\\
				0.12	0.192	0.239653660798561\\
				0.12	0.196	0.239913432525843\\
				0.12	0.2	0.24\\
				0.12	0.352	0.24\\
				0.12	0.356	0.24\\
			};
			\addplot3 [color=red, line width=2.0pt]
			table[row sep=crcr] {%
				0.15	0.016	-0.00\\
				0.15	0.02	0.0109190871347869\\
				0.15	0.024	0.0232334757463429\\
				0.15	0.028	0.034777095270067\\
				0.15	0.032	0.0456537570606818\\
				0.15	0.036	0.0559409451192309\\
				0.15	0.04	0.0656986780767844\\
				0.15	0.044	0.0749747874493584\\
				0.15	0.048	0.0838082595862093\\
				0.15	0.052	0.0922314555083967\\
				0.15	0.056	0.100271643492857\\
				0.15	0.06	0.10795209141492\\
				0.15	0.064	0.11529286639217\\
				0.15	0.068	0.122311433623266\\
				0.15	0.072	0.129023113718484\\
				0.15	0.076	0.135441437965169\\
				0.15	0.08	0.141578428466209\\
				0.15	0.084	0.147444821979382\\
				0.15	0.088	0.15305025088859\\
				0.15	0.092	0.158403391063961\\
				0.15	0.096	0.163512083815357\\
				0.15	0.1	0.16838343733774\\
				0.15	0.104	0.173023911747785\\
				0.15	0.108	0.177439390862453\\
				0.15	0.112	0.181635243168039\\
				0.15	0.116	0.185616373901821\\
				0.15	0.12	0.189387269769208\\
				0.15	0.124	0.192952037513357\\
				0.15	0.128	0.196314437317348\\
				0.15	0.132	0.199477911833918\\
				0.15	0.136	0.202445611491909\\
				0.15	0.14	0.205220416612647\\
				0.15	0.144	0.207804956776648\\
				0.15	0.148	0.21020162780612\\
				0.15	0.152	0.21241260666789\\
				0.15	0.156	0.214439864551526\\
				0.15	0.16	0.216285178336401\\
				0.15	0.164	0.217950140627347\\
				0.15	0.168	0.219436168510096\\
				0.15	0.172	0.220744511153676\\
				0.15	0.176	0.221876256366617\\
				0.15	0.18	0.222832336196343\\
				0.15	0.184	0.223613531646068\\
				0.15	0.188	0.224220476570322\\
				0.15	0.192	0.224653660798561\\
				0.15	0.196	0.224913432525843\\
				0.15	0.2	0.225\\
				0.15	0.352	0.225\\
				0.15	0.356	0.225\\
			};
			\addplot3 [color=red, line width=2.0pt]
			table[row sep=crcr] {%
				0.18	0.021	0.00\\
				0.18	0.024	0.00823347574634288\\
				0.18	0.028	0.019777095270067\\
				0.18	0.032	0.0306537570606818\\
				0.18	0.036	0.0409409451192309\\
				0.18	0.04	0.0506986780767844\\
				0.18	0.044	0.0599747874493584\\
				0.18	0.048	0.0688082595862093\\
				0.18	0.052	0.0772314555083967\\
				0.18	0.056	0.0852716434928568\\
				0.18	0.06	0.0929520914149201\\
				0.18	0.064	0.10029286639217\\
				0.18	0.068	0.107311433623266\\
				0.18	0.072	0.114023113718484\\
				0.18	0.076	0.120441437965169\\
				0.18	0.08	0.126578428466209\\
				0.18	0.084	0.132444821979382\\
				0.18	0.088	0.13805025088859\\
				0.18	0.092	0.143403391063961\\
				0.18	0.096	0.148512083815357\\
				0.18	0.1	0.15338343733774\\
				0.18	0.104	0.158023911747785\\
				0.18	0.108	0.162439390862453\\
				0.18	0.112	0.166635243168039\\
				0.18	0.116	0.170616373901821\\
				0.18	0.12	0.174387269769208\\
				0.18	0.124	0.177952037513357\\
				0.18	0.128	0.181314437317348\\
				0.18	0.132	0.184477911833918\\
				0.18	0.136	0.187445611491909\\
				0.18	0.14	0.190220416612647\\
				0.18	0.144	0.192804956776648\\
				0.18	0.148	0.19520162780612\\
				0.18	0.152	0.19741260666789\\
				0.18	0.156	0.199439864551526\\
				0.18	0.16	0.201285178336401\\
				0.18	0.164	0.202950140627347\\
				0.18	0.168	0.204436168510096\\
				0.18	0.172	0.205744511153676\\
				0.18	0.176	0.206876256366617\\
				0.18	0.18	0.207832336196343\\
				0.18	0.184	0.208613531646068\\
				0.18	0.188	0.209220476570322\\
				0.18	0.192	0.209653660798561\\
				0.18	0.196	0.209913432525843\\
				0.18	0.2	0.21\\
				0.18	0.352	0.21\\
				0.18	0.356	0.21\\
			};
			\addplot3 [color=red, line width=2.0pt]
			table[row sep=crcr] {%
				0.21	0.0265	-0.002\\
				0.21	0.028	0.00477709527006701\\
				0.21	0.032	0.0156537570606818\\
				0.21	0.036	0.0259409451192309\\
				0.21	0.04	0.0356986780767844\\
				0.21	0.044	0.0449747874493584\\
				0.21	0.048	0.0538082595862093\\
				0.21	0.052	0.0622314555083967\\
				0.21	0.056	0.0702716434928568\\
				0.21	0.06	0.0779520914149201\\
				0.21	0.064	0.0852928663921698\\
				0.21	0.068	0.0923114336232659\\
				0.21	0.072	0.0990231137184839\\
				0.21	0.076	0.105441437965169\\
				0.21	0.08	0.111578428466209\\
				0.21	0.084	0.117444821979382\\
				0.21	0.088	0.12305025088859\\
				0.21	0.092	0.128403391063961\\
				0.21	0.096	0.133512083815357\\
				0.21	0.1	0.13838343733774\\
				0.21	0.104	0.143023911747785\\
				0.21	0.108	0.147439390862453\\
				0.21	0.112	0.151635243168039\\
				0.21	0.116	0.155616373901821\\
				0.21	0.12	0.159387269769208\\
				0.21	0.124	0.162952037513357\\
				0.21	0.128	0.166314437317348\\
				0.21	0.132	0.169477911833918\\
				0.21	0.136	0.172445611491909\\
				0.21	0.14	0.175220416612647\\
				0.21	0.144	0.177804956776648\\
				0.21	0.148	0.18020162780612\\
				0.21	0.152	0.18241260666789\\
				0.21	0.156	0.184439864551526\\
				0.21	0.16	0.186285178336401\\
				0.21	0.164	0.187950140627347\\
				0.21	0.168	0.189436168510096\\
				0.21	0.172	0.190744511153676\\
				0.21	0.176	0.191876256366617\\
				0.21	0.18	0.192832336196343\\
				0.21	0.184	0.193613531646068\\
				0.21	0.188	0.194220476570322\\
				0.21	0.192	0.194653660798561\\
				0.21	0.196	0.194913432525843\\
				0.21	0.2	0.195\\
				0.21	0.352	0.195\\
				0.21	0.356	0.195\\
			};
			\addplot3 [color=red, line width=2.0pt]
			table[row sep=crcr] {%
				0.24	0.031	0.0005\\
				0.24	0.032	0.000653757060681825\\
				0.24	0.036	0.0109409451192309\\
				0.24	0.04	0.0206986780767844\\
				0.24	0.044	0.0299747874493584\\
				0.24	0.048	0.0388082595862093\\
				0.24	0.052	0.0472314555083967\\
				0.24	0.056	0.0552716434928568\\
				0.24	0.06	0.0629520914149201\\
				0.24	0.064	0.0702928663921698\\
				0.24	0.068	0.0773114336232658\\
				0.24	0.072	0.0840231137184839\\
				0.24	0.076	0.0904414379651692\\
				0.24	0.08	0.0965784284662087\\
				0.24	0.084	0.102444821979382\\
				0.24	0.088	0.10805025088859\\
				0.24	0.092	0.113403391063961\\
				0.24	0.096	0.118512083815357\\
				0.24	0.1	0.12338343733774\\
				0.24	0.104	0.128023911747785\\
				0.24	0.108	0.132439390862453\\
				0.24	0.112	0.136635243168039\\
				0.24	0.116	0.140616373901821\\
				0.24	0.12	0.144387269769208\\
				0.24	0.124	0.147952037513357\\
				0.24	0.128	0.151314437317348\\
				0.24	0.132	0.154477911833918\\
				0.24	0.136	0.157445611491909\\
				0.24	0.14	0.160220416612647\\
				0.24	0.144	0.162804956776648\\
				0.24	0.148	0.16520162780612\\
				0.24	0.152	0.16741260666789\\
				0.24	0.156	0.169439864551526\\
				0.24	0.16	0.171285178336401\\
				0.24	0.164	0.172950140627347\\
				0.24	0.168	0.174436168510096\\
				0.24	0.172	0.175744511153676\\
				0.24	0.176	0.176876256366617\\
				0.24	0.18	0.177832336196343\\
				0.24	0.184	0.178613531646068\\
				0.24	0.188	0.179220476570322\\
				0.24	0.192	0.179653660798561\\
				0.24	0.196	0.179913432525843\\
				0.24	0.2	0.18\\
				0.24	0.352	0.18\\
				0.24	0.356	0.18\\
			};
			\addplot3 [color=red, line width=2.0pt]
			table[row sep=crcr] {%
				0.27	0.039	0.00155905488076916\\
				0.27	0.04	0.00569867807678435\\
				0.27	0.044	0.0149747874493584\\
				0.27	0.048	0.0238082595862093\\
				0.27	0.052	0.0322314555083967\\
				0.27	0.056	0.0402716434928568\\
				0.27	0.06	0.0479520914149201\\
				0.27	0.064	0.0552928663921698\\
				0.27	0.068	0.0623114336232658\\
				0.27	0.072	0.0690231137184839\\
				0.27	0.076	0.0754414379651692\\
				0.27	0.08	0.0815784284662087\\
				0.27	0.084	0.0874448219793821\\
				0.27	0.088	0.0930502508885897\\
				0.27	0.092	0.0984033910639613\\
				0.27	0.096	0.103512083815357\\
				0.27	0.1	0.10838343733774\\
				0.27	0.104	0.113023911747785\\
				0.27	0.108	0.117439390862453\\
				0.27	0.112	0.121635243168039\\
				0.27	0.116	0.125616373901821\\
				0.27	0.12	0.129387269769208\\
				0.27	0.124	0.132952037513357\\
				0.27	0.128	0.136314437317348\\
				0.27	0.132	0.139477911833918\\
				0.27	0.136	0.142445611491909\\
				0.27	0.14	0.145220416612647\\
				0.27	0.144	0.147804956776648\\
				0.27	0.148	0.15020162780612\\
				0.27	0.152	0.15241260666789\\
				0.27	0.156	0.154439864551526\\
				0.27	0.16	0.156285178336401\\
				0.27	0.164	0.157950140627347\\
				0.27	0.168	0.159436168510096\\
				0.27	0.172	0.160744511153676\\
				0.27	0.176	0.161876256366617\\
				0.27	0.18	0.162832336196342\\
				0.27	0.184	0.163613531646068\\
				0.27	0.188	0.164220476570322\\
				0.27	0.192	0.164653660798561\\
				0.27	0.196	0.164913432525843\\
				0.27	0.2	0.165\\
				0.27	0.352	0.165\\
				0.27	0.356	0.165\\
			};
			\addplot3 [color=red, line width=2.0pt]
			table[row sep=crcr] {%
				0.3	0.046	0.00300132192321563\\
				0.3	0.048	0.00880825958620931\\
				0.3	0.052	0.0172314555083967\\
				0.3	0.056	0.0252716434928568\\
				0.3	0.06	0.0329520914149201\\
				0.3	0.064	0.0402928663921698\\
				0.3	0.068	0.0473114336232658\\
				0.3	0.072	0.0540231137184839\\
				0.3	0.076	0.0604414379651692\\
				0.3	0.08	0.0665784284662087\\
				0.3	0.084	0.0724448219793821\\
				0.3	0.088	0.0780502508885897\\
				0.3	0.092	0.0834033910639613\\
				0.3	0.096	0.0885120838153567\\
				0.3	0.1	0.0933834373377398\\
				0.3	0.104	0.0980239117477854\\
				0.3	0.108	0.102439390862453\\
				0.3	0.112	0.106635243168039\\
				0.3	0.116	0.110616373901821\\
				0.3	0.12	0.114387269769208\\
				0.3	0.124	0.117952037513357\\
				0.3	0.128	0.121314437317348\\
				0.3	0.132	0.124477911833918\\
				0.3	0.136	0.127445611491909\\
				0.3	0.14	0.130220416612647\\
				0.3	0.144	0.132804956776648\\
				0.3	0.148	0.13520162780612\\
				0.3	0.152	0.13741260666789\\
				0.3	0.156	0.139439864551526\\
				0.3	0.16	0.141285178336401\\
				0.3	0.164	0.142950140627347\\
				0.3	0.168	0.144436168510096\\
				0.3	0.172	0.145744511153676\\
				0.3	0.176	0.146876256366617\\
				0.3	0.18	0.147832336196343\\
				0.3	0.184	0.148613531646068\\
				0.3	0.188	0.149220476570322\\
				0.3	0.192	0.149653660798561\\
				0.3	0.196	0.149913432525843\\
				0.3	0.2	0.15\\
				0.3	0.352	0.15\\
				0.3	0.356	0.15\\
			};
			\addplot3 [color=red, line width=2.0pt]
			table[row sep=crcr] {%
				0.33	0.0515	0.00109174041379067\\
				0.33	0.052	0.00223145550839676\\
				0.33	0.056	0.0102716434928568\\
				0.33	0.06	0.0179520914149201\\
				0.33	0.064	0.0252928663921699\\
				0.33	0.068	0.0323114336232659\\
				0.33	0.072	0.0390231137184839\\
				0.33	0.076	0.0454414379651692\\
				0.33	0.08	0.0515784284662087\\
				0.33	0.084	0.0574448219793821\\
				0.33	0.088	0.0630502508885897\\
				0.33	0.092	0.0684033910639613\\
				0.33	0.096	0.0735120838153567\\
				0.33	0.1	0.0783834373377399\\
				0.33	0.104	0.0830239117477854\\
				0.33	0.108	0.0874393908624527\\
				0.33	0.112	0.0916352431680393\\
				0.33	0.116	0.0956163739018214\\
				0.33	0.12	0.0993872697692078\\
				0.33	0.124	0.102952037513357\\
				0.33	0.128	0.106314437317348\\
				0.33	0.132	0.109477911833918\\
				0.33	0.136	0.112445611491909\\
				0.33	0.14	0.115220416612647\\
				0.33	0.144	0.117804956776648\\
				0.33	0.148	0.12020162780612\\
				0.33	0.152	0.12241260666789\\
				0.33	0.156	0.124439864551526\\
				0.33	0.16	0.126285178336401\\
				0.33	0.164	0.127950140627347\\
				0.33	0.168	0.129436168510096\\
				0.33	0.172	0.130744511153676\\
				0.33	0.176	0.131876256366617\\
				0.33	0.18	0.132832336196343\\
				0.33	0.184	0.133613531646068\\
				0.33	0.188	0.134220476570322\\
				0.33	0.192	0.134653660798561\\
				0.33	0.196	0.134913432525843\\
				0.33	0.2	0.135\\
				0.33	0.352	0.135\\
				0.33	0.356	0.135\\
			};
			\addplot3 [color=red, line width=2.0pt]
			table[row sep=crcr] {%
				0.36	0.059	0.0012835650714321\\
				0.36	0.06	0.00295209141492012\\
				0.36	0.064	0.0102928663921699\\
				0.36	0.068	0.0173114336232659\\
				0.36	0.072	0.0240231137184839\\
				0.36	0.076	0.0304414379651692\\
				0.36	0.08	0.0365784284662087\\
				0.36	0.084	0.0424448219793821\\
				0.36	0.088	0.0480502508885897\\
				0.36	0.092	0.0534033910639613\\
				0.36	0.096	0.0585120838153567\\
				0.36	0.1	0.0633834373377399\\
				0.36	0.104	0.0680239117477854\\
				0.36	0.108	0.0724393908624527\\
				0.36	0.112	0.0766352431680393\\
				0.36	0.116	0.0806163739018214\\
				0.36	0.12	0.0843872697692078\\
				0.36	0.124	0.087952037513357\\
				0.36	0.128	0.0913144373173482\\
				0.36	0.132	0.0944779118339182\\
				0.36	0.136	0.097445611491909\\
				0.36	0.14	0.100220416612647\\
				0.36	0.144	0.102804956776648\\
				0.36	0.148	0.10520162780612\\
				0.36	0.152	0.10741260666789\\
				0.36	0.156	0.109439864551526\\
				0.36	0.16	0.111285178336401\\
				0.36	0.164	0.112950140627347\\
				0.36	0.168	0.114436168510096\\
				0.36	0.172	0.115744511153676\\
				0.36	0.176	0.116876256366617\\
				0.36	0.18	0.117832336196343\\
				0.36	0.184	0.118613531646068\\
				0.36	0.188	0.119220476570322\\
				0.36	0.192	0.119653660798561\\
				0.36	0.196	0.119913432525843\\
				0.36	0.2	0.12\\
				0.36	0.352	0.12\\
				0.36	0.356	0.12\\
			};
			\addplot3 [color=red, line width=2.0pt]
			table[row sep=crcr] {%
				0.39	0.067	0.00170713360783016\\
				0.39	0.068	0.00231143362326585\\
				0.39	0.072	0.00902311371848392\\
				0.39	0.076	0.0154414379651692\\
				0.39	0.08	0.0215784284662087\\
				0.39	0.084	0.0274448219793821\\
				0.39	0.088	0.0330502508885897\\
				0.39	0.092	0.0384033910639613\\
				0.39	0.096	0.0435120838153567\\
				0.39	0.1	0.0483834373377398\\
				0.39	0.104	0.0530239117477853\\
				0.39	0.108	0.0574393908624527\\
				0.39	0.112	0.0616352431680392\\
				0.39	0.116	0.0656163739018214\\
				0.39	0.12	0.0693872697692078\\
				0.39	0.124	0.072952037513357\\
				0.39	0.128	0.0763144373173482\\
				0.39	0.132	0.0794779118339182\\
				0.39	0.136	0.082445611491909\\
				0.39	0.14	0.0852204166126473\\
				0.39	0.144	0.0878049567766477\\
				0.39	0.148	0.0902016278061197\\
				0.39	0.152	0.0924126066678898\\
				0.39	0.156	0.0944398645515262\\
				0.39	0.16	0.0962851783364006\\
				0.39	0.164	0.0979501406273472\\
				0.39	0.168	0.0994361685100961\\
				0.39	0.172	0.100744511153676\\
				0.39	0.176	0.101876256366617\\
				0.39	0.18	0.102832336196343\\
				0.39	0.184	0.103613531646068\\
				0.39	0.188	0.104220476570322\\
				0.39	0.192	0.104653660798561\\
				0.39	0.196	0.104913432525843\\
				0.39	0.2	0.105\\
				0.39	0.352	0.105\\
				0.39	0.356	0.105\\
			};
			\addplot3 [color=red, line width=2.0pt]
			table[row sep=crcr] {%
				0.42	0.075	0.00027688628151606\\
				0.42	0.076	0.000441437965169222\\
				0.42	0.08	0.00657842846620871\\
				0.42	0.084	0.0124448219793821\\
				0.42	0.088	0.0180502508885897\\
				0.42	0.092	0.0234033910639613\\
				0.42	0.096	0.0285120838153567\\
				0.42	0.1	0.0333834373377399\\
				0.42	0.104	0.0380239117477854\\
				0.42	0.108	0.0424393908624527\\
				0.42	0.112	0.0466352431680393\\
				0.42	0.116	0.0506163739018214\\
				0.42	0.12	0.0543872697692078\\
				0.42	0.124	0.057952037513357\\
				0.42	0.128	0.0613144373173482\\
				0.42	0.132	0.0644779118339182\\
				0.42	0.136	0.067445611491909\\
				0.42	0.14	0.0702204166126473\\
				0.42	0.144	0.0728049567766477\\
				0.42	0.148	0.0752016278061198\\
				0.42	0.152	0.0774126066678899\\
				0.42	0.156	0.0794398645515262\\
				0.42	0.16	0.0812851783364006\\
				0.42	0.164	0.0829501406273472\\
				0.42	0.168	0.0844361685100961\\
				0.42	0.172	0.085744511153676\\
				0.42	0.176	0.0868762563666167\\
				0.42	0.18	0.0878323361963425\\
				0.42	0.184	0.0886135316460678\\
				0.42	0.188	0.0892204765703219\\
				0.42	0.192	0.0896536607985605\\
				0.42	0.196	0.089913432525843\\
				0.42	0.2	0.09\\
				0.42	0.352	0.09\\
				0.42	0.356	0.09\\
			};
			\addplot3 [color=red, line width=2.0pt]
			table[row sep=crcr] {%
				0.45	0.084	-0.0025551780206179\\
				0.45	0.088	0.0030502508885897\\
				0.45	0.092	0.00840339106396131\\
				0.45	0.096	0.0135120838153567\\
				0.45	0.1	0.0183834373377398\\
				0.45	0.104	0.0230239117477854\\
				0.45	0.108	0.0274393908624527\\
				0.45	0.112	0.0316352431680392\\
				0.45	0.116	0.0356163739018214\\
				0.45	0.12	0.0393872697692078\\
				0.45	0.124	0.042952037513357\\
				0.45	0.128	0.0463144373173482\\
				0.45	0.132	0.0494779118339182\\
				0.45	0.136	0.052445611491909\\
				0.45	0.14	0.0552204166126473\\
				0.45	0.144	0.0578049567766477\\
				0.45	0.148	0.0602016278061198\\
				0.45	0.152	0.0624126066678899\\
				0.45	0.156	0.0644398645515262\\
				0.45	0.16	0.0662851783364006\\
				0.45	0.164	0.0679501406273472\\
				0.45	0.168	0.0694361685100961\\
				0.45	0.172	0.0707445111536759\\
				0.45	0.176	0.0718762563666167\\
				0.45	0.18	0.0728323361963425\\
				0.45	0.184	0.0736135316460677\\
				0.45	0.188	0.0742204765703219\\
				0.45	0.192	0.0746536607985605\\
				0.45	0.196	0.074913432525843\\
				0.45	0.2	0.075\\
				0.45	0.352	0.075\\
				0.45	0.356	0.075\\
			};
			\addplot3 [color=red, line width=2.0pt]
			table[row sep=crcr] {%
				%
				0.48	0.096	-0.00148791618464333\\
				0.48	0.1	0.00338343733773984\\
				0.48	0.104	0.00802391174778534\\
				0.48	0.108	0.0124393908624527\\
				0.48	0.112	0.0166352431680392\\
				0.48	0.116	0.0206163739018214\\
				0.48	0.12	0.0243872697692078\\
				0.48	0.124	0.027952037513357\\
				0.48	0.128	0.0313144373173482\\
				0.48	0.132	0.0344779118339182\\
				0.48	0.136	0.037445611491909\\
				0.48	0.14	0.0402204166126473\\
				0.48	0.144	0.0428049567766476\\
				0.48	0.148	0.0452016278061197\\
				0.48	0.152	0.0474126066678898\\
				0.48	0.156	0.0494398645515261\\
				0.48	0.16	0.0512851783364006\\
				0.48	0.164	0.0529501406273472\\
				0.48	0.168	0.054436168510096\\
				0.48	0.172	0.0557445111536759\\
				0.48	0.176	0.0568762563666167\\
				0.48	0.18	0.0578323361963425\\
				0.48	0.184	0.0586135316460677\\
				0.48	0.188	0.0592204765703219\\
				0.48	0.192	0.0596536607985605\\
				0.48	0.196	0.059913432525843\\
				0.48	0.2	0.06\\
				0.48	0.352	0.06\\
				0.48	0.356	0.06\\
			};
			\addplot3 [color=red, line width=2.0pt]
			table[row sep=crcr] {%
			%
				0.51	0.108	-0.00256060913754732\\
				0.51	0.112	0.00163524316803922\\
				0.51	0.116	0.00561637390182136\\
				0.51	0.12	0.00938726976920778\\
				0.51	0.124	0.012952037513357\\
				0.51	0.128	0.0163144373173482\\
				0.51	0.132	0.0194779118339182\\
				0.51	0.136	0.022445611491909\\
				0.51	0.14	0.0252204166126473\\
				0.51	0.144	0.0278049567766476\\
				0.51	0.148	0.0302016278061197\\
				0.51	0.152	0.0324126066678898\\
				0.51	0.156	0.0344398645515261\\
				0.51	0.16	0.0362851783364005\\
				0.51	0.164	0.0379501406273472\\
				0.51	0.168	0.039436168510096\\
				0.51	0.172	0.0407445111536759\\
				0.51	0.176	0.0418762563666167\\
				0.51	0.18	0.0428323361963425\\
				0.51	0.184	0.0436135316460677\\
				0.51	0.188	0.0442204765703219\\
				0.51	0.192	0.0446536607985605\\
				0.51	0.196	0.0449134325258429\\
				0.51	0.2	0.045\\
				0.51	0.352	0.045\\
				0.51	0.356	0.045\\
			};
			\addplot3 [color=red, line width=2.0pt]
			table[row sep=crcr] {%
			%
				0.54	0.124	-0.00204796248664301\\
				0.54	0.128	0.00131443731734815\\
				0.54	0.132	0.00447791183391819\\
				0.54	0.136	0.00744561149190895\\
				0.54	0.14	0.0102204166126473\\
				0.54	0.144	0.0128049567766476\\
				0.54	0.148	0.0152016278061197\\
				0.54	0.152	0.0174126066678898\\
				0.54	0.156	0.0194398645515261\\
				0.54	0.16	0.0212851783364005\\
				0.54	0.164	0.0229501406273472\\
				0.54	0.168	0.024436168510096\\
				0.54	0.172	0.0257445111536759\\
				0.54	0.176	0.0268762563666167\\
				0.54	0.18	0.0278323361963425\\
				0.54	0.184	0.0286135316460677\\
				0.54	0.188	0.0292204765703219\\
				0.54	0.192	0.0296536607985605\\
				0.54	0.196	0.0299134325258429\\
				0.54	0.2	0.03\\
				0.54	0.352	0.03\\
				0.54	0.356	0.03\\
			};
			\addplot3 [color=red, line width=2.0pt]
			table[row sep=crcr] {%
			%
			%
				0.57	0.144	-0.00219504322335232\\
				0.57	0.148	0.000201627806119775\\
				0.57	0.152	0.00241260666788987\\
				0.57	0.156	0.00443986455152617\\
				0.57	0.16	0.00628517833640059\\
				0.57	0.164	0.00795014062734726\\
				0.57	0.168	0.00943616851009608\\
				0.57	0.172	0.010744511153676\\
				0.57	0.176	0.0118762563666167\\
				0.57	0.18	0.0128323361963425\\
				0.57	0.184	0.0136135316460678\\
				0.57	0.188	0.0142204765703219\\
				0.57	0.192	0.0146536607985605\\
				0.57	0.196	0.014913432525843\\
				0.57	0.2	0.015\\
				0.57	0.352	0.015\\
				0.57	0.356	0.015\\
			};
			\addplot3 [color=red, line width=2.0pt]
			table[row sep=crcr] {%
				0.6	0	nan\\
				0.6	0.004	-0.275762059231227\\
				0.6	0.008	-0.257567837237454\\
				0.6	0.012	-0.241682442650527\\
				0.6	0.016	-0.227312343275276\\
				0.6	0.02	-0.214080912865213\\
				0.6	0.024	-0.201766524253657\\
				0.6	0.028	-0.190222904729933\\
				0.6	0.032	-0.179346242939318\\
				0.6	0.036	-0.169059054880769\\
				0.6	0.04	-0.159301321923216\\
				0.6	0.044	-0.150025212550642\\
				0.6	0.048	-0.141191740413791\\
				0.6	0.052	-0.132768544491603\\
				0.6	0.056	-0.124728356507143\\
				0.6	0.06	-0.11704790858508\\
				0.6	0.064	-0.10970713360783\\
				0.6	0.068	-0.102688566376734\\
				0.6	0.072	-0.0959768862815161\\
				0.6	0.076	-0.0895585620348308\\
				0.6	0.08	-0.0834215715337913\\
				0.6	0.084	-0.0775551780206179\\
				0.6	0.088	-0.0719497491114103\\
				0.6	0.092	-0.0665966089360387\\
				0.6	0.096	-0.0614879161846433\\
				0.6	0.1	-0.0566165626622601\\
				0.6	0.104	-0.0519760882522146\\
				0.6	0.108	-0.0475606091375473\\
				0.6	0.112	-0.0433647568319607\\
				0.6	0.116	-0.0393836260981786\\
				0.6	0.12	-0.0356127302307922\\
				0.6	0.124	-0.032047962486643\\
				0.6	0.128	-0.0286855626826518\\
				0.6	0.132	-0.0255220881660818\\
				0.6	0.136	-0.022554388508091\\
				0.6	0.14	-0.0197795833873527\\
				0.6	0.144	-0.0171950432233523\\
				0.6	0.148	-0.0147983721938802\\
				0.6	0.152	-0.0125873933321101\\
				0.6	0.156	-0.0105601354484738\\
				0.6	0.16	-0.00871482166359942\\
				0.6	0.164	-0.00704985937265276\\
				0.6	0.168	-0.00556383148990394\\
				0.6	0.172	-0.00425548884632405\\
				0.6	0.176	-0.00312374363338328\\
				0.6	0.18	-0.00216766380365748\\
				0.6	0.184	-0.00138646835393226\\
				0.6	0.188	-0.000779523429678108\\
				0.6	0.192	-0.000346339201439472\\
				0.6	0.196	-8.65674741570288e-05\\
				0.6	0.2	0\\
				0.6	0.352	0\\
				0.6	0.356	0\\
			};
	
			%
%
		[%
		width=6.028in,
		height=4.754in,
		at={(1.011in,0.642in)},
		scale only axis,
		unbounded coords=jump,
		xmin=0,
		xmax=0.6,
		tick align=outside,
		xlabel style={font=\color{white!15!black}},
		ymin=0,
		ymax=0.35,
		ylabel style={font=\color{white!15!black}},
		zmin=-0,
		zmax=0.35,
		zlabel style={font=\color{white!15!black}},
		view={43}{14},
		axis background/.style={fill=white},
		axis x line*=bottom,
		axis y line*=left,
		axis z line*=left,
		xmajorgrids,
		ymajorgrids,
		zmajorgrids,
		z tick label style={/pgf/number format/.cd,%
			scaled z ticks = false,
			set decimal separator={.},
			fixed},
		y tick label style={/pgf/number format/.cd,%
			scaled z ticks = false,
			set decimal separator={.},
			fixed}
		]
		
	\end{axis}
	\end{tikzpicture}%
		\caption{Boundary of the capacity-distortion region $\Cc\Dc$ for   Example 3 in Subsection~\ref{ex:Binary1}.}
		\vspace{-0.5cm}
		\label{fig:binary}
	\end{figure}
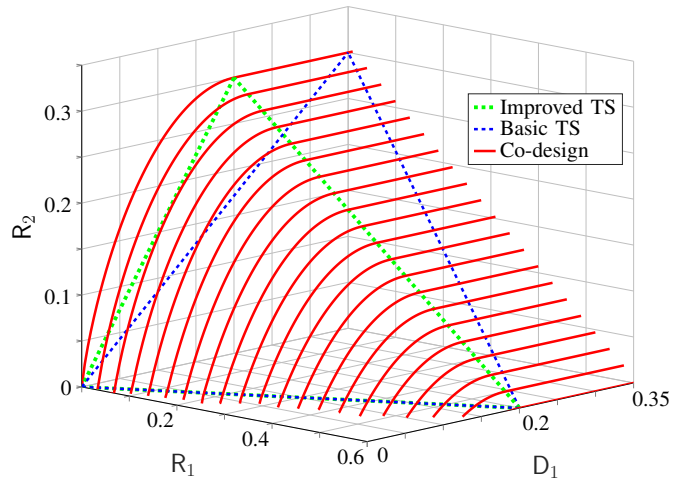
	Fig.~\ref{fig:binary} shows in red colour the  boundary  of the projection of the tradeoff region $\CDc$ of this example onto the $3$-dimensional plane $(\R_1,  \R_2,  \D_1)$, for parameters $\gamma=0.5$ and $q=0.6$.
	The tradeoff with $\D_2$ is omitted for simplicity and because $\D_2$ is a scaled version of $\D_1$. 	
	The figure also shows the boundaries of the basic and improved TS baseline schemes. We again notice a significant gain for an optimal co-design scheme compared to the TS baseline schemes.

	%

	So far,  there was no tradeoff between the two distortion constraints $\D_1$ and $\D_2$. This is different in the next example,  which otherwise is very similar.

	\subsubsection{Example 4: Binary BC with Multiplicative Bernoulli States and Flipping Inputs}\label{ex:BinaryFlipped}
	\newcommand{\bp}{{\bar{p}}}
	\newcommand{\bq}{{\bar{p_s}}}
	\newcommand{\bg}{{\bar{\gamma}}}
	Reconsider the same state pmf $P_{S_1S_2}$ as in the previous example,  but  now {an} SDMBC with a transition law that flips the input for receiver 2:
	\begin{IEEEeqnarray}{rCl}\label{model:flipped}
		Y_1&=& S_1 X,  \qquad Y_2 = S_2(1-X).
	\end{IEEEeqnarray}
	As in the previous example we consider output feedback $Z=(Y_1, Y_2)$. 
	\begin{corollary}\label{corollary:binaryBC2}		
		The capacity-distortion region $\CDc$ of the binary SDMBC with flipping inputs in \eqref{model:flipped} and output feedback is the set of all tuples $(\R_0,\R_1, \R_2, \D_1, \D_2)$ satisfying				\begin{subequations}\label{eq:con2}
			\begin{IEEEeqnarray}{rCl}
				\R_1 &\leq& q H_\textnormal{b}(p) r, \label{r1} \\
				\R_0+	\R_2 &\leq&\gamma  q H_\textnormal{b}(p) (1-r), \label{r2} \\
				\D_1&\geq&p   \min \{ q (1-\gamma), (1-q)\}, \label{d1}\\
				\D_2&\geq&(1-p) q \min \{  \gamma,  1-\gamma \},\label{d2}
			\end{IEEEeqnarray}
		\end{subequations}
		for some choice of the parameters $r,  p\in[0, 1]$.

	\end{corollary}	
	The capacity-distortion region expression above captures the tradeoffs between the two rates through the parameter $r$, between the rates and the distortions through the parameter $p$, and between the two distortions through the parameter $p$.
	
	{Comparing above Corollary~\ref{corollary:binaryBC2} to the previous Corollary~\ref{corollary:binaryBC1}, we remark the identical rate constraints and the relaxed distortion contraints for both $\D_1$ and $\D_2$ in Corollary~\ref{corollary:binaryBC2}.  The reason is that the flipping input allows  the transmitter to perfectly estimate $S_1$ from $(X,Y_1,Y_2)$ not only when  $X=1$ but also when $X=0$ and $Y_2=1$ because they imply that $S_2=1$ and by \eqref{ew:p_ss} also $S_1=1$.}
	
	\begin{IEEEproof}
		{ 
			The proof is similar to the proof of Corollary~\ref{corollary:binaryBC1}, except for the description of the optimal estimators. To determine these optimal estimators, we remark that  only four   input-output relations are possible: $(x, y_1, y_2)\in\{(0,0,0), (0,0,1), (1,0,0), (1,1,0)\}$.  Moreover, when $X=1$, then $Y_1=S_1$, and when $X=0$, then $Y_2=S_2$. In particular, when $X=0$ and $Y_2=1$, then $S_2=1$ and also $S_1=1$, see \eqref{ew:p_ss}.   The optimal estimator for state $S_1$ thus is:
			\begin{equation}\label{eq:estimator1}
				\hat{s}_1^*(x,y_1,y_2)= \begin{cases} y_1, &\textnormal{if } x=1 \\
					1, &\textnormal{if } (x, y_2) = (0,1)\\
					\underset{s}{\arg\min} P_{S_1|S_2} (s|0), &\textnormal{else}  ,
				\end{cases}
			\end{equation}
			and $\hat{s}_1^*(X,Y_1,Y_2)=S_1$ unless $X=0$, $Y_2=0$, and $S_1\neq 	\arg\min_{s} P_{S_1|S_2} (s|0)$, which is equivalent to $(X=0, S_2=0)$ and $S_1\neq 	\arg\min_{s} P_{S_1|S_2} (s|0)$. This
			yields $c_1(1)=0$ and because $S_2$ is independent of $X$:
			\begin{align}
				c_1(0) 
				&= P_{S_2}(0) \min_{s} P_{S_1|S_2} (s|0).
			\end{align}
			Recalling $p=P_X(0)$, we readily obtain the distortion for state $S_1$:
			\begin{align}
				\D_1 = p \min_{s} P_{S_1,S_2} (s,0)= p \min\{q(1-\gamma) , 1-q\}.
			\end{align}
			The optimal estimator and the corresponding distortion for state $S_2$ can be  obtained in a similar way.
		}	
	\end{IEEEproof}	
	
	\subsection{Capacity-Distortion Region for General SDMBCs}
	
	In the remainder of this section, we reconsider general SDMBCs, for which we present bounds on $\CDc$. We start with a simple outer bound.

	\begin{theorem}[Outer Bound on $\CDc$]\label{outer1}
		If $(\R_0,\R_1,  \R_2,  \D_1,  \D_2)$ lies in $\CDc$ for a given SDMBC $P_{Y_1Y_2Z|S_1S_2X}$ with state pmf $P_{S_1S_2}$, then there exist pmfs $P_X, P_{U_1|X}, P_{U_2|X}$ such that the random tuple $(U_k,  X,  S_1, S_2, Y_1, Y_2,  Z)\sim P_{U_k|X}P_X P_{S_1S_2} P_{Y_1Y_2Z\mid S_1S_2X}$ satisfies the  rate constraints 
		\begin{subequations}\label{eq:Rupper}
			\begin{IEEEeqnarray}{rCl}
				\R_0+	\R_k &\leq &I(U_k;Y_k \mid S_k), \quad k=1,2,\label{R1}\\
				\R_0+\R_1+\R_2 &\leq &I(X;Y_1,  Y_2 \mid  S_1,  S_2), \label{R2}
			\end{IEEEeqnarray}
		\end{subequations}
		and the average distortion constraint
		\begin{equation}\label{eq:dis_general_BC}
			\mathbb{E}[d_k(S_k,  \hat{s}_{k}^*(X, Z))]\leq \D_k,  \quad k=1, 2, 
		\end{equation}
		where the function $\hat{s}_{k}^*(\cdot,  \cdot)$ is defined in \eqref{eq:BCestimator}.
	\end{theorem}
	\begin{IEEEproof}  
		See Appendix~\ref{app:converse_proof}. 
	\end{IEEEproof}
	
	
	Achievability results are easily obtained by combining existing achievability results for SDMBCs with generalized feedback with the optimal estimator in Lemma~\ref{BC:lemma:Shat}. We consider the block-Markov coding scheme in  \cite{shayevitz2012capacity}, which in each block applies Marton coding to transmit fresh data  to the receivers  as well as  compression information describing the inputs and outputs of the previous block. The receivers decode the Marton codewords backwards,  starting from the last block, and using both their channel outputs as well as the previously decoded compression information pertaining to the block. Combining this scheme with the optimal estimator in Lemma~\ref{BC:lemma:Shat} yields the following proposition.   
	\begin{proposition}[Inner Bound on $\CDc$]\label{prp:inner}
		Consider {an} SDMBC   $P_{Y_1Y_2Z|S_1S_2X}$ with state pmf $P_{S_1S_2}$. The capacity-distortion region $\CDc$ includes all tuples $(\R_0,\R_1, \R_2, \D_1, \D_2)$ that satisfy  inequalities \eqref{eq:inner} on top of this page and  the distortion constraints \eqref{eq:dis_general_BC}. 
		where $(U_0,   U_1, U_2, X, S_1, S_2,  Y_1, Y_2, Z,  V_0, V_1, V_2)\sim$ 
		$ P_{U_0U_1U_2X} P_{S_1S_2}P_{Y_1Y_2Z|S_1S_2X}   P_{V_0V_1V_2|U_0U_1U_2Z}$,  
		for some choice of (conditional) pmfs $P_{U_0U_1U_2X}$ and $P_{V_0V_1V_2|U_0U_1U_2Z}$.
	\end{proposition}
	
	\begin{figure*}[t!]
		\hrule	\begin{subequations}\label{eq:inner}
			\begin{IEEEeqnarray}{rCl}
				\R_0+	\R_1 &\leq& I(U_0, U_1;Y_1, V_1\mid S_1) - I(U_0, U_1, U_2,  Z;V_0, V_1 \mid S_1, Y_1)\label{eq:inner1}
				\\
				\R_0+	\R_2 &\leq& I(U_0, U_2; Y_2, V_2\mid S_2) - I(U_0, U_1, U_2,  Z;V_0, V_2\mid S_2, Y_2)\label{eq:inner2}
				\\
				\R_0+	\R_1+\R_2 &\leq& I(U_1; Y_1, V_1|U_0,  S_1) + I(U_2; Y_2, V_2\mid U_0,  S_2) + \min_{k\in\{1, 2\}}I(U_0;Y_k, V_k\mid S_k)- I(U_1;U_2 \mid U_0)\nonumber \\
				& &-I(U_0, U_1, U_2,  Z;V_1\mid V_0, S_1, Y_1) -I(U_0, U_1, U_2,  Z;V_2|V_0, S_2, Y_2) \nonumber \\
				&& - \max_{k\in\{1, 2\}} I(U_0, U_1, U_2,  Z;V_0 \mid S_k, Y_k)\label{eq:inner3}\\
				\nonumber 2\R_0+ \R_1+\R_2 &\leq& I(U_0,U_1; Y_1, V_1\mid S_1) + I(U_0, U_2; Y_2, V_2\mid S_2) - I(U_1;U_2\mid U_0) \\
				&& - I(U_0, U_1, U_2,  Z;V_0, V_1 \mid S_1, Y_1) -I(U_0, U_1, U_2, Z;V_0, V_2 \mid S_2, Y_2)\label{eq:inner4}
			\end{IEEEeqnarray}
		\end{subequations}
		\hrule
	\end{figure*}
	\begin{IEEEproof}
		Similar to \cite{shayevitz2012capacity} and omitted.  
		\qedhere
	\end{IEEEproof}

	In analogy to Corollary~\ref{cor1:notradeoff} for the single-receiver case, for some SDMBCs  there is no  tradeoff between the achievable distortions and communication rates. In this case, for the BC, 
	the capacity-distortion region is given by the Cartesian product between the SDMBC's capacity region:
	\begin{IEEEeqnarray}{rCl}
	\hspace{0cm}	\mathcal{C}:= &&\{(\R_0,\R_1,  \R_2)\colon \D_1 \geq 0, \;\D_2 \geq 0\quad \nonumber\\
		&& \hspace{2cm}\textnormal{ s.t. } (\R_0,\R_1, \R_2, \D_1, \D_2) \in \CDc  \}, \IEEEeqnarraynumspace
	\end{IEEEeqnarray}
	and its distortion region:
	\begin{IEEEeqnarray}{rCl}
		\mathcal{D}:= &&\{(\D_1,  \D_2)\colon \R_0\geq 0, \;\R_1\geq 0,\; \R_2 \geq 0  
		\nonumber\\
	&& \hspace{2cm} \textnormal{ s.t. } (\R_0,\R_1, \R_2, \D_1, \D_2) \in \CDc  \}. \IEEEeqnarraynumspace
	\end{IEEEeqnarray}

	\begin{proposition}[No Rate-Distortion Tradeoff]\label{prp1:BC:notradeoff}
		Consider an SDMBC $P_{Y_1Y_2Z|S_1S_2X}$ with state pmf $P_{S_1S_2}$ for which there exist functions $\psi_1$ and $\psi_2$ with domain $\Xc\times \Zc$ so that irrespective of the input distribution $P_X$  the relations 
		\begin{IEEEeqnarray}{rCl}
			&(S_k, \psi_{k}(Z,  X)) \perp  X, \label{cond1}\\
			&S_k\markov \psi_{k}(Z,X)\markov (Z, X),  \quad k=1, 2,  \label{cond2}
		\end{IEEEeqnarray} 
		hold for $(S_1,S_2,X_2\textcolor{blue}{, }Z) \sim P_{S_1}P_{S_2}P_X P_{Z|XS_1,X_2}$. The capacity-distortion region of this  SDMBC is the product of the capacity region and the distortion region:
		\begin{IEEEeqnarray}{rCl}\label{eq:product}
			\CDc =  \mathcal{C} \times \mathcal{D}.
		\end{IEEEeqnarray}	
	\end{proposition}
	
	\begin{IEEEproof}
		Analogous to the proof of Corollary~\ref{cor1:notradeoff}. Specifically, the proof is obtained from Appendix~\ref{app:notradeoff} by replacing $(S,\hat{S},\psi,Y,T)$ with $(S_k,\hat{S}_k,\psi_k,Y_k,T_k)$, for $k=1,2$.
	\end{IEEEproof}

	\subsection{Example 5: Erasure BC with Noisy Feedback} Our first example satisfies Conditions \eqref{cond1} and \eqref{cond2} in Proposition~\ref{prp1:BC:notradeoff} for an appropriate choice of $\psi_1$ and $\psi_2$, and its capacity-distortion region is thus given by the product of the capacity region and the distortion region.
	
	Let $(E_1,  S_1, E_2, S_2) \sim  P_{E_1S_1E_2S_2}$ over $\{0,1\}^4$ be given but arbitrary. 
	Consider the  state-dependent erasure BC
	\begin{IEEEeqnarray}{rCl}\label{eq:Y}
		Y_{k}=\begin{cases}
			X & \text{if }  S_{k}=0, \\
			? & \text{if } S_{k}=1, 
		\end{cases}, \qquad k=1, 2, 
	\end{IEEEeqnarray}
	where the feedback signal $Z=(Z_1, Z_2)$ is given by
	\begin{IEEEeqnarray}{rCl}\label{eq:Z}
		Z_k=\begin{cases}
			Y_{k} & \text{if }  E_k=0, \\
			? 
			& \text{if } E_k=1, 
		\end{cases}, \qquad k=1, 2.
	\end{IEEEeqnarray}
	Further consider  Hamming distortion measures $d_k(s,  \hat{s}) = s \oplus \hat{s}$,  for $k=1, 2$. 
	For the choice
	\begin{equation}
		\psi_k(Z_k)=
		\begin{cases}
			1, & \text{if} \quad Z_k=?,\\
			0,   & \text{else},  
		\end{cases}
	\end{equation}
	the described SDMBC satisfies the conditions in Proposition~\ref{prp1:BC:notradeoff}, thus yielding the following corollary. 
	
	\begin{corollary}The capacity-distortion region of the state-dependent erasure BC with noisy feedback in \eqref{eq:Y}--\eqref{eq:Z} is the Cartesian product between the capacity region of the SDMBC and its distortion region:
		\begin{equation}
			\CDc = \Cc \times \Dc.
		\end{equation}
		When $P_{E_1S_1E_2S_2}=P_{E_1S_1}P_{E_2S_2}$, then the distortion region is given by:
		\begin{equation}\label{eq:d}
			\Dc=\{(\D_1,\D_2) \colon \D_k \geq P_{E_kS_k}(1,0)\}.
		\end{equation}

		
	\end{corollary}
			
	\begin{proof}
		{The state can  perfectly be estimated ($S_k=0)$ with zero distortion  if $(S_k, E_k)=(0,0)$.  Otherwise, the feedback is $Z_k=?$ and provides no information. The optimal estimator is then given by the best  constant estimator, which in this example is:
			\begin{align}
				\hat{s}_{{\rm const},k} =\onev\{P_{S_k}(1)\geq P_{S_kE_k}(0,1)\}.
			\end{align}
			This immediately yields the distortion constraint in \eqref{eq:d}.}
	\end{proof}
	
	Notice that the capacity region $\Cc$ of the  SDMBC \eqref{eq:Y} is unknown even with perfect feedback. In \cite{wang2012capacity,gatzianas2013multiuser}, the capacity region of this SDMBC  with perfect feedback was characterized when each receiver is informed about the state realizations at \emph{both} receivers. 	
\subsection{Example 6: State-Dependent Dueck's BC with { Multiplicative Bernoulli States}} \label{subsection: Dueck}
We consider a state-dependent  version of  Dueck's example in \cite{DUECK}, which first served to show that feedback can increase capacity of a memoryless BC. Interestingly, despite its simplicity, the state-dependent extension of this example allows observing various kinds of tradeoffs between communication and sensing performances and also between performances at the various receivers.  For example, for specific choices of parameters, the problems of sensing and communication decompose (Corollary~\ref{deuck_noTradeoff}), and it is possible to simultaneously achieve the optimal sensing and communication performances. 
For other parameters a tradeoff arises.  The present example also shows nicely that  our presented co-design scheme can significantly outperform the two TS methods.

	\begin{figure}
		\centering
	\includegraphics[scale=0.92]{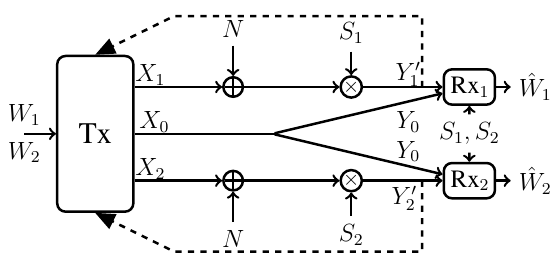}
		\caption{State-dependent Dueck Broadcast Channel.}
		\label{fig:Dueck}
	\end{figure}
	Consider  the state-dependent  BC  in Figure~\ref{fig:Dueck} 
	with input  $X=(X_0, X_1, X_2)\in \{0, 1\}^3$, i.i.d. Bernoulli states $S_1, S_2\sim \mathcal{B}(q)$, for $q \in [0,1]$,  and outputs 
	\begin{equation}
		Y_k=(X_0,  Y_k',  S_1, S_2),  \qquad  {k=1, 2}, 
	\end{equation}
	where 
	\begin{subequations}\label{eq:Mehrasa}
		\begin{equation}
			Y_k' =  S_k (X_k \oplus N),  \qquad {k=1, 2}, 
		\end{equation}
	\end{subequations}
	and  the noise $N\sim \mathcal{B}(1/2)$ is independent of the inputs and the states. 
	The feedback signal is  
	\begin{equation}
		Z= (Y_1',  Y_2'),
	\end{equation}
	and for simplicity we again ignore the common rate $\R_0$. 
	
	We notice that only $X_1$ and $X_2$ are corrupted by the state and the noise. Since $X_0$ is received without any state or noise, it is thus completely useless for sensing. 
	In fact, the optimal estimator of Lemma~\ref{BC:lemma:Shat} for $k=1,2$ is  (see Appendix~\ref{appendix-sec:ex_estimator}) 
	\begin{IEEEeqnarray}{rCl}\label{ex:optimal_estimator}
	&&\hspace{-0.3cm}	\hat{s}_k^*(x_1,x_2,y_1',y_2')\nonumber\\&&\hspace{-0.2cm} =\hspace{-0.12cm}\begin{cases} 
			\mathbbm{1}\{q \geq (1-q)\} &\hspace{-0.1cm}y_k'=0, y_{\bar{k}}'=1, x_1\neq x_2\\
			0 &\hspace{-0.1cm}y_k'=0, y_{\bar{k}}'=1, x_1= x_2\\
		1 &\hspace{-0.1cm}y_k'=1 \\
			0 &\hspace{-0.1cm}y_1=y_2=0, x_1\neq x_2\\
			\mathbbm{1}\{q \geq (1-q)(2-q)\} &\hspace{-0.1cm}y_1'=y_2'=0, x_1= x_2
		\end{cases},\nonumber\\
	\end{IEEEeqnarray}
	where we slightly abuse notation by omitting the argument $x_0$ for the estimator $\hat{s}_k^*$ because this latter does not depend on $x_0$. 
	
	For a  given input pmf  with probability $t:=\Prv{X_1\neq X_2}$, the expected distortion achieved by the optimal estimators in \eqref{ex:optimal_estimator} is (see Appendix~\ref{app:ex_mindist}): 
	\begin{IEEEeqnarray}{rCl}\label{ex:distortion}
		\lefteqn{		\hspace{-2cm}\E{d_k(S_k,\hat{s}_k^*(X_1,X_2,Y_1',Y_2'))} }
	&&
		\nonumber \\
	\hspace{-0.7cm}&=&	\frac{1}{2} t q \left( \min \{q, 1-q\} + (1-q) \right)
		\nonumber\\
		\hspace{1.5cm} &+&\frac{1}{2} (1-t)  \min \{q,(1-q)(2-q)\} 
	\end{IEEEeqnarray}
	We observe different cases: i) for $q\in[0,1/2]$, both minima are achieved by $q$; ii) for $q \in \big(1/2, 2- \sqrt{2}\big]$,  the first and second minima are achieved by $1-q$ and $q$, respectively; iii) for $q \in \big(2-\sqrt2,1\big]$,  the first and second minimum are achieved by $(1-q)$ and $(1-q)(2-q)$, respectively. 
	The distortion constraint~\eqref{eq:dis_general_BC} thus evaluates to:
	\begin{equation}\label{ex:distortin}
		\D_k \geq \begin{cases}  q/2 & q \in[0,1/2] \\
			q(1-t(2q-1))/2 & \hspace{-0.4cm}q \in \big(1/2, 2- \sqrt{2}\big]\\
			(1-q)(2-q+t(3q-2))/2 & q \in \big(2-\sqrt2,1\big].
		\end{cases}
	\end{equation} 
	
	We notice that for $q \in [0,1/2]$, the distortion constraint is independent of $t$ and thus of $P_X$, and the minimum expected distortions are  $\D_{\min,1}=\D_{\min,2}=\frac{1}{2} q$.  For $q \in \big(1/2, 2- \sqrt{2}\big]$, the minimum expected distortions are achieved for $t=1$ and the same holds also for $q\in\big(2-\sqrt{2}, 2/3\big]$. For $q \in [2/3,1]$, the distortions are minimized for $t=0$. We thus have $\D_{\min,1}=\D_{\min,2}=\D_{\min}$, where
	\begin{IEEEeqnarray}{rCl}\label{eq:DueckMinDist}
		\D_{\min}:=\begin{cases} q/2, & q \in [0,1/2]\\
			q(1-q), & q\in [1/2, 2/3] \\
			(1-q)(2-q)/2, & q \in[2/3,1].
		\end{cases}
	\end{IEEEeqnarray}
	We obtain a characterization of the distortion region:
	\begin{IEEEeqnarray}{rCl}\label{ex:D}
		\Dc & = & \{(\D_1, \D_2) \colon \D_1\geq \D_{\min}, \; \D_2 \geq \D_{\min}\}.
	\end{IEEEeqnarray}
	{The private-messages} capacity region is:
	\begin{IEEEeqnarray}{rCl}\label{ex:capacity}
	\hspace{-2cm}	\Cc & = & \{(\R_1, \R_2)\colon \R_1 \leq 1,\; \R_2 \leq 1,\; 
		\nonumber\\
		&&\hspace{3cm}\textnormal{and} \quad \R_1+\R_2\leq 1+q^2\}.
	\end{IEEEeqnarray}    	
	The converse and achievability proofs are provided in Appendices~\ref{app:ex_converse} and \ref{app:ex_ach}, respectively.

	Reconsider now the case where $q\in[0,1/2]$. As previously explained, the distortion is independent of the input distribution, and the capacity-distortion region $\CDc$ degenerates to the product of the capacity and distortion regions: 
	\begin{corollary}\label{deuck_noTradeoff}[No Rate-Distortion Tradeoff] For above state-dependent Dueck example with $q\in [0,1/2]$:
		\begin{equation}
			\CDc = \Cc \times \Dc.
		\end{equation}
	\end{corollary}
	For  the general case,  we only have bounds on the capacity-distortion region $\CDc$. We first present our outer bound, which is based on Theorem~\ref{outer1} and proved in Appendix~\ref{app:ex_converse}.
	\begin{corollary}[Outer Bound]\label{corollary:DueckOuter}
		The capacity-distortion region $\CDc$ (without common message) of Dueck's state-dependent BC  is included in the set of  tuples $(\R_1, \R_2, \D_1, \D_2)$ that for some choice of the parameters $ t \in [0, 1]$ satisfy the rate-constraints 	
		\begin{IEEEeqnarray}{rCl}
			\R_k&\leq& 1, \qquad k=1,2,\\
			\R_1+\R_2 &\leq& 1+q^2 H_\textnormal{b}(t) \label{eq:R1-outer}
		\end{IEEEeqnarray}
		and   the distortion constraints in \eqref{ex:distortin}.
	\end{corollary}

	The inner bound is based on Proposition~\ref{prp:inner}, see Appendix~\ref{app:ex_ach}. Together with the outer bound in Corollary~\ref{corollary:DueckOuter} it characterizes both the distortion region $\Dc$ and the capacity region $\Cc$ in \eqref{ex:D} and \eqref{ex:capacity}.

	\begin{corollary}[Inner bound]\label{corollary:DueckInner}
		The capacity-distortion region $\CDc$ of the state-dependent Dueck BC includes all rate-distortion tuples $(\R_1,\R_2,\D_1,\D_2)$ that for some choice of $t\in[0,1]$ satisfy \eqref{ex:distortin}  and 
		\begin{IEEEeqnarray}{rCl}
			\R_k &\leq &1, \quad k=1,2,  \\  
			\R_1 +\R_2 & \leq & 1+q H_{\textnormal{b}}(t) - q(1-q), 
		\end{IEEEeqnarray}
		as well as the convex hull of all these tuples.
	\end{corollary}

	%
	
	Fig.~\ref{general_example} shows our outer and inner bounds in Corollaries~\ref{corollary:DueckOuter} and \ref{corollary:DueckInner}  for $q=3/4$, where in the inner bound we consider the convex hull operation through convex combinations between values of $t> 0$ and $t=0$.  The figure also shows the performances of the basic and improved TS baseline schemes, {whose modes we explain next. (Recall that the basic TS scheme time-shares the sensing mode without communication and the communication mode without sensing, and the improved TS scheme time-shares the sensing mode with communication and the communication mode with sensing.)}

	\underline{Sensing mode with and without communication:}\\ [0.2em]
	In the sensing mode with communication, one can choose an arbitrary pmf for $X_0$, e.g., $X_0$ Bernoulli-$1/2$ because this input does not  affect the sensing. From \eqref{eq:DueckMinDist}, the minimum distortions of $\D_{\min,1}=\D_{\min,2}=5/32$ are achieved by setting $X_1=X_2$ with probability 1. For $X_1=X_2$ the sum-rate cannot exceed $\R_1+\R_2 \leq1$,  because $I(X_0,X_1,X_2; Y_1,Y_2)=I(X_0,X_2; Y_1,Y_2) \leq H(X_0) + I(X_2; Y_1',Y_2'|X_0) \leq 1$ as $Y_1'$ and $Y_2'$ are corrupted by the Bernoulli-$1/2$ noise $N$. On the other hand,  any rate pair $(\R_1,\R_2)$ of sum-rate $\R_1+\R_2=1$ is trivially achievable by communicating only with the noiseless    $X_0$-input. 
	
	We conclude that the sensing mode with communication achieves the rate-distortion tuple $(\R_1,\R_1,\D_1,\D_2)$ satisfying 
	\begin{IEEEeqnarray}{rCl}
		\R_1+\R_2\leq 1 \quad \textnormal{and} \quad  \D_k \geq  5/32 ,\;\; k=1,2. 
	\end{IEEEeqnarray}
	{If the transmitter cannot perform communication and sensing tasks simultaneously,  the same minimum distortions are achieved but the rates are trivially  zero. 
		\begin{IEEEeqnarray}{rCl}
			\R_1+\R_2=0 \quad \textnormal{and} \quad  \D_k \geq  5/32 ,\;\; k=1,2. 
		\end{IEEEeqnarray}
	}
	
	\underline{Communication mode with and without sensing: } \\ [0.2em]
	The optimal pmf $P_X$ achieving the capacity region in \eqref{ex:capacity} corresponds to  i.i.d. Bernoulli-$1/2$ distributed $X_0, X_1,X_2$ ( Appendix~\ref{app:ex_ach}). { The corresponding sum rate is $\R_1+\R_2=1+q^2=25/16$.}
	The minimum achievable distortions are thus obtained from \eqref{ex:distortin} by setting $t=\Prv{X_1\neq X_2}=1/2$, i.e., $\D_{\max,1}=\D_{\max,2}= 11/64$. The best constant estimator is $\hat{S}_1=\hat{S}_2=1$ because $3/4=P_{S_k}(1)>P_{S_k}(0)=1/4$, which achieves distortions $\D_{\textnormal{trivial},1}=\D_{\textnormal{trivial},2}=1/4$.
	We can conclude that the communication mode with sensing achieves all rate-distortion tuples $(\R_1,\R_1,\D_1,\D_2)$  satisfying
	\begin{IEEEeqnarray}{rCl}
			\R_1+\R_2&\leq& 25/16 ,
		\nonumber\\
		\R_k &\leq& 1   \quad \textnormal{and} \quad  \D_k \geq  11/64    \;\;\; k=1, 2 	
	\end{IEEEeqnarray}
	and the communication mode without sensing achieves 
	all rate-distortion tuples $(\R_1,\R_1,\D_1,\D_2)$  satisfying
	\begin{IEEEeqnarray}{rCl}
		 \R_1+\R_2&\leq& 25/16 
		 \nonumber\\
		 	\R_k&\leq& 1  \quad \textnormal{and}  \quad   \D_k\geq  1/4,  \;\;\; k=1, 2. 
	\end{IEEEeqnarray}
	
	\begin{figure}[ht]
		\centering
		\hspace{-1cm}
		\begin{tikzpicture}[every pin/.style={fill=white}, scale=.85]
			\begin{axis}[%
				xlabel={ \large{$\D_1=\D_2$ }}, 
				ylabel={ \large{$\R_1+\R_2$ }}, 
				at={(1.011in, 0.642in)}, 
				scale only axis, 
				xmin=0.1562, 
				xmax=0.181875, 
				ymin=0, 
				ymax=1.6, 
				axis background/.style={fill=white}, 
				legend style={at={(1, 0.6)}, anchor=north east}
				]
				\addplot [color=black,  line width=2.0pt]
				table[row sep=crcr]{%
					0.15625	1\\
					0.15725	1.11493262742641\\
					0.15825	1.19300663526176\\
					0.15925	1.25660468269536\\
					0.16025	1.31045924881934\\
					0.16125	1.35679912448532\\
					0.16225	1.39692059113385\\
					0.16325	1.43166498432848\\
					0.16425	1.46161569767371\\
					0.16525	1.48719457078673\\
					0.16625	1.50871456997003\\
					0.16725	1.52641087358079\\
					0.16825	1.54046022361587\\
					0.16925	1.55099339848532\\
					0.17025	1.55810337593429\\
					0.17125	1.5618506139973\\
					0.171875	1.5618506139973\\
					0.181875	1.5618506139973\\
				};
				\addlegendentry{Outer Bound of Corollary~\ref{corollary:DueckOuter}}
				
				\addplot [color=red, ,  line width=2.0pt]
				table[row sep=crcr]{%
					0.15625	1\\
					0.16125	1.28823216598042\\
					0.16225	1.34172745484513\\
					0.16325	1.38805331243798\\
					0.16425	1.42798759689828\\
					0.16525	1.46209276104897\\
					0.16625	1.49078609329337\\
					0.16725	1.51438116477439\\
					0.16825	1.53311363148783\\
					0.16925	1.5471578646471\\
					0.17025	1.55663783457905\\
					0.17125	1.5616341519964\\
					0.171875	1.5616341519964\\
					0.181875	1.5616341519964\\
				};
				\addlegendentry{Inner Bound  of Corollary~\ref{corollary:DueckInner}.}

				\addplot [color=mycolor1, dashed,   line width=2.0pt]
				table[row sep=crcr]{%
					0.15625	1\\
					0.1719 1.5626\\};
				\addlegendentry{Improved TS}

				\addplot [color=blue, dashed,   line width=2.0pt]
				table[row sep=crcr]{%
					0.15625	0\\
					0.25 1.5626\\};
				\addlegendentry{Basic TS}

			\end{axis}
		\end{tikzpicture}%
		\vspace{-0.2cm}
		\caption{Sum-rate $\R_1+\R_2$ vs. symmetric distortion $\D_1=\D_2$ for the state-dependent Dueck BC with $q=3/4$.}
		\label{general_example}
		\vspace{-.15cm}
		
	\end{figure}
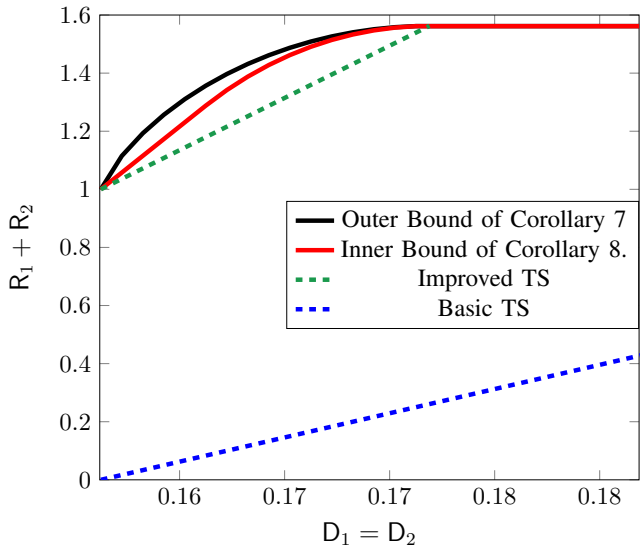
	
	
	\section{Conclusion}\label{section:conclusion}
	
	Motivated by the paradigm of  integrated sensing and communication systems, we studied joint sensing and communication in memoryless state-dependent channels. We fully characterized the capacity-distortion tradeoff for the single-user channels as well as physically-degraded broadcast channels. For  general broadcast channels, we presented inner and outer bounds on the capacity-distortion region.  Through a number of illustrative examples, we demonstrated that the optimal co-design scheme offers non-negligible gain compared to the basic time-sharing scheme that performs either sensing or communication, as well as compared to the improved time-sharing scheme that  \textcolor{black}{integrates both tasks into a single system but chooses the common waveform to \emph{prioritize} one of the  tasks}. Interestingly, there are ideal situations where the capacity is achieved without compromising  the sensing performance. \textcolor{black}{Our results also showed that for the single-transmitter systems studied in this paper the optimal sensing depends only on the employed waveform, but not on the underlying coding scheme. This holds also for broadcast channels where the two tasks are not only connected through the  employed waveform but also through the generalized feedback, which in this case should be exploited to improve the set of achievable rates.  Notice that the situation is different in multi-transmitter situations \cite{ahmadipour2022jointMAC}, such as multiple-access channels, where coding can be used to improve the sensing performance a the multiple transmitters (by conveying information from one transmitter to the other through the generalized feedback links) and thus the code construction used for data communication should be adapted to integrate also coding for sensing.}
	
An interesting line of future research is the characterization of the capacity-distortion tradeoff  for channels with memory.  In this case, feedback increases capacity even on the point-to-point channel. For channels with memory, obtaining good state estimation (sensing) and communication performances seem less contradicting goals, because a good state estimation is also useful to  improve communication.%
	
	
	\section*{Acknowledgements}
	The authors would like to thank Gerhard Kramer for his support on an early version of this paper.
	

	\appendices

	\section{Proof of Lemma~\ref{1user:lemma:Shat}} \label{app:lemma:Shat}
	Recall that $\hat{S}^n=h(X^n,  Z^n)$, and write for each $i=1, \cdots, n$:
	\begin{IEEEeqnarray}{rCl}\label{eq:step1}
	&&	\mathbb{E} \left[d(S_{i},  \hat{S}_{i})\right] 	=
	\nonumber\\
	&& \mathbb{E}_{X^n, Z^n}\left[ \mathbb{E}[d(S_{i},  \hat{S}_{i}) |X^n, Z^n] \right]  
	\nonumber	\\
		&\overset{\mathrm{(a)}}=&\sum_{x^n, z^n} P_{X^nZ^n}(x^n, z^n) 
		\nonumber
		\\
		&&
		\hspace{1cm}
		\sum_{\hat{s}\in \Sc}  P_{\hat{S}_{i}|X^nZ^n} (\hat{s}|x^n, z^n)
		\nonumber\\
	&&\hspace{2cm}	\sum_{s} P_{S_{i}|X_iZ_i}(s|x_i, z_i)  d(s,  \hat{s})
\nonumber		\\
		&\geq& \sum_{x^n, z^n} P_{X^nZ^n}(x^n, z^n) 		
		\nonumber\\
	&&\hspace{1.5cm}
		\min_{\hat{s}\in \Sc}\sum_{s} P_{S_{i}|X_iZ_i} (s|x_i,  z_i)  d( s, \hat{s}) 
	\nonumber	\\
		&=& \mathbb{E}[d(S_{i},  \hat{s}^*(X_i, Z_i))], 
	\end{IEEEeqnarray}
	where $(a)$ holds by the Markov chain 
	\begin{equation}
		\Big(X^{i-1}, X_{i+1}^{n},  Z^{i-1}, Z_{i+1}^n,  \hat{S}_{i}\Big) \markov (X_i, Z_i)\markov S_{i}.
	\end{equation}
	Summing over all $i=1,\ldots, n$, we thus obtain: 
	\begin{IEEEeqnarray}{rCl}\label{eq:sum}
		\Delta^{(n)} & = & \frac{1}{n}\sum_{i=1}^n \mathbb{E} \left[d(S_{i},  \hat{S}_{i})\right] \\
		& \geq &  \frac{1}{n} \sum_{i=1}^n \mathbb{E}[d(S_{i},  \hat{s}^*(X_i, Z_i))],
	\end{IEEEeqnarray}which yields the desired conclusion.

	\section{Proof of Theorem~\ref{th:tradeoff}} \label{app:main_result_P2P}

\subsubsection{Converse} 
Fix a sequence (in $n$) of $(2^{n\R},n)$ codes such that  Limits~\eqref{eq:asymptotics_user1} hold. By Fano's inequality there exists a sequence $\epsilon_n \to 0$ as $n\to \infty$ so that:
\begin{IEEEeqnarray}{rCl}
n\R & \leq& I(W; Y^n,  S^n) + n\epsilon_n \nonumber\\
&=&  I(W; Y^n \mid  S^n)+ n\epsilon_n\nonumber\\
&=& \sum_{i=1}^n H(Y_i \mid Y^{i-1}, S^n) 
\nonumber\\&&\hspace{1cm}
- H(Y_i  \mid W,  Y^{i-1} ,S^n)+ n\epsilon_n \nonumber \\
&\overset{{(a)}}\leq & \sum_{i=1}^n H(Y_i \mid   S_i) 
\nonumber\\&&\hspace{1cm}
- H(Y_i \mid X_i,  Y^{i-1},  W,  S^n) + n\epsilon_n \nonumber\\
&\overset{{(b)}} =&  \sum_{i=1}^n H(Y_i \mid  S_i) - H(Y_i \mid  X_i,   S_i)+ n\epsilon_n\nonumber \\
&=& \sum_{i=1}^n I(X_i;Y_i \mid  S_i)+ n\epsilon_n  
\end{IEEEeqnarray}
where $(a)$ holds because conditioning can only reduce entropy; and $(b)$ holds because $(W,   Y^{i-1},  S^{i-1}, S_{i+1}^n) - (S_i,  X_i) - Y_i $ form a Markov chain.
We continue as:
\begin{IEEEeqnarray}{rCl}
\R & \leq& \frac{1}{n} \sum_{i=1}^n I(X_i;Y_i \mid  S_i)+ \epsilon_n \nonumber \\
&\stackrel{(c)}{\leq}& \frac{1}{n} \sum_{i=1}^n \C_{\textnormal{inf}}
\Bigg(
\sum_x P_{X_i} (x) c(x),
\nonumber\\&&\hspace{3cm}
\sum_x P_{X_i} (x) b(x) 
 \Bigg) 
+ \epsilon_n \nonumber\\
&\stackrel{(d)}{\leq}& \C_{\textnormal{inf}}
\Bigg(\frac{1}{n} \sum_{i=1}^n\sum_x P_{X_i} (x) c(x), 
\nonumber\\
&&\hspace{2cm}\frac{1}{n} \sum_{i=1}^n \sum_x P_{X_i} (x) b(x) \Bigg)
+
 \epsilon_n \nonumber\\
&\stackrel{(e)}{\leq}& \C_{\textnormal{inf}}(\D, \B)
\end{IEEEeqnarray}
where $(c)$ holds by the definition of $\C_{\textnormal{inf}}(\D, \B)$, and  $(d)$ and $(e)$ hold by Lemma \ref{lemma:properties}.
\subsubsection{Achievability}
 Fix $P_X(\cdot)$ and functions $\hat{h}(x,z)$ that achieve $C(\D/(1+\epsilon), \B)$, where $\D$ is the desired distortion and $\B$ is the target cost,  for a small positive number $\epsilon>0$. We define the joint pmf $P_{SXY}:=P_S P_X P_{Y|SX}$. 
\paragraph{Codebook generation} 
Generate $2^{n\R}$ sequences $\{x^n(w)\}_{w=1}^{2^{n\R}}$ by randomly and 
independently  drawing each entry according to $P_X$. This defines the codebook $\Cc=\{x^n(w)\}_{w=1}^{2^{n\R}}$, which is revealed to the encoder and the decoder. 
\paragraph{Encoding} To send a message $w\in \Wc$,  the encoder transmits $x^n(w)$. 
\paragraph{Decoding} Upon observing outputs $Y^n=y^n$ and state sequence $S^n=s^n$, the decoder looks for an  index $\hat{w}$ such that 
\begin{equation}
(s^n,   x^n(\hat{w}),y^n) \in \Tc_{\epsilon}^{(n)} (P_{SXY}).
\end{equation}
If exactly one such index exists, it declares $\hat{W}=\hat{w}$. Otherwise, it declares an error.
\paragraph{Estimation} Assuming that it sent the input sequence $X^n=x^n$ and observed the feedback signal $Z^n=z^n$, the  encoder computes the reconstruction sequence as:
\begin{equation}
\hat{S}^n = ( \hat{s}^*(x_1, z_1), \hat{s}^*(x_2, z_2), \ldots, \hat{s}^*(x_n, z_ n)).
\end{equation}
\paragraph{Analysis} We start by  analyzing  the probability of error and the distortion averaged over the random code construction.  Given the symmetry of the code construction, we can condition on the event $W=1$. 

We then notice that the decoder makes an error, i.e., declares nothing or $\hat{W}\neq 1$ if, and only if, one or both of the following events occur: 
\begin{IEEEeqnarray}{rCl}
\Ec_1 &=& \big\{  (S^n,  X^n(1), Y^n) \notin \Tc_{\epsilon}^{(n)}  (P_{SXY})\big\} \\
\Ec_2 &=& \big\{  ( S^n,  X^n({w}'),Y^n) \in \Tc_{\epsilon}^{(n)} (P_{SXY}) 
\nonumber\\
&&\hspace{3cm}\text{for some $w'\neq 1$} \big\}. 
\end{IEEEeqnarray}
where we defined $P_{SXY}:=P_SP_XP_{Y\mid SX}$. Thus, by the union bound:
\begin{equation}
P_e^{(n)} = P(\Ec_1 \cup \Ec_2 )  \leq P(\Ec_1) +  P( \Ec_2).
\end{equation}
The first term goes to zero as $n\rightarrow \infty$ by the weak law of large numbers. The second term also tends to zero as $n\rightarrow \infty$ if $\R < I(X;Y|S)$ by the independence of the codewords and the packing lemma \cite[Lemma 3.1]{el2011network}. Therefore,  $ P_e^{(n)}$ tends to zero as $n\rightarrow \infty$ whenever $R < I(X;Y|S)$.

The expected distortion (averaged over the random codebook, state and channel noise) can   be upper bounded as
\begin{IEEEeqnarray}{rCl}
&&\hspace{-0.65cm}\Delta^{(n)}
\nonumber\\
&= &\frac{1}{n} \sum_{i=1}^n \E{d(S_i,  \hat{S}_i)} \\
& = &\frac{1}{n} \sum_{i=1}^n \E{d(S_i,  \hat{S}_i ) \big| \hat{W} \neq 1}  \Pr(\hat{W} \neq 1)  
\nonumber\\
&+& \frac{1}{n} \sum_{i=1}^n \E{d(S_i,  \hat{S}_i ) \big |\hat{W}=1} \Pr( \hat{W}=1) \\
&\leq& \D_{\max} P_e
 \nonumber\\
&&
\hspace{0.5cm}
+
 \frac{1}{n} \sum_{i=1}^n \E{d(S_i,  \hat{S}_i )\big |\hat{W}=1}   
 \cdot (1-P_e). \IEEEeqnarraynumspace \label{eq:Dn_code}
\end{IEEEeqnarray}
In the event of correct decoding, i.e., $\hat{W}=1$, 
\begin{equation}
(S^n,  X^n(1), Y^n) \in \Tc_{\epsilon}^{(n)}(P_{S}P_{X} P_{Y|SX}),
\end{equation}
and since $\hat{S}_i=\hat{s}^*(X_i,Z_i)$, also 
{\begin{equation}
(S^n,  X^n(1), \hat{S}^n) \in \Tc_{\epsilon}^{(n)}\left(P_{SX\hat{S}}\right),
\end{equation}
where $P_{SX\hat{S}}$ denotes the joint marginal pmf of  $P_{SXZ\hat{S}}(s,x,z,\hat{s}):=P_{S}(s)P_{X} (x)P_{Z|SX}(z|s,x) \mathbbm{1}\{\hat{s}=\hat{s}^*(x,z)\}$.}
Then,  
\begin{equation}\label{eq:Dn_comp}
\varlimsup_{n\to \infty} \frac{1}{n} \sum_{i=1}^n \E{d(S_i,  \hat{S}_i ) |\hat{W}=1}  \leq (1+\epsilon)  \E{d(S,  \hat{S})},
\end{equation}
{for $(S,\hat{S})$ following the marginal of the pmf $P_{SXZ\hat{S}}$ defined above.} Assuming that $\R<I(X;Y|S)$, and thus $P_e \to 0$ as $n\to \infty$, we obtain from  \eqref{eq:Dn_code} and \eqref{eq:Dn_comp}: 
\begin{equation}
\varlimsup_{n\to\infty} \Delta^{(n)}=  (1+\epsilon)   \E{d(S,  \hat{S})}.
\end{equation}
Taking finally $\epsilon \downarrow 0$, we can conclude that the error probability and distortion constraint \eqref{eq:asymptotics:Pe}, \eqref{eq:asymptotics:dist} hold (averaged over the random code constructions, the random states, and the noise in the channel) whenever 
\begin{IEEEeqnarray}{rCl}
\R &<& I(X;Y \mid S), \\
\E{d(S,  \hat{S})} & <&  \D.
\end{IEEEeqnarray}
Notice that the cost constraint \eqref{eq:cost_constraint} is fullfilled by  construction. 
By standard arguments it can then be shown that there must exist at least one sequence of  deterministic code books $\mathcal{C}_n$ so that constraints \eqref{eq:asymptotics_user1} hold.

	%
	
	\section{Blahut-Arimoto Type Algorithm  to evaluate Theorem~\ref{th:tradeoff}}\label{app:BA}
	Through simple time-sharing arguments, it can be shown that for given feasible $\B$, the set of achievable $(\R,\D)$ pairs over the single-receiver channel is convex. Its boundary is thus characterized by solving the following parameterized optimization problem for each $\mu \geq 0$:  
	\begin{IEEEeqnarray}{rCl}
	\hspace{-1cm}	L_{\mu}(\B)&:=&\max_{P_X\in \Pc_{\Bc}}  \Bigg[ \Ic(P_X, P_{Y\mid XS} \mid P_S) 
		\nonumber\\&&\hspace{3cm} - \mu \sum_{x\in \Xc} P_X(x) c(x) \Bigg].
		\label{obj3} 
	\end{IEEEeqnarray}
	Notice that the conditional mutual information functional can  explicitly be written as:
	\begin{IEEEeqnarray}{rCl}
		&&\hspace{-1cm}\Ic(P_X, P_{Y|XS} \mid  P_S) \nonumber\\
	&&= \sum_{x\in \Xc}\sum_{s\in \Sc}  \sum_{y\in \Yc} P_S(s)  P_X(x) P_{Y|XS}(y|xs)
	\nonumber\\
	&&\hspace{4.2cm} \log \frac{P_{Y|XS}(y|xs)}{P_{Y|S}(y|s)}. 
	\end{IEEEeqnarray} 
	for the state pmf $P_S$ and the SDMB transition law $P_{YZ|XS}$.  
	
For $\mu=0$, the optimization in \eqref{obj3}  yields the capacity of the SDMC under the input cost constraint (disgarding the distortion constraint), 
	while for $\mu \to \infty$, it yields the minimum possible distortion subject to the same input cost constraint. 
	We remark that the parameterized optimization problem above differs from the standard Blahut-Arimoto algorithm with cost constraints \cite[Section IV]{blahut1972computation} only in that 1) the objective function \eqref{obj3} includes an additional penalty term $- \mu \sum_{x\in \Xc} P_X(x) c(x)$
	and 2) {the mutual information functional} is $I(X;Y \mid S )$ 
	instead of $I(X;Y)$, {which  reflects} the state-dependent channel and the state knowledge at the receiver.  Since the penalty  
	term is additive and linear in $P_X$,  all concavity properties desired for a Blahut-Arimoto type algorithm remain valid. 
	The following Theorem~\ref{theorem:modifiedBA} 
	can then be proved by standard alternating optimization techniques, in analogy to the proof of the Blahut-Arimoto algorithm \cite{arimoto1972algorithm,blahut1972computation}. 
	
For any  conditional pmf $Q_{X|YS}$ on $X$ given  $(Y, S)$, define the function
		\begin{IEEEeqnarray}{rCl} 
\lefteqn{	J_\mu(P_X, P_{Y|XS},P_S,  Q_{X|YS})  : =  }  \hspace{1.3cm} 
	\nonumber\\
	&&	\sum_{x \in \Xc} \sum_{s\in \Sc} \sum_{y\in\Yc} P_X(x) P_S(s)P_{Y|XS}(y|x,s)  \hspace{1cm} \nonumber\\
	& & \hspace{2.3cm}\cdot\log \frac{Q_{X|YS}(x|y,s)}{P_X(x)} \nonumber \\
	&&- \mu \sum_{x\in \Xc} P_X(x) c(x).
	\end{IEEEeqnarray}
	\begin{theorem} \label{theorem:modifiedBA}
		Let the state pmf $P_S$ and the SDMC transition law $P_{YZ|XS}$ be given. The following statements hold: \\
		a) For any $\mu, \B \geq 0$:
		\begin{align} 
			L_{\mu}(\B) = \max_{P_X \in \Pc(\B)}  \max_{Q_{X|YS}} J_\mu(P_X,P_S, P_{YZ|XS}, Q_{X|YS} ). \label{BAa}
		\end{align} 
		b) Fix $P_X \in  \Pc(\B)$. Then, $J_\mu(P_X,P_S, P_{YZ|XS}, Q_{X|YS} )$ is maximized by choosing $Q_{X|YS}$ as 
		\begin{IEEEeqnarray}{rCl}
			Q^\star_{X|YS}(x|y,s) &=& \frac{P_X(x) P_{Y|XS}(y|xs)}{\sum_{x'} P_X(x') P_{Y|XS}(y|x's)}, \nonumber\\
			&&\qquad (x,y,s) \in\Xc\times \Yc \times \Sc, \label{BAb}
		\end{IEEEeqnarray}
		c) Fix $ Q_{X|YS}$.  Then, $J_\mu(P_X,P_S, P_{YZ|XS}, Q_{X|YS} )$
		is maximized by choosing $P_X \in \Pc(B)$ as 
		\begin{equation} \label{BAc}
			P_X^\star(x) = \frac{2^{g(x)}}
			{\sum_{x'} 2^{g(x')}}, \quad x \in \Xc,
		\end{equation}
		where 
		\begin{IEEEeqnarray}{rCl}
			g(x) &=&  \sum_s \sum_y P_S(s) P_{Y|XS}(y|xs) \log Q_{X|YS}(x|ys) \nonumber\\
			&&\hspace{3.5cm}- \lambda b(x)  -  \mu c(x)
		\end{IEEEeqnarray}
		and $\lambda \geq 0$ is chosen so that $\sum_{x\in \Xc}P_X^\star(x) b(x)= B$  	
	 when evaluated for $P_X^\star$ in \eqref{BAc}, or if no such $\lambda$ exists, then it is set to $\lambda=0$. 
		\hfill $\square$
	\end{theorem}
	\begin{IEEEproof}We  give the proofs for  the three results $a)$--$c)$
	\begin{enumerate}
	\item[$a)$] Fix pmfs $P_S,P_X, P_{Y|XS}$ and define $P_{SXY}(s,x,y):=P_S(s)P_X(x) P_{Y|XS}(y|x,s)$. Notice that 
\begin{IEEEeqnarray}{rCl} 
	\lefteqn{	J_\mu(P_X, P_{Y|XS},P_S,  Q_{X|YS})  } \nonumber \\
	& = & 
		\sum_{x \in \Xc} \sum_{s\in \Sc} \sum_{y\in\Yc} P_{SXY}(s,x,y) 
		\nonumber \\
	&&\hspace{2.3cm}
		\cdot\log \frac{Q_{X|YS}(x|ys)}{P_X(x)}
		\nonumber\\
		&& - \mu \sum_{x\in \Xc} P_X(x) c(x)\label{aa}\\
		& =& 		\sum_{x \in \Xc} \sum_{s\in \Sc} \sum_{y\in\Yc} P_{SXY}(s,x,y) 
			\nonumber \\
		&&\hspace{2.3cm}
		\cdot
		\log \frac{Q_{X|YS}(x|ys) P_{SY}(s,y)}{P_X(x)P_{SY}(s,y) } 
			\nonumber \\
		&&
		- \mu \sum_{x\in \Xc} P_X(x) c(x) \IEEEeqnarraynumspace\\
		& =& 		\sum_{x \in \Xc} \sum_{s\in \Sc} \sum_{y\in\Yc} P_{SXY}(s,x,y)  	\nonumber \\
		&& \hspace{1cm}\Bigg[ \log \frac{Q_{X|YS}(x|ys)P_{SY}(s,y)}{P_{SXY}(s,x,y)} 
		\nonumber\\
		&&\hspace{3cm}+ \log \frac{P_{SXY}(s,x,y)}{P_X(x) P_{SY}(s,y)}\Bigg]
			\nonumber \\
		&&
		-  \mu \sum_{x\in \Xc} P_X(x) c(x)
		\\
			& =&  -  \underbrace{D(P_{SXY} \| Q_{X\mid YS} P_{SY})}_{\leq 0}
			\nonumber \\
			&&\hspace{0cm}
			+ \sum_{x \in \Xc} \sum_{s\in \Sc} \sum_{y\in\Yc} P_{SXY}(s,x,y) 
			\nonumber \\
			&&\hspace{2.3cm}
		\cdot \log \frac{P_{SXY}(s,x,y)}{P_X(x) P_{SY}(s,y)}
			\nonumber\\
			&&
			-  \mu \sum_{x\in \Xc} P_X(x) c(x)\ \\
			& \leq & 	 \sum_{x \in \Xc} \sum_{s\in \Sc} \sum_{y\in\Yc} P_{SXY}(s,x,y)  
			\nonumber \\
			&&\hspace{2.3cm}
		\cdot \log \frac{P_{SXY}(s,x,y)}{P_X(x) P_{SY}(s,y)}
			\nonumber \\
		&&	- \mu \sum_{x\in \Xc} P_X(x) c(x)\ \\	
			& = & 	\sum_{x \in \Xc} \sum_{s\in \Sc} \sum_{y\in\Yc} P_{SXY}(s,x,y)  
				\nonumber \\
			&&\hspace{2.3cm}
		\cdot
			 \log \frac{P_{Y|XS}(y|x,s)}{P_{Y|S}(y|s)}
				\nonumber \\
			&&-  \mu \sum_{x\in \Xc} P_X(x) c(x)\ \\			
				&=& 	\Ic(P_X, P_{Y|XS} \mid  P_S)
			  -  \mu \sum_{x\in \Xc} P_X(x) c(x),\label{ee}\IEEEeqnarraynumspace
	\end{IEEEeqnarray}
	where $D( \cdot \| \cdot)$ denotes the Kullback-Leibler Divergence \cite{cover2006elements}. Above inequality  holds with equality when $Q_{X|YS}=P_{X|YS}$, where the latter stands for the conditional marginal pmf of $P_{SXY}$. Therefore, $\max_{Q_{X|YS}} J_\mu(P_X, P_{Y|XS},P_S,  Q_{X|YS})$ equals the right-hand side of \eqref{ee}, which directly implies \eqref{BAa}.
	\item[$b)$] For fixed $P_X$, according to \eqref{aa}--\eqref{ee}, $J_\mu(P_X, P_{Y|XS},P_S,  Q_{X|YS})$ is maximized by the choice
	\begin{IEEEeqnarray}{rCl}
	&&\hspace{-1.5cm}Q^\star_{X|YS}(x|y,s)
	\nonumber\\
	&=&\frac{P_{SXY}(s,x,y)}{P_{SY}(s,y)} 
	\nonumber\\
	&=& \frac{P_S(s) P_{X}(x) P_{Y|XS}(y|x,s)}{\sum_{x'} P_S(s)P_{X}(x') P_{Y|XS}(y|x',s) }.
	\end{IEEEeqnarray}
	\item[$c)$] The function $J_\mu(P_X, P_{Y|XS},P_S,  Q_{X|YS})$ is  concave in $P_X$ and we can thus  use the KKT conditions to find the maximum value $\max_{P_X}J_\mu(P_X, P_{Y|XS},P_S,  Q_{X|YS})$ over all pmfs $P_X$ satisfying $\sum_{x \in \Xc} P_X(x) b(x)\leq B$. In this case, the KKT conditions are summarized by the two constraints 
	\begin{IEEEeqnarray}{rCl}\label{eq:f}
&&	\sum_{s\in \Sc} \sum_{y\in\Yc}  P_S(s)P_{Y|XS}(y|x,s) \log \frac{Q_{X|YS}(x|y,s)}{P_X(x)} 
	\nonumber\\
	 &&-\sum_{s\in \Sc} \sum_{y\in\Yc}  P_S(s)P_{Y|XS}(y|x,s)  \ln(2)^{-1}
	 	\nonumber\\
	 &&
	 - \mu  c(x) - \lambda b(x) 	\nonumber\\
	 &&= \xi ,
	\end{IEEEeqnarray}
	and 
		\begin{equation}
	\sum_{x \in \Xc} P_X(x) b(x) \leq B,
	\end{equation}
	and $\lambda =0$ if above inequality is strict, and the Lagrange multiplier $\xi$ ensures that the pmf $P_X$ sums to $1$. 
Since 	$\sum_{s\in \Sc} \sum_{y\in\Yc}  P_S(s)P_{Y|XS}(y|x,s) =1$, Equation~\eqref{eq:f} is equivalent to 
		\begin{IEEEeqnarray}{rCl}
\lefteqn{\log P_X(x) } \; \nonumber\\
&=&\sum_{s\in \Sc} \sum_{y\in\Yc}  P_S(s)P_{Y|XS}(y|x,s) 
 \log Q_{X|YS}(x|y,s)  
 \nonumber\\
 &&\hspace{1.3cm}- \mu  c(x) - \lambda b(x)-\xi - \ln(2)^{-1},
	\end{IEEEeqnarray}
	and thus to
\begin{IEEEeqnarray}{rCl}
\lefteqn{P_X(x) =}
\nonumber\\
&&2^{ \underset{s\in \Sc}{\sum} \underset{y\in\Yc}{\sum}  P_S(s)P_{Y|XS}(y|x,s)  \log Q_{X|YS}(x|y,s)  - \mu  c(x) - \lambda b(x)}
\nonumber\\
&& \cdot  2^{-\xi -\ln(2)^{-1} }.
	\end{IEEEeqnarray}
	Choosing finally the Lagrange multipliers $\xi$ and  $\lambda$ so that the pmf $P_X$ sums to $1$ and the cost constraint   $\sum_{x \in \Xc} P_X(x) b(x)\leq B$ holds with equality,  we obtain the result in \eqref{BAc}. If no such $\lambda$ exists, then we set $\lambda=0$. 
	\end{enumerate}
	\end{IEEEproof}
Each of the two  maximizations in \eqref{BAa} is a convex optimization problem. The solution $L_{\mu}(B)$ can thus be obtained by an alternating maximization procedure. For our problem at hand, this alternating maximization procedure is described in Algorithm 1. 
	The algorithm  produces an optimal  convergent input distribution $P^{\infty}_{X,\mu}$, which can be used to compute
a pair of capacity-distortion values $(\C_\mu(\B) , \D_{\mu}(\B))$ on the boundary of the capacity-distortion tradeoff   for given input cost $\B$: 
	\begin{subequations}
		\begin{eqnarray}
			\C_\mu(\B) & = & \Ic\left(P^{(\infty)}_{X,\mu}, P_{Y|XS} \Big| P_S \right) \label{rate-mu} \\
			\D_\mu(\B) & = & \sum_x c(x) P^{(\infty)}_{X,\mu}(x). \label{distortion-mu}
		\end{eqnarray}
	\end{subequations}
Varying $\mu$,  the entire capacity-distortion tradeoff  is obtained for fixed input cost $\B$.
	Moreover, by varying the input cost $\B$, the whole 
	boundary of the achievable capacity-distortion-cost tradeoff region is obtained.

	\begin{algorithm}\label{alg:BA}
		\caption{Blahut-Arimoto Type Algorithm for SDMCs}
		Fix  $\mu \geq 0$. 
		\begin{algorithmic}[1]
			
			\Procedure{Tradeoff }{$\C_\mu(\B),\D_{\mu}(\B)$}
			\State Initialize 	$P^{(0)}_X(x) = \frac{1}{|\Xc|}$ for all $x\in \Xc$. 
			\For{$k = 1, 2, 3,  \ldots$} 
			
			\State \begin{equation} 
				Q^{(k)}_{X|YS}(x|y,s) = \frac{P^{(k-1)}_X(x) P_{Y|XS}(y|x,s)}{\sum_{x'} P^{(k-1)}_X(x') P_{Y|XS}(y|x',s)}.  \label{BAstep1}
			\end{equation}
			\State Choose $\lambda^{(0)} > 0$.
			\For {$\ell = 1, 2, \ldots$}  
			\State Compute  $p^{(\ell)}(x) = \frac{e^{g^{(\ell)}(x)}}{\sum_{x'} e^{g^{(\ell)}(x')}}$ with
			
			\begin{IEEEeqnarray}{rCl}\label{BAstep21}
				g^{(\ell)}(x) &=  &\sum_{s, y}  P_S(s) P_{Y|XS}(y|x,s) \log Q^{(k)}_{X|YS}(x|y,s) \nonumber \\
				&&	- \lambda^{(\ell-1)} b(x) 
				\nonumber \\
				&&\hspace{-1cm}
				- \mu \sum_{(x,s,z)\in \Xc\times \Sc \times \Zc}  P_X(x)P_S(s) P_{Z|XS}(z|x,s) 
					\nonumber \\
				&&\hspace{3.5cm}
				d(s, s^*(x,z))
			\end{IEEEeqnarray}
			
			\State  Update dual variables:
			\begin{IEEEeqnarray}{rCl} 
			\hspace{-0.5cm}	\lambda^{(\ell)} & = & \left [ \lambda^{(\ell-1)} + \alpha_\ell  \left ( \sum_x b(x) p^{(\ell)} (x) - B \right ) \right ]_+ 
			\end{IEEEeqnarray}
			\Comment {where $\alpha_\ell $ is the gradient adaptation step}
			
			\EndFor
			\State  Let $P^{(k)}_X(x) =\lim_{\ell \rightarrow \infty}  p^{(\ell)}(x)$. 
			\EndFor

			\EndProcedure
		\end{algorithmic}
	\end{algorithm}
	%
	%

	\section{Proof of Corollary~\ref{cor1:notradeoff}}\label{app:notradeoff}
	It suffices to show that under the described conditions, the distortion constraint \eqref{eq:asymptotics:dist} 
	does not depend on $P_X$. 
	To this end, we define $T=\psi(X,Z)$ and rewrite  the expected distortion as:
	\begin{IEEEeqnarray}{rCl}
		&&\mathbb{E}[d(S,  \hat{S})]
			\nonumber\\
		&&=\sum_{(x,z)\in\Xc \times  \Zc } P_{X Z}(x,z)\sum_{s \in\Sc} P_{S| XZ}(s| x,z) 
		\nonumber\\
		&&\hspace{4.5cm}
		\cdot d(s,  \hat{s}^*(x, z))  \\
		&&\stackrel{(a)}=\sum_{(x,z)\in\Xc \times  \Zc } P_{X Z}(x,z) 
			\nonumber\\
		&&\hspace{0.5cm}
	\cdot	\min_{s'\in\hat{\Sc}} \sum_{(s,t) \in\Sc\times \mathcal{T}} P_{ST| XZ}(s,t| x,z) d(s,  s')
		\\
		&&\stackrel{(b)}= \hspace{-4mm} \sum_{(x,z,t)\in\Xc \times  \Zc\times \mathcal{T}} \hspace{-4mm}
		P_{X Z}(x,z)\mathbbm{1}\{t=\psi(x,z)\} 
			\nonumber\\
		&&\hspace{2.5cm}
		 \cdot \min_{s' \in \hat \Sc}\sum_{s \in\Sc} P_{S| T}(s| t) d(s, s')
		\\
		&&=\sum_{t\in \mathcal{T}} P_{T}(t) \min_{s' \in \hat \Sc}\sum_{s \in\Sc} P_{S| T}(s| t) d(s, s')
	\end{IEEEeqnarray}
	{where ${(a)}$ holds by the definition of  $\hat{s}^*(x, z)$ and the law of total probability; and  $(b)$  by   the Markov chain  $S \markov T \markov (X,Z)$, see  \eqref{1user:cond2}, and because $T$ is a function of $X,Z$. } The independence of the pair $(T,S)$ with $X$ from \eqref{1user:cond1}, 
	together with the above expression implies that the expected distortion does not depend on the choice of the input distribution $P_X$. Hence, we can conclude that for any given $\B\geq 0$, the rate-distortion tradeoff function $\C(\D, \B)$ is constant over  all $\D\geq \D_{\min}$ and coincides with the capacity of the SDMC  $\C_{\textnormal{NoEst}}(\B)$.
\section{Proof of  Remark~\ref{rem:CSI}}\label{app:imperfectCSIR}
\subsubsection{Converse} 
Fix a sequence (in $n$) of $(2^{n\R},n)$ codes such that  Limits~\eqref{eq:asymptotics_user1} hold. By Fano's inequality there exists a sequence $\epsilon_n \to 0$ as $n\to \infty$ so that:
\begin{IEEEeqnarray}{rCl}
	n\R & \leq& I(W; Y^n,  S_R^n) + n\epsilon_n \nonumber\\
	&=&  I(W; Y^n \mid  S_R^n)+ n\epsilon_n\nonumber\\
	&=& \sum_{i=1}^n H(Y_i \mid Y^{i-1}, S_R^n)
		\nonumber\\&&\hspace{1cm}
	 - H(Y_i  \mid W,  Y^{i-1} ,S^n_R)+ n\epsilon_n \nonumber \\
	&\overset{{(a)}}\leq & \sum_{i=1}^n H(Y_i \mid   S_{R,i})
	\nonumber\\&&\hspace{1cm}
	 - H(Y_i \mid X_i,  Y^{i-1},  W,  S_R^n) + n\epsilon_n \nonumber\\
	&\overset{{(b)}} =&  \sum_{i=1}^n H(Y_i \mid  S_{R,i}) 
		\nonumber\\&&\hspace{1cm}
		- H(Y_i \mid  X_i,   S_{R,i})+ n\epsilon_n\nonumber \\
	&=& \sum_{i=1}^n I(X_i;Y_i \mid  S_{R,i})+ n\epsilon_n  
\end{IEEEeqnarray}
where $(a)$ holds because conditioning can only reduce entropy; and $(b)$ holds because $(W,   Y^{i-1},  S_{R}^{i-1}, S_{R,i+1}^n) - (S_{R,i},  X_i) - Y_i $ form a Markov chain.

Define
			\begin{equation}\label{eq:Cinf_imp}
			\C^{\textnormal{imp}}_{\textnormal{inf}}(\D, \B) := \max_{P_X \in \Pc_{\D} \cap \Pc_{\B} } I(X;Y\mid S_R) .
		\end{equation}
Then, we have
\begin{IEEEeqnarray}{rCl}
	\R & \leq& \frac{1}{n} \sum_{i=1}^n I(X_i;Y_i \mid  S_{R,i})+ \epsilon_n \nonumber \\
	&\stackrel{(c)}{\leq}& \frac{1}{n} \sum_{i=1}^n \C^{\textnormal{Imp}}_{\textnormal{inf}}
	\Bigg(
	\sum_x P_{X_i} (x) c(x),
	\nonumber\\
	&&
	\hspace{3cm}
	\sum_x P_{X_i} (x) b(x) 
	\Bigg) + \epsilon_n \nonumber\\
	&\stackrel{(d)}{\leq}& \C^{\textnormal{Imp}}_{\textnormal{inf}}
	\Bigg(
	\frac{1}{n} \sum_{i=1}^n\sum_x P_{X_i} (x) c(x), 
	\nonumber\\
	&&
	\hspace{2cm}\frac{1}{n} \sum_{i=1}^n \sum_x P_{X_i} (x) b(x) 
	\Bigg) + \epsilon_n \nonumber\\
	&\stackrel{(e)}{\leq}& \C^{\textnormal{Imp}}_{\textnormal{inf}}(\D, \B)
\end{IEEEeqnarray}
where $(c)$ holds by the definition of $\C^{\textnormal{Imp}}_{\textnormal{inf}}(\D, \B)$  in \eqref{eq:Cinf_imp}, and  $(d)$ and $(e)$ hold by similar monotonicity and concavity properties as stated in  Lemma \ref{lemma:properties}.
\subsubsection{Achievability}
Fix $P_X(\cdot)$ and a function $\hat{h}(x,z)$ that achieve $C(\D/(1+\epsilon), \B)$, where $\D$ is the desired distortion and $\B$ is the target cost,  for a small positive number $\epsilon>0$. We define the joint pmf $P_{SS_RXY}:=P_{SS_R} P_X P_{Y\mid SS_R X}$.  Codebook generation, encoding, and estimation are  as described in the proof of Theorem~\ref{th:tradeoff}; the only  difference is in the  decoding at the receiver, where the state  $S^n$ has to be replaced by $S_R$. In more details: 
\paragraph{Decoding} Upon observing outputs $Y^n=y^n$ and state sequence $S_R^n=s_R^n$, the decoder looks for an  index $\hat{w}$ such that 
\begin{equation}
	(s_R^n,   x^n(\hat{w}),y^n) \in \Tc_{\epsilon}^{(n)} (P_{S_RXY})
\end{equation}
where $P_{S_RXY}=\sum_{\mathcal{S}} P_{SS_RXY}$.
If exactly one such index exists, it declares $\hat{W}=\hat{w}$. Otherwise, it declares an error.
\paragraph{Analysis} We start by  analyzing  the probability of error and the distortion averaged over the random code construction.  Given the symmetry of the code construction, we can condition on the event $W=1$. 
We then notice that the decoder makes an error, i.e., declares nothing or $\hat{W}\neq 1$ if, and only if, one or both of the following events occur: 
\begin{IEEEeqnarray}{rCl}
	\Ec_1 &=& \big\{  (S_R^n,  X^n(1), Y^n) \notin \Tc_{\epsilon}^{(n)}  (P_{XS_RY})\big\} 
	\\
	&\textnormal{or}&\nonumber
	\\
	\Ec_2 &=& \big\{  ( S_R^n,  X^n({w}'),Y^n) \in \Tc_{\epsilon}^{(n)} (P_{XS_RY})
	\nonumber
	\\
	&&\hspace{3.3cm} \text{for some $w'\neq 1$} \big\}. \IEEEeqnarraynumspace
\end{IEEEeqnarray}
Thus, by the union bound:
\begin{equation}
	P_e^{(n)} = P(\Ec_1 \cup \Ec_2 )  \leq P(\Ec_1) +  P( \Ec_2),
\end{equation}
where we consider the average probability of error not only over the random channel noise and states but also over the random codeconstruction. 
The first term goes to zero as $n\rightarrow \infty$ by the weak law of large numbers.  By the independence of the codewords and the packing lemma \cite[Lemma 3.1]{el2011network}, the second term also tends to zero as $n\rightarrow \infty$ 
\begin{equation}\label{imp:rate}
	R < I(X;Y|S_R).
\end{equation}
%
%

\textcolor{black}{Following similar steps as in the analysis in Appendix~\ref{app:main_result_P2P}, and using the fact that by the weak law of large numbers with probability tending  to 1 as $n\to \infty$:
\begin{equation}
	(S^n,S_R^n,  X^n(1), Y^n) \in \Tc_{\epsilon}^{(n)}(P_{X}P_SP_{S_R} P_{Y|SS_RX}), 
\end{equation}
it can be shown that \begin{equation}
\varlimsup_{n\to\infty} \Delta^{(n)}=  (1+\epsilon)   \E{d(S,  \hat{s}^*(X,Z))}.
\end{equation}
Thus when $\epsilon \downarrow 0$,  the distortion constraint \eqref{eq:asymptotics:dist} holds (averaged over the random code constructions, the random states, and the noise in the channel) whenever 
\begin{IEEEeqnarray}{rCl}\label{imp:dist} 
\E{d(S,  \hat{s}^*(X,Z))} & <&  \D.
\end{IEEEeqnarray}
Notice that the cost constraint \eqref{eq:cost_constraint} is fullfilled by  construction. }

{\color{black}By standard arguments it can then be shown that there must exist at least one sequence of  deterministic code books $\mathcal{C}_n$ so that constraints \eqref{eq:asymptotics_user1} are satisfied under conditions  \eqref{imp:rate} and \eqref{imp:dist}.
}
	\section{Converse Proof of Theorem~\ref{Th:physically_BC}}\label{app:converse_proof}
	Fix a  sequence (in $n$) of $(2^{n\R_0},2^{n\R_1},  2^{n\R_2},  n)$ codes satisfying \eqref{eq:asymptotics}. 
	Fix a blocklength $n$ and start with Fano's inequality:
	\begin{IEEEeqnarray}{rCl}
		\R_0+\R_2&=&\frac{1}{n} H(W_0,W_2)\nonumber\\
		&\leq  & \frac{1}{n}\sum_{i=1}^{n}I(W_0,W_2;Y_{2i},  S_{2,i}\mid Y_2^{i-1},  S_{2}^{i-1})+\epsilon_n\nonumber\\
		&\leq& \frac{1}{n}\sum_{i=1}^{n}I(W_0,W_2,  Y_2^{i-1},  S_2^{i-1};Y_{2, i},  S_{2, i} )+\epsilon_n\nonumber\\
		&=&I(W_0,W_2,  Y_2^{T-1},  S_2^{T-1};  Y_{2, T},S_{2, T}\mid T)+\epsilon_n\nonumber\\
		&\leq &I(W_0,W_2,  Y_2^{T-1},  S_2^{T-1}, T;  Y_{2, T},S_{2, T})+\epsilon_n\nonumber \\
		&\stackrel{(a)} = & I(U;Y_2 \mid S_2) +\epsilon_n,\label{eq:R1_bound}
	\end{IEEEeqnarray}
	where $T$  is chosen  uniformly over $\{1, \cdots,  n\}$ and independent of $X^n,  Y^n_1,  Y^n_2,  W_0, W_1, W_2, S^n_1,  S^n_2$; $\epsilon_n$ is a sequence that tends to 0 as $n\to \infty$; and $U := (W_0,W_2,  Y_2^{T-1},  S_2^{T-1}, T)$,  $Y_2 := Y_{2, T}$ and $S_2 := S_{2, T}$. Here $(a)$ holds because $S_2 \sim P_{S_2}$  independent of $(U, X)$,  where we define $X:= X_T$. 
	
	Following similar steps,  we obtain:
	\begin{IEEEeqnarray}{rCl}
		R_1&=&\frac{1}{n}H(W_1 \mid W_0,W_2)
		\nonumber\\
		&\leq&\frac{1}{n}I(W_1;Y_1^n,  S^n_1 \mid W_0,W_2)+\epsilon_n
		\nonumber\\
		&\leq&\frac{1}{n}I(W_1;Y_1^n,   S^n_1,  Y_2^n,   S^n_2
		\mid W_0, W_2)+\epsilon_n
		\nonumber\\
		&=&\frac{1}{n}\sum_{i=1}^{n}I(W_1;Y_{1,i},  Y_{2,i}, S_{1,i},  S_{2,i}
			\nonumber\\
		&&\hspace{0.5cm}
		\mid Y_1^{i-1}, Y_2^{i-1}, 
		S^{i-1}_1,  S^{i-1}_2, W_0,W_2)+\epsilon_n\nonumber
		\\
		&\leq& \frac{1}{n}\sum_{i=1}^{n}I(X_i, W_1,  Y_1^{i-1},  S_{1}^{i-1}; Y_{1, i},  Y_{2, i},  S_{1, i},  S_{2, i}
		\nonumber\\
	&&\hspace{2cm}
		\mid Y_2^{i-1},  S_2^{i-1},  W_0, W_2)+\epsilon_n\nonumber
		\\
		&	\stackrel{(b)}=&\frac{1}{n}\sum_{i=1}^{n} I(X_i ; Y_{1, i},  S_{1, i} 
			\nonumber\\
		&&\hspace{2.5cm}
		\mid Y_2^{i-1},   S_2^{i-1},  W_0,W_2)+\epsilon_n
		\\
		&=&I(X_T; Y_{1T},  S_{1, T}\mid Y_2^{T-1}, S_2^{T-1},  {W_0}, W_2,  T)+\epsilon_n \nonumber \\
		& \stackrel{(c)}=& I(X; Y_1 \mid S_1, U) +\epsilon_n,  \label{eq:R2_bound}
	\end{IEEEeqnarray}
	where  we defined $Y_1 := Y_{1, T}$ and $S_{1}:= S_{1,T}$; and where $(b)$  holds by the physically degradedness of the SDMBC which implies the Markov chain $(W_0,W_2,W_1,Y_{1}^{i-1},S_1^{i-1}, Y_{2}^{i-1}, S_{2}^{i-1} )\to X_i \to (S_{1,i}, Y_{1,i}) \to (S_{2,i},Y_{2,i})$, and  $(c)$ holds because $S_1\sim P_{S_1}$  independent of $(U, X)$.
	
	Recall that we assume the optimal estimators \eqref{eq:BCestimator}  in Lemma~\ref{BC:lemma:Shat}. Using the definitions of $T$,  $X$,  $S_k$ above and  defining $Z:= Z_T$,  we can write the average expected distortions as:
	\begin{IEEEeqnarray}{rCl}\label{eq:D_boundd}
		\frac{1}{n} \sum_{i=1}^n \mathbb{E}[d_k(S_{k, i},  \hat{s}_{k}^*(X_i, Z_i)] =\mathbb{E}[ d_k ( S_{k},  \hat{s}_{k}^*(X, Z) ]. \IEEEeqnarraynumspace
	\end{IEEEeqnarray}
	
	Combining  \eqref{eq:R1_bound},  \eqref{eq:R2_bound},  and \eqref{eq:D_boundd} and letting $n\to \infty$,  we obtain that there exists a limiting pmf $P_{UX}$ such that the tuple $(U, X, S_1, S_2, Y_1, Y_2, Z) \sim P_{UX} P_{S_1S_2} P_{Y_1Y_2Z|S_1S_2X}$ satisfies the rate-constraints 
	\begin{IEEEeqnarray}{rCl}
		\R_0+\R_2 & \leq & I(U; Y_2\mid S_2)\\
		\R_1&\leq & I(X; Y_1 \mid S_1,  U) 
	\end{IEEEeqnarray}
	and the distortion constraints 
	\begin{IEEEeqnarray}{rCl} 
		\mathbb{E}[ d_k ( S_{k},  \hat{s}^*_{k}(X,  Z) ] \leq \D_k,  \quad k=1,    2,  
	\end{IEEEeqnarray}
	This completes the  proof.
	
	\section{Proof of Theorem~\ref{outer1}}\label{app:Converse_general}
	Fix a  sequence (in $n$) of $(2^{n\R_0},2^{n\R_1},  2^{n\R_2},  n)$ codes satisfying \eqref{eq:asymptotics}.  
	Fix then a blocklength $n$ and consider an enhanced SDMBC where Receiver  1 observes the pair of states $\tilde{S}_1=(S_1,S_2)$ and the pair of outputs $\tilde{Y}_1=(Y_1,Y_2)$. 
	The enhanced SDMBC is clearly physically degraded because for any input pmf $P_X$ the Markov chain
	\begin{equation}
		X \markov (\tilde{S}_1,\tilde{Y}_1) \markov (S_2,Y_2)
	\end{equation}
	holds. 
	
	Following the steps in the previous Appendix~\ref{app:converse_proof}, we can conclude that 
	\begin{IEEEeqnarray}{rCl}
		\R_0+\R_2 & \leq &  I(U_2;Y_2 \mid S_2) +\epsilon_n\\
		\R_0+\R_1+\R_2& \leq & I(X; Y_1, Y_2 \mid S_1, S_2) +\epsilon_n
	\end{IEEEeqnarray}
	and for $k=1,2$
	\begin{IEEEeqnarray}{rCl}\label{eq:D_bound}
		\frac{1}{n} \sum_{i=1}^n \mathbb{E}[d_k(S_{k, i},  \hat{s}_{k}^*(X_i, Z_i)] =\mathbb{E}[ d_k ( S_{k},  \hat{s}_{k}^*(X, Z) ]. \IEEEeqnarraynumspace
	\end{IEEEeqnarray}
	
	Consider next a reversely enhanced SDMBC where Receiver 1 observes only $(Y_1,S_1)$ but Receiver 2 observes both state sequences $\tilde{S}_2:=(S_1,S_2)$ and both outputs $\tilde{Y}_2:=(Y_1,Y_2)$.
	Following again the steps in the previous Appendix~\ref{app:converse_proof}, but now with exchanged indices $1$ and $2$, we obtain: 
	\begin{IEEEeqnarray}{rCl}
		\R_0+\R_1 & \leq &  I(U_1;Y_1 \mid S_1) +\epsilon_n\\
		\R_0+	\R_1+\R_2& \leq & I(X; Y_1, Y_2 \mid S_1, S_2) +\epsilon_n.
	\end{IEEEeqnarray}
	
	Combining all these inequalities and letting first $n\to \infty$ and then $\epsilon_n \downarrow 0$, establishes the desired converse result. 
	
	\section{Proofs for Dueck's State-Dependent BC}\label{app:B}
	
	\subsection{Optimal Estimator of Lemma~\ref{BC:lemma:Shat}}\label{appendix-sec:ex_estimator}
	\newcommand{\bx}{\overline{x}}
	We first derive the optimal estimator $\hat{s}^{*}_k (x_1,  x_2,  y_1',y_2')$ of Lemma~\ref{BC:lemma:Shat} for this example.

	\noindent{\bf Case $y_1'=y_2'=1$: } In this case,  $S_1=S_2=1$ deterministically,  and thus 
	\begin{align}\label{eq:est-11}
		\hat{s}^*_k (x_1,  x_2,  1, 1) = 1, \;\quad \forall (x_1,x_2),  \qquad  k=1, 2.
	\end{align}

	\noindent{\bf Case} $y_1=1'$ and $y_2'= 0$: In this case,  $S_1=1$ deterministically and 
	\begin{equation}\hat{s}^*_1 (x_1,  x_2,  1, 0) = 1, \;\quad \forall (x_1,x_2).
	\end{equation}
	To derive the optimal estimator for state $S_2$, we notice that $y_1'=1$ implies  $x_1\oplus N =1$, i.e., $N=x_1 \oplus 1$. As a consequence,
	\begin{equation}y_2'= (x_2 \oplus x_1 \oplus 1) S_2.
	\end{equation}
	So, for $x_2 = x_1$ we have $y_2'=S_2=0$ and  the optimal estimator sets
	\begin{equation}\label{eq:xeq}
		\hat{s}^*_2 (x_1,  x_2,  1,  0)=  0, \quad x_1=x_2.
	\end{equation}  
	Instead for $x_2 \neq x_1$,  the feedback output $y_2'=0$, irrespective of the state  $S_2$. The optimal estimator then is the constant estimator  
	\begin{equation}\label{eq:xneq}
		\hat{s}^*_2 (x_1,  x_2,  1,  0)=  \argmax_{\hat{s}\in\{0,1\}} P_S(\hat{s}), \quad x_1\neq x_2.
	\end{equation}

	{\bf Case} $y_1'=1, y_2'=0$: Symmetric to the previous case $y_1'=0, y_2'=1$. The optimal estimators are as in \eqref{eq:xeq} and \eqref{eq:xneq}, but with exchanged indices $1$ and $2$.
	\vspace{1mm}
	
	{\bf Case} $y_1'=y_2'=0$:   To find the optimal estimators, we calculate the conditional probabilities $P_{S_k|X_1X_2Y_1'Y_2'}(\cdot|x_1,x_2,y_1',y_2')$ for $y_1'=y_2'=0$. 
	
	We again distinguish the two cases $x_1=x_2$ and $x_1\neq x_2$ and start by considering $x_1=x_2$.  In this case,     $x_1\oplus N = x_2 \oplus N$,  and so if $S_k=1$ then $y_1'=y_2'=0$ only if $x_{1}\oplus N=x_2 \oplus N=0$,  which happens with probability $1/2$ because $N$ is Bernoulli-$1/2$. By the independence of the states and the inputs for  $x_1=x_2$ and $k=1,2$:
	\begin{subequations} \label{eq:x1eqx2}
		\begin{eqnarray}
			\lefteqn{P_{S_k|X_1X_2 Y_1'Y_2'}(1| x_1, x_2,  0,  0)  } \nonumber \quad \\
			& = & \frac{P_{S_k}(1) P_{Y_1'Y_2'|X_1X_2S_k}(0, 0|x_1, x_2, 1)}{P_{Y_1'Y_2'|X_1X_2}(0, 0|x_1, x_2)}\nonumber \\ 
			& = & \frac{P_S(1) 1/2}{P_{Y_1'Y_2'|X_1X_2}(0, 0|x_1, x_2)}.
		\end{eqnarray} 
		
		Let  $\bar{k}:=3-{k}$ for $k=1,2$. If $S_k=0$,  then $y_1'=y_2'=0$ happens when either $x_1\oplus N=x_2 \oplus N =0$ or when $S_{\bar{k}}=0$ and $x_1\oplus N=x_2 \oplus N =1$.  Since these are exclusive events and have total probability of $1/2 +  P_S(0) 1/2$,   we obtain for $x_1=x_2$ and $k\in\{1,2\}$:
		\begin{eqnarray}
			\lefteqn{P_{S_k|X_1X_2Y_1'Y_2'}(0| x_1, x_2, 0,  0)  } \nonumber \quad \\
			& = & \frac{P_{S_k}(0) P_{Y_1'Y_2'|X_1X_2S_k}(0, 0|x_1, x_2, 0)}{P_{Y_1'Y_2'|X_1X_2}(0, 0|x_1, x_2)}\nonumber \\ 
			& = & \frac{P_S(0) (1/2+  P_S(0)  1/2)}{P_{Y_1'Y_2'|X_1X_2}(0, 0|x_1, x_2)}.
		\end{eqnarray} 
	\end{subequations}

	We conclude from  \eqref{eq:x1eqx2} that for $y_1'=y_2'=0$ and $x=x_1=x_2$,  the optimal estimators are 
	\begin{IEEEeqnarray}{rCl}
&&	\hat{s}^*_k (x,x,  0,  0)  \nonumber\\
	&&
	\hspace{0cm}=	\mathbbm{1}\left\{P_{S}(0)(1+P_{S}(0))<P_{S}(1)\right\}, \; k=1,2.
	\end{IEEEeqnarray}

	We turn to the case $x_1\neq x_2$,  where $x_1 \oplus N =1 \oplus ( x_2  \oplus N)$. As before,  if $S_k =1$,  then $Y_k'=0$ only if $x_1\oplus N=0$,  which happens with probability $1/2$. Now this implies  $x_2 \oplus N=1$,  and thus  $Y_{\bar{k}}'=0$ only if $S_{\bar{k}}=0$,  which happens with probability $P_S(0)$. We thus obtain  for $x_1 \neq x_2$ and $k=1,2$:
	\begin{subequations}
		\label{eq:x1neqx2}
		\begin{eqnarray}
			\lefteqn{P_{S_k|X_1X_2Y_1'Y_2'}(1| x_1, x_2, 0,  0)  } \quad \nonumber \\
			& = & \frac{P_{S_k}(1) P_{Y_1'Y_2'|X_1X_2S_k}(0, 0|x_1, x_2, 1)}{P_{Y_1'Y_2'|X_1X_2}(0, 0|x_1, x_2)} \nonumber\\ 
			& = & \frac{P_S(1) P_{S}(0) 1/2}{P_{Y_1'Y_2'|X_1X_2}(0, 0|x_1, x_2)}.
		\end{eqnarray} 
		If $S_k=0$,  then $Y_1'=Y_2'=0$ happens when $x_{\bar{k}} \oplus N =0$ or when $x_{\bar{k}} \oplus N =1$ and $S_{\bar{k}}=0$. Since these are exclusive events with total probability $1/2+ P_S(0) 1/2$,  we obtain for $x_1 \neq x_2$ and $k=1,2$:
		\begin{eqnarray}
			\lefteqn{P_{S_k|X_1X_2Y_1'Y_2'}(0| x_1, x_2,  0,  0)  } \quad \nonumber \\
			& = & \frac{P_{S_k}(0) P_{Y_1'Y_2'|X_1X_2S_k}(0, 0|x_1, x_2, 0)}{P_{Y_1'Y_2'|X_1X_2}(0, 0|x_1, x_2)} \nonumber \\ 
			& = & \frac{P_S(0) (1/2+P_S(0)  1/2)}{P_{Y_1'Y_2'|X_1X_2}(0, 0|x_1, x_2)}.
		\end{eqnarray} 
	\end{subequations}
	{Since $P_S(1) < 1+ P_S(0)$,} we conclude that for $y_1'=0, y_2'= 0$ and $x_1\neq x_2$,  the optimal estimator is
	\begin{align}\label{eq:OptEstimator00}
		\hat{s}^*_k (x_1,  x_2,  0,  0)  
		=0,    \qquad  x_1\neq x_2, \; k=1,2.
	\end{align}

	\subsection{Minimum distortion} \label{app:ex_mindist}
	
	We evaluate the expected distortion of the optimal estimators in \eqref{ex:optimal_estimator}, for a given input pmf $P_{X_0X_1X_2}$. Let $t:= \Prv{X_1 \neq X_2}$. We first consider the distortion on state $S_2$:
	\begin{IEEEeqnarray}{rCl}
		\lefteqn{
			\mathbb{E}[ d( S_2, \hat{s}_2^*(X_1,X_2, Y_1',Y_2'))] } \nonumber \\
		& = & \sum_{(x_1,x_2,y_1',y_2')\in\{0,1\}^4}  \hspace{-4mm} P_{X_1X_2Y_1'Y_2'}(x_1,x_2,y_1',y_2')  
		\nonumber\\
		&&\cdot \text{Pr}\big[ S_2  \neq \hat{s}^*_2(x_1,x_2, y_1',y_2') \mid X_1=x_1,X_2=X_2,
			\nonumber\\&& 
			\hspace{5cm} Y_1'=y_1,Y_2'=y_2\Big] \nonumber \\  \\
		&\stackrel{(a)}{=}&\sum_{(x_1,x_2,y_1',y_2')\in\{0,1\}^4}P_{X_1X_2Y_1'Y_2'}(x_1,x_2,y_1',y_2') 	\nonumber\\&&
	\hspace{1.2cm}\cdot	\min_{\hat{s}\in\{0,1\}} P_{S_2|X_1X_2Y_1'Y_2'}(\hat{s}|x_1,x_2, y_1',y_2')  , \nonumber \\ \label{eq:L2}  
	\end{IEEEeqnarray}
	where $(a)$  follows by the definition of the function $\hat{s}^*_2$.
	
	In the previous Subsection~\ref{appendix-sec:ex_estimator}, we argued that for $y_2'=1$ or for $(y_2'=0,y_1'=1,x_1=x_2)$, the state $S_2$ is  deterministic   ($S_2=1$ in the former case and $S_2=0$ in the latter) and thus $\min_{\hat{s}\in\{0,1\}} P_{S_2|X_1X_2Y_1'Y_2'}(\hat{s}|x_1,x_2, y_1',y_2')=0$. We further argued that for $(y_1'=1,y_2'=0, x_1\neq x_2)$ the transmitter learns nothing about state $S_2$, which is thus still distributed according to $P_S$. Based on these observations, we  continue from \eqref{eq:L2} as:
	\begin{IEEEeqnarray}{rCl}
		\lefteqn{
			\mathbb{E}[ d( S_2, \hat{s}_2^*(X_1,X_2, Y_1',Y_2')] } \nonumber \\
		&= & \Prv{ X_1\neq X_2, Y_1'=1, Y_2'=0} \min \{ P_S(0), P_S(1)\} \nonumber \\
		& &+\sum_{(x_1,x_2)\in\{0,1\}^2}  P_{X_1X_2Y_1'Y_2'}(x_1,x_2,0,0)
		\nonumber\\
		&&\hspace{2cm}
		\min\{ P_{S_1|X_1X_2Y_1'Y_2'}(0|x_1,x_2,0,0), 
		\nonumber\\
		&&\hspace{3.5cm}
		P_{S_1|X_1X_2Y_1'Y_2'}(1|x_1,x_2,0,0)\} \nonumber \\
		&\stackrel{(b)}{=} & \Prv{ X_1\neq X_2, N=X_1\oplus 1, S_1=1}
		\nonumber\\
		&&\hspace{5cm}
		 \min \{ P_S(0), P_S(1)\} \nonumber \\
		& & +   \Prv{ X_1= X_2} \frac{1}{2} \min\{P_S(1), P_S(0) (1+  P_S(0))  \}\nonumber \\
		&& + \Prv{ X_1\neq X_2}  \frac{1}{2} P_{S}(0)   P_S(1)  \\
		& =& \frac{1}{2} t q \Big( \min \{q, (1-q)\} + (1-q) \Big) 
		\nonumber\\
		&&
		+ \frac{1}{2} (1-t) q 
			 \min \{q,(1-q)(2-q)\}. 
	\end{IEEEeqnarray}
	where in $(b)$ we used  \eqref{eq:x1eqx2}--\eqref{eq:OptEstimator00} and the fact that when $X_1\neq X_2$, then event $\{Y_1'=1,Y_2'=0\}$ is equivalent to event $\{N=X_1 \oplus 1, S_1=1\}$.
	
	\subsection{Proof of the Outer Bound}\label{app:ex_converse}		
	The outer bound is based on Theorem~\ref{outer1},  as detailed out in the following. 
	The single-rate constraints   \eqref{R1} specialize to
	\begin{IEEEeqnarray}{rCl}
		\R_k&\leq& I(U_k; Y_k',  X_0\mid S_1,  S_2)
		\\
		&\stackrel{(c)}{=}&I(U_k;  X_0)\\
		& \leq & 1, \label{eq:Rk_conv}
	\end{IEEEeqnarray}
	where the equality holds by the chain rule, because $(U_1,X_0)$ and $(S_1,S_2)$ are independent,  and because $I(U_1;Y_1'\mid X_0,S_1,S_2)=0$ due to the Bernoulli-$1/2$ noise $N$.
	
	Defining   $t := \Pr[X_1 \neq X_2]$, 	Bound \eqref{R2} specializes to:
	\begin{IEEEeqnarray}{rCl}
\lefteqn{\R_1+\R_2} \quad
		\nonumber\\
		&\leq& I(X_0, X_1, X_2; Y_1',  Y_2',X_0\mid S_1,  S_2)
		\\
		&\stackrel{(d)}{=}& H(X_0)+I(X_1, X_2;Y_2'\mid  S_1,  S_2, Y_1', X_0)\\
		&\stackrel{(e)}{=}& H(X_0)
			\nonumber\\
		&&
		+I(X_1, X_2;Y_2'\mid  S_1=1, S_2=1,Y_1',
		 X_0) \IEEEeqnarraynumspace\\
		&=& H(X_0)
			\nonumber\\
		&&
		+I(X_1, X_2;Y_2' \oplus Y_1' \mid  S_1=1,  S_2=1,  X_0)\IEEEeqnarraynumspace\\
		&\stackrel{(f)}{\leq}&H(X_0)+P_{S_1S_2}(1, 1) H(X_1 \oplus X_2) \IEEEeqnarraynumspace
		\\
		&\leq &1+q^2 H_b(t). \label{eq:sum-rate_conv}
	\end{IEEEeqnarray}
	where $(d)$ holds by the chain rule and because  $I(X_1,X_2;Y_1'\mid X_0,S_1,S_2)=0$ due to the Bernoulli-$1/2$ noise $N$;  $(e)$ holds because for $(s_1,s_2)\neq (1,1)$ the mutual information term $I(X_1, X_2;Y_2'\mid  S_1=s_1,  S_2=s_2, Y_1', X_0)=0$ due to the Bernoulli-$1/2$ noise $N$; and $(f)$ holds because for $S_1=S_2=1$ we have $Y_2'\oplus Y_{1}' =(X_2 \oplus N) \oplus (X_1 \oplus N) = X_2 \oplus  X_1$  and because conditioning can only reduce entropy.
	
	The sum-rate constraint \eqref{eq:sum-rate_conv} is maximized for $t=1/2$, which combined with  \eqref{eq:Rk_conv} establishes the converse to the capacity region in \eqref{ex:capacity}.

	\subsection{Proof of Achievability Results}\label{app:ex_ach}

	We evaluate Proposition~\ref{prp:inner} for different choices of the involved random variables. Since we ignore the common rate $\R_0$, bound \eqref{eq:inner4} is not active and can be ignored.
	
	\subsubsection{First choice}
	\begin{itemize}
		\item   $X_0,  X_1,  X_2$  Bernoulli-${1}/{2}$ with $X_0$ independent of $(X_1, X_2)$ and $X_1=X_2=x$ with probability  $\frac{1-t}{2}$ for all $x\in\{0, 1\}$; 
		\item $U_k=X_k$,   for $k=0,  1,  2$; 
		\item $V_1=(X_0, X_1)$,  $V_2=(X_0, X_2)$,     $V_0=X_1\oplus Y_1'$.
	\end{itemize}
	We plug this choice into Proposition~\ref{prp:inner}. Constraint \eqref{eq:inner1} evaluates to:
	\begin{IEEEeqnarray}{rCl}
		\R_1 &\leq& I(U_0, U_1;Y_1, V_1\mid S_1) 
		\nonumber\\
		&&- I(U_0, U_1, U_2,  Z;V_0, V_1 \mid S_1, Y_1) \\
		& = & I(X_0,X_1;X_0,Y_1', S_1,S_2, X_1 \mid S_1) 	\nonumber\\
		&&- I(X_0,X_1,X_2,Y_1',Y_2'; X_1\oplus Y_1',X_0,X_1 
			\nonumber\\
		&&\hspace{4cm}\mid S_1, S_2,X_0,Y_1') \IEEEeqnarraynumspace\\
		& \stackrel{(e)}{=} & H(X_0)+H(X_1) - H(X_1 \mid Y_1')\\
		& = & H(X_0) = 1
	\end{IEEEeqnarray}
	where $(e)$ holds because $Y_1'$ is independent of $X_1$ due to the Bernoulli-$1/2$ noise $N$. 
	
	Constraint \eqref{eq:inner2} evaluates to:
	\begin{IEEEeqnarray}{rCl}
		\R_2 &\leq& I(U_0, U_2; Y_2, V_2\mid S_2)
			\nonumber\\
		&&\hspace{1.4cm}
		 - I(U_0, U_1, U_2,  Z;V_0, V_2|S_2, Y_2)\\
		& = & I(X_0,X_2;X_0,Y_2', S_1,S_2, X_2 \mid S_1) 
		\nonumber\\
		&&
		- I(X_0,X_1,X_2,Y_1',Y_2'; X_1\oplus Y_1',X_0,X_2 
			\nonumber\\
		&&\hspace{4cm}
		\mid S_1, S_2,X_0,Y_2') \IEEEeqnarraynumspace \\
		& \stackrel{(f)}{=} & H(X_0)+H(X_2) - H(X_2) 
			\nonumber\\
		&&
		- H(X_1 \oplus Y_1' \mid S_1, S_2,X_0,Y_2' , X_2) \\
		& \stackrel{(g)}{=} & 1 - (1-q) (H_{\text{b}}(t)+q) ,
	\end{IEEEeqnarray}
	where $(f)$ holds because of the chain rule and the independence of $X_2$ and $Y_2'$;  and $(g)$ holds because for $S_1=0$ the XOR $X_1\oplus Y_1'=X_1$ and thus $H(X_1\oplus Y_1'\mid S_1, S_2,X_0,Y_2' , X_2)=H(X_1\mid X_2)$, for $S_1=S_2=1$ the XOR $X_1\oplus Y_1'= X_2 \oplus Y_2'$,  and finally for $S_1=1$ and $S_2=0$ the XOR $X_1 \oplus Y_1'=N$ independent of $(Y_2'=0,X_2)$.
	
	Constraint \eqref{eq:inner2} evaluates to:
	\begin{IEEEeqnarray}{rCl}
\lefteqn{\R_1+\R_2 } \;
		\nonumber\\
		&\leq & I(U_1; Y_1, V_1 \mid U_0,  S_1) + I(U_2; Y_2, V_2 \mid U_0,  S_2) 
		\nonumber\\
		&& + \min_{k\in\{1, 2\}}I(U_0;Y_k, V_k\mid S_k)
		- I(U_1;U_2|U_0)\nonumber \\
		&&-I(U_0, U_1, U_2,  Z;V_1 \mid V_0, S_1, Y_1) \nonumber \\
		&&-I(U_0, U_1, U_2,  Z;V_2 \mid V_0, S_2, Y_2)  \nonumber \\
		& & - \max_{k\in\{1, 2\}} I(U_0, U_1, U_2,  Z;V_0\mid S_k, Y_k)\\
		&=&\underbrace{ I(X_1;X_0,Y_1',S_1,S_2, X_1 \mid X_0, S_1)}_{=H(X_1)} 
		\nonumber\\
		&&+\underbrace{ I(X_2; X_0, Y_2' , S_1,S_2, X_2 \mid X_0, S_2) }_{=H(X_2)}  \nonumber \\
		& & + \min_{k \in \{1,2\}} \underbrace{I(X_0; X_0, Y_k', S_1, S_2, X_k \mid S_k)}_{=H(X_0)} 
		\nonumber\\&&- \underbrace{I(X_1;X_2)}_{=H(X_1)-H(X_1|X_2)} \nonumber \\
		&& - \underbrace{ I(\underline{X}, Y_1', Y_2'; X_0, X_1 \mid X_1 \oplus Y_1', X_0, \underline{S} , Y_1') }_{=0} \nonumber \\
		&& - \underbrace{I(\underline{X}, Y_1',Y_2'; X_0, X_2 \mid X_1 \oplus Y_1', X_0, \underline{S}, Y_2')}_{= H(X_2 \mid X_1\oplus Y_1', S_1, S_2, Y_2')}  \nonumber \\
		&& - \max_{k\in \{1,2\}} \underbrace{I(\underline{X}, Y_1',Y_2' ; X_1 \oplus Y_1' \mid \underline{S}, X_0, Y_k')}_{=H(X_1 \oplus Y_1' \mid S_1,S_2, Y_k')} \IEEEeqnarraynumspace \\
		& \stackrel{(h)}{=} & 2 + H_{\text{b}}(t)  - H(X_2 \mid X_1\oplus Y_1', S_1, S_2, Y_2') 
		\nonumber\\
		&&
		- H(X_1\oplus Y_1') \\
		&  \stackrel{(i)}{=}  & 1+H_{\text{b}}(t) - (1-q) H_{\text{b}}(t) - q(1-q) , 
	\end{IEEEeqnarray}
	where we used the abbreviations $\underline{X}=(X_0,X_1,X_2)$ and $\underline{S}=(S_1,S_2)$ and  $(h)$ holds because $X_1\oplus Y_1$ is independent of  $(S_1,S_2,Y_k')$, for $k=1,2$; and $(i)$ holds because for $S_1=S_2=1$ we have $X_2= Y_2' \oplus Y_1'\oplus X_1$ and thus  $H(X_2 \mid X_1\oplus Y_1', S_1=1, S_2=1, Y_2')=0$, for $S_1=0$ the XOR $X_1 \oplus Y_1'=X_1$ and thus  $H(X_2 \mid X_1\oplus Y_1', S_1=1, S_2, Y_2')=H(X_2|X_1)=H_{\text{b}}(t)$, and finally for $S_1=1$ and $S_2=0$, we have $X_1 \oplus Y_1'=N$ and $Y_2'=0$ and thus $H(X_2 \mid X_1\oplus Y_1', S_1=1, S_2=0, Y_2')=H(X_2)=1$.
	
	The presented choice of parameters can thus achieve all rate-distortion tuples $(\R_0,\R_1,\R_2,\D_1,\D_2)$ satisfying the distortion constraints in \eqref{ex:distortin} (which only depends on the probability $t:=\Prv{X_1\neq X_2}$) and 
	\begin{subequations}
		\begin{IEEEeqnarray}{rCl}
			\R_1 &\leq &1  \\
			\R_2 &\leq & 1- (1-q) (H_{\text{b}}(t)+q) \\
			\R_1 +\R_2 & \leq & 1+q H_{\text{b}}(t) - q(1-q) .
		\end{IEEEeqnarray}
	\end{subequations}
	
	\subsubsection{Second choice}
	
	Same as the first choice except that $V_0=X_2 \oplus Y_2'$. Following symmetric arguments as above, we conclude that for this choice the constraints in \eqref{eq:inner} evaluate to: 
	\begin{subequations}
		\begin{IEEEeqnarray}{rCl}
			\R_1 &\leq &1 - (1-q) (H_{\text{b}}(t)+q) \\
			\R_2 &\leq & 1 \\
			\R_1 +\R_2 & \leq & 1+q H_{\text{b}}(t) - q(1-q).
		\end{IEEEeqnarray}
	\end{subequations}
	
	\subsubsection{Combining the Choices and Time-Sharing}
	
	From the two previous subsections, we conclude that for any $t\in[0,1]$ the set of rate-distortion tuples $(\R_0,\R_1,\R_2,\D_1,\D_2)$ is achievable if it satisfies \eqref{ex:distortin} and 
	\begin{subequations}\label{ex:region}
		\begin{IEEEeqnarray}{rCl}
			\R_0+\R_1 &\leq &1  \label{eq:R1_ach}\\
			\R_0+\R_2 &\leq & 1  \label{eq:R2_ach}\\
			\R_0 +\R_1 +\R_2 & \leq & 1+q H_{\text{b}}(t) - q(1-q). \label{eq:R12_ach}
		\end{IEEEeqnarray}
	\end{subequations}
	
	As previously discussed, for $q \leq 1/2$, the distortion constraints \eqref{ex:distortin} do not depend on $t$, and thus without loss in optimality in \eqref{ex:region} one can set $t=1/2$, which results in a sum-rate constraint 
	\begin{equation}
		\R_0+\R_1+\R_2 \leq 1 + q^2.
	\end{equation}
	Combined with \eqref{eq:R1_ach} and \eqref{eq:R2_ach}, this sum-rate bound establishes the achievability of the capacity region in \eqref{ex:capacity}.

	For $q >1/2$ the distortion constraints \eqref{ex:distortin} are either increasing or decreasing in  $t$. The set of achievable rate-distortion tuples is then obtained by varying $t$ either over $[0,1/2]$ or over $[1/2,1]$. Numerical results indicate that the so obtained set is  not convex and the convex hull is obtained by considering convex combinations between different values of $t>0$ and $t=0$ for $q\in[2/3,1]$ and $t=1$ for $q\in[1/2,2/3]$. 

	\bibliographystyle{IEEEtran}
	\bibliography{main_v2,JRC}	
	%
		\begin{IEEEbiographynophoto}{Mehrasa Ahmadipour} received the B.Sc. in electrical engineering from the Iran University of Science and Technology in 2015 and received the M.Sc.\ degree in electrical engineering from the University of Tehran in 2019.
			She is currently a Ph.D. student in Telecom Paris, France. Her particular fields of interest include multi-terminal information theory, networks with states,  user cooperation, and statistical learning.
			\end{IEEEbiographynophoto}
		\begin{IEEEbiographynophoto}
			{Mari Kobayashi} (M’06–SM’15) received the B.E. degree in electrical engineering from Keio University, Yokohama, Japan, in 1999, and the M.S. degree in mobile radio and the Ph.D. degree from École Nationale Supérieure des Télécommunications, Paris, France, in 2000 and 2005, respectively. From November 2005 to March 2007, she was a postdoctoral researcher at the Centre Tecnològic de Telecomunicacions de Catalunya, Barcelona, Spain. In May 2007, she joined the Telecommunications department at Centrale Supélec, Gif-sur-Yvette, France, where she has been professor since 2016. She is currently with Apple Technology Engineering B.V.  Co. KG. 
			She is the recipient of the Newcom++ Best Paper Award in 2010, and IEEE Comsoc/IT Joint Society Paper Award in 2011, and ICC Best Paper Award in 2019. She was an Alexander von Humboldt Experienced Research Fellow (September 2017- April 2019) and an August-Wihelm Scheer Visiting Professor (August 2019-April 2020) at Technical University of Munich (TUM). 
			\end{IEEEbiographynophoto}
	\begin{IEEEbiographynophoto}{Mich\`ele Wigger}
		(S'05--M'09--SM'14) received the M.Sc.\ degree in electrical
		engineering, with distinction, and the Ph.D.\ degree in electrical
		engineering both from ETH Zurich in 2003 and 2008, respectively. In
		2009, she was first a post-doctoral fellow at the University of
		California, San Diego, USA, and then joined Telecom Paris, France,
		where she is currently a full professor. She has held
		visiting professor appointments at the Technion--Israel Institute of
		Technology and ETH Zurich.  Mich\`ele Wigger has previously served as an
		Associate Editor of the IEEE Communication Letters and as an
		Associate Editor for Shannon Theory for the IEEE Transactions on
		Information Theory. During 2016--2019 she also served on the Board
		of Governors of the IEEE Information Theory Society. Her 
		research interests are in multi-terminal information theory, in
		particular in distributed source coding and in capacities of
		networks with states, feedback, user cooperation, or caching.
	\end{IEEEbiographynophoto}

	\begin{IEEEbiographynophoto}
{Giuseppe Caire} (S '92 -- M '94 -- SM '03 -- F '05) 
was born in Torino in 1965. He received the 
B.Sc. in Electrical Engineering  from Politecnico di Torino in 1990, 
the M.Sc. in Electrical Engineering from Princeton University in 1992, and the Ph.D. from Politecnico di Torino in 1994. 
He has been a post-doctoral research fellow with the European Space Agency (ESTEC, Noordwijk, The Netherlands) in 1994-1995,
Assistant Professor in Telecommunications at the Politecnico di Torino, Associate Professor at the University of Parma, Italy, 
Professor with the Department of Mobile Communications at the Eurecom Institute,  Sophia-Antipolis, France,
a Professor of Electrical Engineering with the Viterbi School of Engineering, University of Southern California, Los Angeles,
and he is currently an Alexander von Humboldt Professor with the Faculty of Electrical Engineering and Computer Science at the
Technical University of Berlin, Germany.

He received the Jack Neubauer Best System Paper Award from the IEEE Vehicular Technology Society in 2003,  the
IEEE Communications Society and Information Theory Society Joint Paper Award in 2004 and in 2011, 
the Okawa Research Award in 2006,   
the Alexander von Humboldt Professorship in 2014, the Vodafone Innovation Prize in 2015, an ERC Advanced Grant in 2018, 
the Leonard G. Abraham Prize for best IEEE JSAC paper in 2019, the IEEE Communications Society Edwin Howard Armstrong Achievement Award in 2020, 
and he is a recipient of the 2021 Leibniz Prize  of the German National Science Foundation (DFG). 
Giuseppe Caire is a Fellow of IEEE since 2005.  He has served in the Board of Governors of the IEEE Information Theory Society from 2004 to 2007,
and as officer from 2008 to 2013. He was President of the IEEE Information Theory Society in 2011. 
His main research interests are in the field of communications theory, information theory, channel and source coding
with particular focus on wireless communications.   
	\end{IEEEbiographynophoto}
\end{document}